\PassOptionsToPackage{prologue,dvipsnames}{xcolor}

\documentclass[11pt]{article}

\usepackage[margin=1in]{geometry}
\usepackage{setspace}

\usepackage[dvipsnames,usenames]{color}
\usepackage[colorlinks=true,urlcolor=Blue,citecolor=Green,linkcolor=BrickRed]{hyperref}
\usepackage[usenames,dvipsnames]{xcolor}

\usepackage[utf8]{inputenc}
\usepackage[T1]{fontenc}

\usepackage{amsmath}
\usepackage{amsthm}
\usepackage{thmtools}
\usepackage[linesnumbered, ruled, vlined]{algorithm2e}

\usepackage[capitalise, noabbrev]{cleveref}

\usepackage{enumitem,comment}
\usepackage{url}
\usepackage{standalone}
\usepackage{mathtools}
\usepackage{graphicx,subfigure}
\usepackage{array}
\usepackage{xspace}
\usepackage{graphicx}
\usepackage{float}
\usepackage{amsfonts}
\usepackage{bm}
\usepackage{tikz}
\usepackage{todonotes}
\usepackage{xcolor,soul}

\usepackage{booktabs}
\usepackage{multirow}
\usepackage{tablefootnote}
\usepackage{threeparttable}

\usepackage{centernot}
\usepackage{wasysym}
\usepackage{pifont}

\graphicspath{{Figures/}}

\newcommand{\cross}{
  \centernot{\mathrel{-}\joinrel\mathrel{-}}
}

\newcommand{\goodcross}{\smash{\scalebox{0.6}{$\stackrel{\text{good}}{\cross}$}}}

\newcommand{\badcross}{\smash{\scalebox{0.6}{$\stackrel{\text{bad}}{\cross}$}}}

\allowdisplaybreaks

\makeatletter
\newcommand*{\bdiv}{
  \nonscript\mskip-\medmuskip\mkern5mu
  \mathbin{\operator@font div}\penalty900\mkern5mu
  \nonscript\mskip-\medmuskip
}
\makeatother

\newcommand{\renewtheorem}[1]{
  \expandafter\let\csname #1\endcsname\relax
  \expandafter\let\csname c@#1\endcsname\relax
  \expandafter\let\csname end#1\endcsname\relax
  \newtheorem{#1}
}

\definecolor{light-gray}{gray}{0.85}
\definecolor{light-yellow}{rgb}{1.0, 0.96, 0.0}
\newcommand{\hlgray}[1]{{\sethlcolor{light-gray}\hl{#1}}}

\def\eps{\varepsilon}

\def\Gf{G\setminus \{f\}}

\def\u{u}
\def\I{I}

\def\Psf{P'}
\def\Hsf{H'}
\def\Htf{\hat{H}}
\def\Q{Q}
\def\P{P}
\def\Pa{P_1}
\def\Pb{P_2}

\def\Opoly{\tilde O(\mathsf{poly}(\frac 1\eps))}
\def\Oone{\tilde O(1)}

\def\LabelstoP{\mathcal{L}^{\mathtt{s\rightarrow P}}}
\def\LabelPtoP{\mathcal{L}^{\mathtt{P\rightarrow P\setminus f}}}
\def\Labelstoa{\mathcal{L}^\mathtt{s \rightarrow P'}}
\def\LabelatoP{\mathcal{L}^{\mathtt{P'\rightarrow P\setminus f}}}
\def\LabelatfoP{\mathcal{L}^\mathtt{P'\xrightarrow{f} P}}
\def\LabelstfoP{\mathcal{L}^\mathtt{s\xrightarrow{f} P}}
\def\LabelstoP{\mathcal{L}^\mathtt{s\xrightarrow{} P \setminus\{f\}}}

\def\LabelPitoPi{\mathcal{L}^{\mathtt{P_1\rightarrow P_1}}}
\def\LabelPitoPz{\mathcal{L}^{\mathtt{P_1\rightarrow P_2\rightarrow P}}}
\def\LabelPztoPi{\mathcal{L}^{\mathtt{P_2\rightarrow P_1\rightarrow P}}}
\def\LabelPztoPz{\mathcal{L}^{\mathtt{P_2\rightarrow P_2}}}

\def\Labelonlysuffix{\mathcal{L}^{\mathtt{a\rightsquigarrow P}}}

\def\Pz{P}
\def\PP{\hat P}

\newcommand{\ceil}[1]{\left\lceil{#1}\right\rceil}

\def\len{\mathsf{len}}
\def\dist{\mathsf{dist}}

\newcommand{\first}[3]{\mathsf{first_{#1}}(#2,#3)}
\newcommand{\last}[3]{\mathsf{last_{#1}}(#2,#3)}
\newcommand{\dfirst}[4]{\mathsf{first^{#4}_{#1}}(#2,#3)}
\newcommand{\ddfirst}[5]{\mathsf{#5\text{-}first^{#4}_{#1}}(#2,#3)}
\newcommand{\dlast}[4]{\mathsf{last^{#4}_{#1}}(#2,#3)}

\newtheorem{theorem}{Theorem}[section]
\newtheorem{claim}[theorem]{Claim}
\newtheorem{invariant}[theorem]{Invariant}

\newtheorem{lemma}[theorem]{Lemma}
\newtheorem{definition}[theorem]{Definition}
\newtheorem{observation}[theorem]{Observation}

\crefname{claim}{Claim}{Claims}
\crefname{invariant}{Invariant}{Invariant}
\crefname{observation}{Observation}{Observations}

\usepackage{thmtools}
\usepackage{authblk}

 \date{}
\begin{document}
\title{\~{O}ptimal Fault-Tolerant Labeling for Reachability and Approximate Distances in Directed Planar Graphs\thanks{This research was partially funded by Israel Science Foundation grant 810/21.\\
The second author was supported by the European Research Council (ERC) under the European Union’s Horizon 2020 research and innovation programme (Grant Agreement No. 803118, UncertainENV).}}

\author[1,2]{Itai Boneh}
\author[3]{Shiri Chechik}
\author[1,2]{Shay Golan}
\author[1]{Shay Mozes}
\author[2]{Oren Weimann}

\affil[1]{Reichman University, Israel}
\affil[2]{University of Haifa, Israel}
\affil[2]{Tel-Aviv University,  Israel}

\maketitle

\begin{abstract}
We present a labeling scheme that assigns labels of size $\tilde O(1)$ to the vertices of a directed weighted planar graph $G$, such that for any fixed $\varepsilon>0$ from the labels of any three vertices $s$, $t$ and $f$ one can determine in $\tilde O(1)$ time a $(1+\varepsilon)$-approximation of the $s$-to-$t$ distance in the graph $G\setminus\{f\}$.
For approximate distance queries, prior to our work, no efficient solution existed, not even in the centralized oracle setting.
Even for the easier case of reachability, $\tilde O(1)$ queries were known only with a centralized oracle of size $\tilde O(n)$ [SODA 21].

\end{abstract}

\maketitle

\thispagestyle{empty}
\newpage
\tableofcontents
\thispagestyle{empty}

\newpage

\setcounter{page}{1}

\section{Introduction}

In network optimization, computing distances is essential for applications such as routing in transportation, communication, and logistics. Real-world networks often face temporary inaccessibility of nodes or edges due to maintenance, damage, or congestion, necessitating algorithms that can efficiently report distances in dynamic conditions.
A \emph{distance oracle} is a data structure that can report the distance between any two vertices of a given graph, with the objective of achieving an efficient tradeoff between time and space, ideally approaching constant query time and linear space.
Expanding on this, \emph{distance labeling schemes} assign compact labels to vertices, allowing queries based solely on these labels, which is especially useful in distributed systems where local information alone determines distances or reachability.
Research on labeling schemes spans various graph properties, including adjacency \cite{Kannan,alstrup2015optimal,petersen2015near,alstrup2015adjacency,AlonN17,bonichon2007short}, distances \cite{GPPR04,Bar-NatanCGMW22,Thorup04,AbrahamCG12,GavoilleKKPP01}, connectivity \cite{KatzKKP04,HsuL09,Korman10}, and Steiner trees \cite{Peleg05}. See~\cite{rotbart2016new} for a survey.
Similarly, \emph{reachability oracles} and \emph{reachability labeling schemes} aim to answer reachability queries rather than distances, also benefiting distributed applications.
The resilience of real-world networks to failures has led to research on robust data structures capable of accommodating disruptions.

In this paper, we focus on labeling schemes for approximate distances and reachability in directed edge-weighted planar graphs in the presence of a single vertex failure, also referred to as a {\em fault-tolerant} approximate distance (reachability) labeling scheme. The objective is to assign a compact label to each vertex, such that given the labels of any three vertices $s, t,f$, one can efficiently approximate the distance between $s$ and $t$ in the graph $\Gf$, or determine if $t$ is reachable from $s$ in the graph $\Gf$. Fault-tolerant labeling schemes (also called forbidden-set labeling schemes) have been extensively studied for various problems, such as connectivity, distances, and routing, and across multiple graph families (see, e.g., \cite{Courcelle09,CourcelleT07,FeigenbaumKMS07,TwiggPhD,AbrahamCG12,Bar-NatanCGMW22,AbrahamCGP16,rotbart2016new}).
Let us first review the most relevant related work. We focus on the directed case, as it is the focus of this paper and is generally more challenging than the undirected case. In many instances, techniques developed for undirected graphs do not extend to directed graphs.

\medskip
\noindent
{\bf Exact distance oracles and labeling for directed planar graphs.}
The problem of exact distance oracles for directed planar graphs has been extensively studied over the past few decades~\cite{ArikatiCCDSZ96,Djidjev96,ChenX00,FR06,Klein05,Wulff-Nilsen10,Nussbaum11,Cabello12,MozesS2012,Cohen-AddadDW17,ourSODA2018,Charalampopoulos19,CharalampopoulosGLMPWW23}. Notably, recent advances have led to very strong solutions~\cite{CharalampopoulosGLMPWW23} that achieve almost optimal $n^{1+o(1)}$ space and near-optimal $\tilde O(1)$ query time.
This result is particularly interesting because it reveals a significant gap between oracles and labeling schemes. Specifically, it was shown in~\cite{GPPR04} that exact distance labeling for planar graphs requires polynomial-sized labels of size $\Omega(\sqrt{n})$, regardless of the query time (and there is a known tight upper bound of $O(\sqrt n)$ \cite{GPPR04,GawrychowskiU23}).

\medskip
\noindent
{\bf Exact distance oracles and labeling for directed planar graphs with failures.}
Exact distance oracles for directed planar graphs in the presence of failures have bounds that are not as favorable compared to those without failures.
~\cite{Baswana} introduced a single-source fault-tolerant distance oracle with near-optimal $\tilde O(n)$ space and $\tilde O(1)$ query time. They further extended their construction to handle the all-pairs variant of the problem, resulting in increased space of $\tilde O(n^{1.5})$ and query time of $\tilde O(\sqrt{n})$.
Subsequently, \cite{faultyOracle} presented an improved fault-tolerant distance oracle for the all-pairs version in planar graphs. Their oracle accommodates multiple failures; however, it features a polynomial tradeoff between the oracle's size and query time that may be less advantageous compared to previous constructions.
Recently, an exact fault-tolerant distance labeling scheme for planar graphs with label size $\tilde{O}(n^{2/3})$ (accommodating a single failure) was presented in~\cite{Bar-NatanCGMW22}.

Importantly, in the context of the all-pairs version, the fault-tolerant oracles mentioned above have considerably worse bounds compared to the best-known distance oracles without faults for planar graphs.

\medskip
\noindent
{\bf Approximate distance labeling for directed planar graphs.}
Since exact distance labels require polynomial-sized labels~\cite{GPPR04}, researchers have pursued more compact labels that yield {\em approximate} distances. \cite{GavoilleKKPP01} studied such approximate labels across general graphs and various graph families. Specifically, for planar graphs, they introduced $O(n^{1/3} \log n)$-bit labels that provide a 3-approximation of distances. In the same year, \cite{GuptaKR01} developed even smaller 3-approximate labels requiring only $O(\log^2 n)$ bits, while Thorup presented $(1+\varepsilon)$-approximate labels of size $O(\log n / \varepsilon)$ for any fixed $\varepsilon > 0$~\cite{Thorup04}.

\medskip
\noindent
{\bf Reachability oracles and labeling for directed planar graphs.}
The reachability question was also very well studied in both general and planar graphs.
\cite{Henzinger2017} provided conditional lower bounds for combinatorial constructions of reachability oracles, showing that no non-trivial combinatorial reachability oracle constructions exist for general directed graphs.\footnote{The term combinatorial is often referred to algorithms that do not utilize fast matrix multiplications.}
Specifically, they proved that it is impossible to design a reachability oracle that simultaneously achieves $O(n^{3-\varepsilon})$ preprocessing time and $O(n^{2-\varepsilon})$ query time, for any $\varepsilon > 0$.

Since non-trivial reachability oracles are not attainable for general graphs, efforts have been directed towards developing improved reachability oracles for specific graph families. Notably, graphs possessing separators of size $s(n)$ admit a straightforward reachability oracle of size $\tilde O(n\cdot s(n))$ and query time $\tilde O(s(n))$.
Consequently, planar graphs (and more extensive graph classes such as H-minor free graphs) admit oracles of size $\tilde O(n^{1.5})$ and query time $\tilde O(\sqrt{n})$.
In a groundbreaking result, Thorup \cite{Thorup04} introduced a near-optimal reachability oracle for directed planar graphs with $\tilde O(n)$ space and $\tilde O(1)$ query time. This result can also be adapted to a labeling scheme with label size $\tilde O(1)$.
Subsequently, \cite{HolmRT15} further improved this construction to a truly optimal oracle with $O(n)$ space and $O(1)$ query time.

\medskip
\noindent
{\bf Fault-tolerant reachability oracles for directed planar graphs.}
Fault-tolerant reachability oracles have been studied extensively in general graphs, see e.g. \cite{brand2019sensitivity,GeorgiadisIP17,choudhary2016DualFaultTolerant,BaswanaCR18,king2002fully,GeorgiadisGIPU17}. In planar graphs, one can leverage the more powerful fault-tolerant {\em distance} oracles mentioned above \cite{Baswana,faultyOracle,Bar-NatanCGMW22}. However,
  in the all-pairs version, these
  oracles have considerably worse bounds when compared to the best known distance oracles without faults for planar graphs.
In a groundbreaking development, \cite{Reachability} in SODA 2021 introduced a nearly optimal fault-tolerant reachability oracle for directed planar graphs. Their innovative approach finally achieved near-optimal $\tilde O(n)$ size, $\tilde O(n)$ construction time, and $\tilde O(1)$ query time. It is not known how to turn the oracle of \cite{Reachability} into a fault-tolerant reachability  labeling scheme or into a fault-tolerant approximate distance oracle.

\paragraph{Fault-tolerant approximate distance labeling for undirected planar graphs.}
For undirected planar graphs \cite{AbrahamCG12} presented labels of size $\tilde{O}(1)$ that for any fixed $\varepsilon>0$, from the labels of vertices $s,t,$ and the labels of a set $F$ of failed vertices, can report in $\tilde O(|F|^2)$ time a $(1+\varepsilon)$-approximation of the shortest $s$-to-$t$ path in the graph $G\setminus F$.
One would hope to generalize this result to the directed case, even just settling for the seemingly easier task of reachability, and even for a single fault.
Unfortunately, it seems this result crucially relies on the graph being undirected.

\medskip
\noindent
{\bf Remaining research questions.} Previously, there was no efficient labeling scheme even just for reachability in directed planar graphs, only an oracle.
For reachability in planar graphs, the best previously known fault-tolerant labeling scheme was the one for fault-tolerant exact distances by \cite{Bar-NatanCGMW22}, in which the label size is $\tilde{O}(n^{2/3})$.
Is this the best possible?
Ideally, the goal would be to devise a labeling scheme in which the sum of the label sizes is roughly equal to the size of the state-of-the-art oracle. However, achieving this goal is not always possible.
For instance, as mentioned above, in \cite{CharalampopoulosGLMPWW23}, an almost optimal exact distance oracle is given for directed planar graphs (without faults) of size $O(n^{1+o(1)})$ and query time $\tilde{O}(1)$. On the other hand, it was shown in \cite{GPPR04} that exact distance labels (even without faults) for planar graphs necessitate polynomial-sized labels regardless of the query time.
Given this discrepancy between oracles and labeling schemes for distances in planar graphs, a natural question arises: does the same discrepancy exist for fault-tolerant reachability or for approximate distances?
In other words, is it possible to design a labeling scheme for directed planar graphs with label size $\tilde{O}(1)$ and query time $\tilde{O}(1)$? Or does a similar gap exist between fault-tolerant reachability oracles and labeling schemes, as in the case of exact distances?

Furthermore, in the case of approximate distances, there is not even an oracle with near-optimal bounds capable of handling a single failure in planar directed graphs.
If one wants an approximate distance oracle for directed planar graphs, the best option up to our work is to use an exact fault-tolerant oracle, which is far from the optimal bounds we aim for both in terms of space and query time.

A natural question is whether it is possible to devise an efficient approximate fault-tolerant distance oracle for directed planar graphs with near-optimal bounds of $\tilde{O}(n)$ size and $\tilde{O}(1)$ query time?
If the answer to this question is positive, a further question would be whether it is also possible to obtain an approximate fault-tolerant distance labeling scheme with near-optimal $\tilde{O}(1)$ label size. A positive answer to this would also resolve the open question for the simpler case of reachability.

\medskip
\noindent
{\bf Our results.} We answer the above two questions in the affirmative by providing a near optimal fault tolerant approximate distance labeling scheme and reachability in directed planar graphs with $\tilde{O}(1)$ label size and $\tilde{O}(1)$ query time, see \cref{thm:dist}.

\section{Technical Overview}
In this section, we first discuss the main challenges in extending the non-faulty reachability labels of Thorup~\cite{Thorup04} to handle faults.
Then, we introduce a high level overview of these labels.
Finally, we discuss the challenges in extending our single-fault reachability labels to approximate distance labels, and how we overcome them.

\subsection{Challenges in extending \cite{Thorup04}}
Thorup's non-faulty reachability labeling~\cite{Thorup04}, stores for each vertex $s$ and each relevant path separator $P$, the first vertex on $P$ that is reachable from $s$ in $G$ and the last vertex of $P$ that can reach $s$ in $G$, denoted as $\first{G}{s}{P}$ and $\last{G}{s}{P}$, respectively.
To determine if vertex $s$ can reach vertex $t$ by a path that intersects the path separator $P$, one can simply check if $\first{G}{s}{P}$ precedes $\last{G}{t}{P}$ on $P$ (denoted as $\first{G}{s}{P}\le_P \last{G}{t}{P}$).
We call the general idea of reducing $s$-to-$t$ reachability to finding the first/last reachable vertices on some path $P$ separating $s$ and $t$ as the \textit{`Find the First'} approach.

One of the main challenges we face with applying the `Find the First' approach is the occurrence of failures anywhere along the relevant path separator $P$.
A faulty vertex $f$ on $P$ requires us to store additional information, including the closest vertex to $f$ that appears after $f$ on $P$ and is reachable from the starting vertex $s$ via a path internally disjoint from $P$.
Considering that failures can happen at any vertex $P$, this means we would need to store all vertices that are reachable from $s$ on $P$ via a path internally disjoint from $P$.
This requirement renders the approach impractical.\footnote{The same difficulty arises when trying to adapt the fault-tolerant labeling scheme for undirected planar graphs of \cite{AbrahamCG12}.}

To address this issue, we need to adopt a different approach and develop new techniques specifically tailored for accommodating failures in the directed case.

It is worth mentioning that in both our reachability labeling scheme and the reachability oracle of \cite{Reachability}, the most challenging scenario arises already when all three vertices $s$, $t$, and $f$ are situated on the same path separator $P$.
In \cite{Reachability}, this situation was addressed by employing a complex data structure that extends dominator trees and previous data structures designed for handling strong-connectivity in general (non-planar) graphs under failures \cite{GeorgiadisIP17}.
However, these data structures do not seem to be distributable into a labeling scheme (for example, they rely on an orthogonal range data structure and on a binary search step which do not seem suitable for a labeling scheme).
We tackle this case without relying on dominator trees or similar sophisticated techniques.
This conceptually simpler solution is amenable to extension into an approximate distance labeling.
We believe that it should be possible to extend our labeling scheme to multiple failures and to other graph families.

\subsection{High level overview of our reachability labels}
In this section we provide a high-level overview of the techniques and ideas used in order to obtain the reachability.
Specifically, we show how to break down the reachability task to a series of 'Find the First' style sub-tasks.
Then, the same conceptual partition into sub-tasks can be applied to approximate distances labeling, with some modification to each sub-task that takes path lengths into account.
We apply a fully recursive decomposition of the graph $G$ using shortest path separators~\cite{Thorup04}, which induces a hierarchical decomposition of the graph (with $O(\log n)$ levels), where every subgraph is partitioned in the next level into two subgraphs, separated by $O(1)$ shortest paths.
For simplicity, we will assume here that each separator is composed of a single shortest path.

Consider first the task of reachability labeling \emph{without} faults~\cite{Thorup04}.
In this case, there exists a separator $P$ that separates $s$ and $t$.
Let $a=\first{G}{s}{P}$ be the first vertex on $P$, reachable from $s$ in $G$ and let $b=\last{G}{t}{P}$ be the last vertex on $P$ that can reach $t$ in $G$.
It is straightforward that $s$ can reach $t$ in $G$ if and only if $a$ appears on $P$ earlier than $b$ (which we denote by $a\le_P b$). See \cref{fig:middlepoint}.
Thus, a labeling scheme for this simple problem is that every vertex $v$ stores $\first{G}{v}{P}$ and $\last{G}{v}{P}$ for every separator $P$ above it in the recursive decomposition.

\begin{figure}[htb]
  \begin{center}
 \includegraphics[scale=0.5]{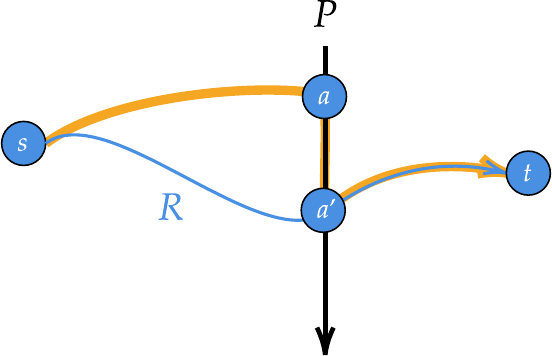}
  \caption{If $R$ is a (blue) path from $s$ to $t$ that crosses $P$ (at $a'$), then there exists an (orange) path from $s$ to $t$ that goes through $a=\first{}{s}{P}$, then from $a$ to $a'$ along $P$ and finally from $a'$ to $t$ along $R$).\label{fig:middlepoint}}
   \end{center}
\end{figure}
\paragraph{The `Find the First' framework.}
To handle faults, we repeatedly employ the high-level approach of the non-faulty labels.
If a path $R$ in some graph $H$ (in particular $H=\Gf$) from $s$ to $t$ visits a separator $P$, then there is a path from $s$ to $t$ that visits $a= \first{H}{s}{P}$.
Therefore, if we know that a path from $s$ to $t$ visits $P$, we can reduce $s$-to-$t$ reachability to $a$-to-$t$ reachability, and to finding $a$.

We consider two cases regarding the faulty vertex $f$: it is either on $P$ or it is not on $P$.

\begin{figure}[htb]
\begin{center}
\includegraphics[width=0.35\textwidth]{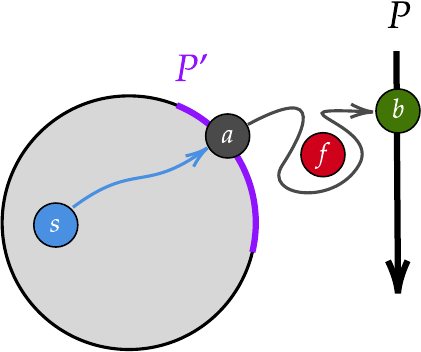}
\caption{An lustration of the case $f\notin P$, a path from $s$ to $b=\first{\Gf}{s}{P}$.
The blue subpath is from $s$ to $a\in P'$, and the gray subpath is from $a$ to $b \in P$.
\label{fig:vtfoPB}
}
\end{center}
\end{figure}

\paragraph{If \boldmath$f\notin P$ (see \cref{fig:vtfoPB}),} then the task is similar to the non-faulty case, except that now we are interested in $\first{\Gf}{s}{P}$ instead of $\first{G}{s}{P}$ which depends on $f$.
We therefore need to introduce a labeling scheme that, given the labels of $s$ and $f$ can compute $b=\first{\Gf}{s}{P}$.

The `Find the First' logic allows us to break the problem of finding $b$ itself into other sub-problems of the form {\em find the first} reachable vertex on other separators;
In addition to the separator $P$ that separates $s$ from $t$, we will use the separator $P'$ that separates $s$ from $f$.
Let $R$ be a path from $s$ to $b$ in $\Gf$.
It is easy to handle the case where $R$ does not cross $P'$.
Otherwise, $R$ crosses $P'$.
Notice that we can assume without loss generality that $R$ visits $a=\first{G_{P'}}{s}{P'}$, where $G_{P'}$ is the side of the separator $P'$ that contains $s$ and not $f$.
It is clear that $\first{\Gf}{s}{P}=\first{\Gf}{a}{P}$.
The label of $s$ stores $a$, which allows us to reduce the problem to a specialized labeling scheme with the following settings:
We are given two paths $P'$ and $P$ and two vertices $a\in P'$ and $f\notin P'\cup P$ and we wish to find $\first{\Gf}{a}{P}$. We denote this labeling scheme by $\LabelatfoP$. Designing $\LabelatfoP$ turns out to be the main technical challenge in the case $f\notin P$, which we solve as follows.

For a vertex $v\in P'$ let $u$ be the last vertex in $P'$ such that $\first{G}{v}{P}=\first{G}{u}{P}$.
Let $C_u$ be a path in $G$ from $u$ to $\first{G}{u}{P}$, we call such a path {\em canonical}.
Moreover, we set $C_v=C_u$.
Notice that when considering a path from $a\in P'$ to $\first{G}{a}{P}$, we can always consider a path that goes from $a$ to the beginning of $C_a$ on $P'$ and then continues on $C_a$.
A helpful property of the paths $C_v$ for $v\in P'$ is that they are disjoint, i.e. if $C_x\ne C_y$ then $C_x\cap C_y=\emptyset$.
Therefore, each vertex $f$ can store the (at most one) $u$ such that $f \in C_u$ it lies on.
If $f\notin C_a$, then $\first{\Gf}{a}{P}=\first{G}{a}{P}$.
Otherwise $f\in C_a$, let $R$ be a shortest path from $a$ to $\first{\Gf}{a}{P}$ in $\Gf$, let $C_1$ and $C_2$ be the prefix and suffix of $C_a$ up to $f$.
We distinguish between three cases:
If $R\cap C_a=\emptyset$, the label of $a$ stores $\first{G\setminus C_a}{a}{P}$.
If $R$ intersects $C_1$, then $\first{\Gf}{a}{P}=\first{\Gf}{u}{P}$ (where $u\in P'$ is the first vertex of $C_a$), and $f$ stores $\first{\Gf}{u}{P}$.
Otherwise, $R$ intersects $C_2$ and $\first{\Gf}{a}{P}=\first{G}{a}{P}$. We recognize this case by storing in the label of $a$ the last $f$ on $C_a$ such that $\first{\Gf}{a}{P}=\first{G}{a}{P}$.

\begin{figure}[h]
\begin{center}
\includegraphics[width=0.35\textwidth]{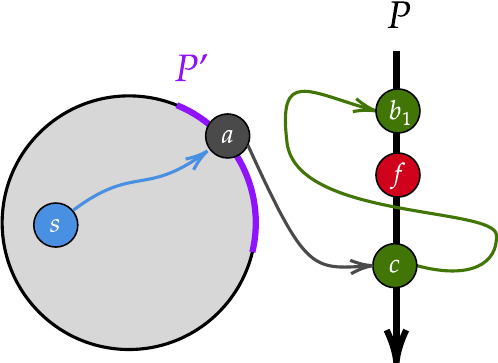}
\caption{An lustration of the case $f\in P$.
\label{fig:vtoPfB}
}
\end{center}
\end{figure}

\paragraph{If \boldmath$f\in P$ (see \cref{fig:vtoPfB}),} it is no longer true that $s$ can reach $t$ in $\Gf$ if and only if $\first{\Gf}{s}{P}\le_P \last{\Gf}{t}{P}$.
Instead, let $P_1$ and $P_2$ be the prefix and suffix of $P$ before and after $f$ (excluding $f$) respectively. Now it holds that $s$ can reach $t$ in $\Gf$ if either $\first{\Gf}{s}{P_1}\le_P \last{\Gf}{t}{P_1}$ or $\first{\Gf}{s}{P_2}\le_P \last{\Gf}{t}{P_2}$.
Therefore, our goal is to retrieve $\first{\Gf}{s}{P_1}$ and $\first{\Gf}{s}{P_2}$ from the labels of $s$ and $f$.
We focus here on labels for finding $b_1=\first{\Gf}{s}{P_1}$, the labels for $\first{\Gf}{s}{P_2}$ are similar.

Let $R$ be a path in $\Gf$ from $s$ to $\first{\Gf}{s}{P_1}$.
For the sake of simplicity we assume that $R$ visits some vertex $a'$ of $P'$, and that the first vertex on $P$ that is on $R$ after $a'$ is some vertex $c'$ in $P_2$.
We denote as $G_P$ the side of the separator $P$ that contains $s$.
We break the task of computing $\first{\Gf}{s}{P_1}$ into three sub-tasks.
\begin{enumerate}
    \item Finding $a=\first{G_{P'}}{s}{P'}$ (where again, $P'$ separates $s$ and $f$, and $G_{P'}$ is the side of $P'$ that contains $s$ and not $f$).
    \item Finding $c = \first{G_P\setminus\{f\}}{a}{P_2}$.\footnote{This is an inaccurate simplification for the sake of reducing clutter.
    In reality, the sub-problem is to find a first vertex using paths that are \emph{internally disjoint} from $P$.}
    \item Finding $\first{\Gf}{c}{P_1}$.
\end{enumerate}

We show that $\first{\Gf}{s}{P_1} = \first{\Gf}{c}{P_1}$, which immediately implies the usefulness of these three tasks.
Let $a'$ be the first vertex on $R$ that is on $P'$, and let  $c'$ be the first vertex on $R$ after $a'$ that is on $P$ (specifically on $P_2$, due to our assumption).
Notice that $R[s,a']$ is a path in $G_{P'}$ and $R[a',c']$ is a path in $G_P\setminus \{f\}$.
Since $a'$ is reachable from $s$ in $G_{P'}$, the vertex $a$ precedes $a'$ on $P'$.
We can therefore assume without loss of generality that $R$ visits $a$, since $s$ can reach $a$ and $a$ can reach $a'$ via $P'$, which implies that $a$ can reach $\first{\Gf}{s}{P_1}$.
By the same reasoning, we get that $\first{\Gf}{c}{P_1} = \first{\Gf}{a}{P_1} = \first{\Gf}{s}{P_1}$.

The most challenging out of the above three sub-problems is the third one.
In order to tackle it, we further partition it into three sub-problems.

\begin{enumerate}
    \item Given two vertices $b\in P_2$ and $f$, find $\first{G\setminus P_1}{b}{P_2}$
    \item Given two vertices $b \in P_2$ and $f$, find the first vertex reachable on $P_1$ from $b$ via a path that is internally disjoint from $P_1$.
    \item Given two vertices $b\in P_1$ and $f$, find $\first{\Gf}{b}{P_1}$.
\end{enumerate}

Again, it follows from `Find the First' arguments that the above three problems, when put together, allow us to find $\first{\Gf}{c}{P_1}$ for a vertex $c \in P_2$.

\subsection{Extending the reachability labels to approximate-distance labels}
We proceed to describe the labels for $(1+\eps)$-distance approximation.
Our $(1+\eps)$-approximate  distance  labeling builds upon the same high-level approach of the reachability labeling.
In order to adapt our reachability labeling for $(1+\eps)$-distance approximation, we first need to modify some of the notations and ideas used in the reachability settings.

 \begin{figure}[htb]
  \begin{center}
 \includegraphics[scale=0.7]{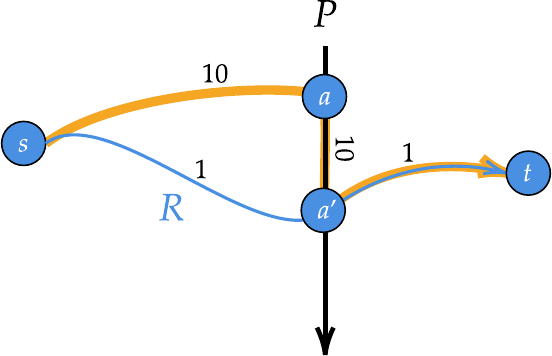}
  \caption{A naive attempt to use the reachability strategy for approximate distances may find a path which is much longer than the shortest path, hence fails to approximate the distance.\label{fig:nonsense}}
   \end{center}
\end{figure}

In order to motivate our modified definitions, let us consider a naive application of the `Find the First' approach for approximate distances: i.e., given $s$, $t$, $f$, and a separator $P$ that separates $s$ from $t$, find $a = \first{\Gf}{s}{P}$ and reduce the task of approximating $\dist_{\Gf}(s,t)$ to the task of approximating $\dist_{\Gf}(s,a)$ and $\dist_{\Gf}(a,t)$ (see \cref{fig:nonsense}).
This approach fails due to two reasons:
First, while $s$ can reach $t$ via $a$, it may be the case that the subpath from $s$ to $a$ is very expensive, while reaching a later vertex on $P$ is significantly cheaper.
Second, even if reaching $a$ is cheap, the subpath $P[a,a']$ might be significantly more expensive than the distance from $s$ to $t$ ($a'$ is the actual vertex on $P$ used by a shortest $s$-to-$t$ path).
Therefore, we have no guarantee on the distance from $a$ to $t$.
Compactly, the two problems can be put down as follows:
\begin{equation*}
   \text{(1)} \quad     \dist_{\Gf}(s,a) \not \le (1+\eps)\dist_{\Gf}(s,a')
 \quad  \qquad \text{(2)} \quad  \dist_{\Gf}(a,t) \not \le (1+\eps)\dist_{\Gf}(a',t)
\end{equation*}

We adjust our approach to treat these issues as follows.
Using the framework of Thorup~\cite{Thorup04} for (non-faulty) approximate distance labeling, we can assume that the separators of the recursive decomposition have the following structure:
For some number $r$ (think of $r$ as an estimation of the distance from $s$ to $t$ in $\Gf$), the length of each separator in $G$ is $O(r)$.
In particular, we can partition each separator path into $O(\frac{1}{\eps})$ subpaths, each of length $O(\eps r)$.
Let $R$ be a shortest $s$ to $t$ path in $\Gf$.
Now, when applying `Find the First'  strategy, we would like to find the first vertex on the specific subpath $P$ that is visited, rather than on the entire separator.
This resolves the second problem above.
Since the cost of $P$ is now bounded by $O(\eps r)$, we have that $\dist(\first{\Gf}{s}{P},t) \le \dist(a',t) + O(\eps r)$.
Since $\dist_{\Gf}(s,t) \approx r$, an additive factor of $O(\eps r)$ to the answer is acceptable.

To handle the first problem, we refine the definition of $\first{G}{s}{P}$ to take the length of the path into account.
Specifically, we define $\dfirst{G}{s}{P}{\alpha}$ to be the first vertex reachable on $P$ from $s$ \textit{with a path of length at most $\alpha$}.
By guessing $\alpha$ such that $\dist_{\Gf}(s,a') \le \alpha \le \dist_{\Gf}(s,a') + O(\eps r)$, we have that $\hat a=\dfirst{\Gf}{s}{P}{\alpha}$ is still earlier than $a'$ on $P$, and $\dist_{\Gf}(s,\hat a) \le \dist_{\Gf}(s,a') + O(\eps r)$.

It follows from the above discussion that $\dfirst{\Gf}{s}{P}{\alpha}$ with an appropriate value of $\alpha$ could be used to approximate the $s$ to $t$ distance similar to the way $\first{\Gf}{s}{P}$ is used in reachability.
Notice that each application of the `Find the First' approach in the approximate distance settings may introduce an additional additive factor of $O(\eps r)$ to the final answer.
When finding an approximate shortest path, we only apply this approach a constant number of times, so the additive $O(\eps r)$ factors do not aggregate too much and sum up to $O(\eps r)$ over all `Find the First' applications.

Sadly, we do not know how to design a labeling scheme for finding $\dfirst{\Gf}{s}{P}{\alpha}$.
We therefore introduce a further relaxation by defining the $\ddfirst{\Gf}{s}{P}{\alpha}{\delta}$ property.
We say that a vertex $a^*$ is a $\ddfirst{\Gf}{s}{P}{\alpha}{\delta}$ if $a^* \le_P \dfirst{\Gf}{s}{P}{\alpha}\le_P a'$ and $\dist(s,a^*) \le \alpha + \delta$.
Notice that for $\delta = O(\eps r)$, the vertex $a^*$ functions as a 'legitimate' middle point between $s$ and $t$.
That is, we have $\dist_{\Gf}(s,a^*) \le \alpha + O(\eps r) \le \dist_{\Gf}(s,a') +O(\eps r)$ and $\dist_{\Gf}(a^*, t) \le \len(P[a^*,a']) + \dist_{\Gf}(a',t) \le \dist_{\Gf}(a',t) + O(\eps r)$.
We can therefore define our sub-problems as the tasks of finding $\ddfirst{\Gf}{s}{P}{\alpha}{\delta}$.

\paragraph{Adjusting reachability sub-routines to approximate distance sub-routines.}
The high level approach of `Find the First' allows us to break the approximate-distance task into sub-problems in the same way as in reachability.
However, adapting the labeling schemes for these problems from reachability to approximate distances poses significant technical difficulties, and requires new involved machinery.
To demonstrate these difficulties, we focus on one of the sub-problems defined for the reachability labeling.

Consider the following sub-problem, which we call the {\em easy problem}.
We are given a path $P$ in $G$ such that all the edges touching $P$ emanate or enter $P$ from or to the left, and we are interested in assigning labels to the vertices of $P$ such that given the labels of two vertices $f <_P b$, one can find $\first{\Gf}{b}{P_1}$ such that $P_1$ is the prefix of $P$ preceding $f$.
Due to the context in which this sub-problem is used, we can also assume that $P_1$ has no outgoing edges to $G\setminus P_{1}$, and $P_2$ (the suffix of $P$ following $f$) does not have incoming edges from $G\setminus P_2$.
Essentially, this means that we are only interested in paths from $P_2$ to $P_1$ that are internally disjoint from $P$, apart from a prefix that is a subpath of $P$ (see \cref{fig:G1pathB}).

\begin{figure}[htb]
\begin{center}
\includegraphics[width=0.5\textwidth]{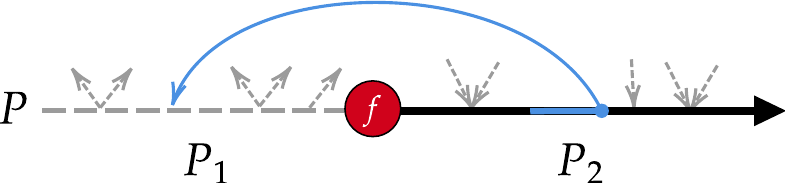}
\caption{An illustration of $G_1$.
We are interested in $P_2$ to $P_1$ paths  (like the blue path) that may start with some edges of $P_2$ and then continue with a subpath which is internally disjoint from $P$.
Formally, this is achieved by removing all in-going edges to $P_2$ and all out-going edges from $P_1$ (the removed edges are displayed in gray in the figure).
\label{fig:G1pathB}
}
\end{center}
\end{figure}

Indeed, it is easy to solve the easy problem;
Given some faulty vertex $f$, we define for every $v\in P_2$ the vertex $v_f=\first{\Gf}{v}{P_1}$.
Notice that $v_f$ may be undefined: it is possible that $v$ cannot reach any vertex on $P_1$.
Let $D_v$ be a path from $v$ to $v_f$ in $\Gf$.
We show that for every two vertices $x,y$ on $P_2$ that can reach $P_1$, it holds that $x_f = y_f$.
This follows from the following two arguments:
First, if $x <_P y$ we must also have $x_f \le_P y_f$ since $x$ can reach $y$.
Then, if $x_f$ is strictly before $y_f$, the paths $D_x$ and $D_y$ intersect (see \cref{fig:nested}).
In particular, an intersection means that $y$ can reach $x_f <_P y_f$, a contradiction.
We use this observation to define the label of $f$.
Specifically, the label of $f$ will store the last vertex $\ell$ on $P_2$ that can reach $P_1$, and $\ell_f$.
It follows from our observation that given the index of $b$ on $P_2$, and the label of $f$, we can set $\first{\Gf}{b}{P_2}$ to be $\ell_f$ if $b \le_P \ell$ or null otherwise.

In the approximate-distance version of the easy problem, the settings are exactly the same but we are also given a number $\alpha$ and an approximation parameter $\eps$.
We are interested in assigning labels to the vertices of $P$ such that given the labels of two vertices $f <_P b$, we can output a vertex that is an $\ddfirst{\Gf}{b}{P_1}{\alpha}{\eps \alpha}$.
In this case, the length of $P$ is zero.

Let us try to apply a similar logic to the approximate-distances variant.
For a faulty vertex $f$, and for every $v\in P_2$ we now define $v_f = \dfirst{\Gf}{v}{P_1}{\alpha}$.
Because $\len(P) = 0$, it is still true that for every two vertices on $P_2$ such that $x<_P y$, we have $x_f \le_P y_f$.
However, it is no longer true that $x_f = y_f$.
Now, if $x_f<_P y_f$, the intersection of $D_x$ and $D_y$ implies a path of length $\mathbf{2}\alpha$ from $y$ to $x_f$.
In particular, we have that  the first vertex $p_f\in P_1$ that can be reached from some vertex $p\in P_2$ within budget $\alpha$, is an $\ddfirst{\Gf}{v}{P_1}{\alpha}{\alpha}$ for every $v\in P_2$ that can reach $P_1$ within budget $\alpha$.

Therefore, for $\eps = 1$ the approximation variant of the problem can be solved using an almost identical labeling scheme to the reachability variant.
The label of vertex $f$ will store the last vertex $\ell$ on $P_2$ that can reach a vertex on $P_1$ via a path of length at most $\alpha$, and also stores $p_f$.

\paragraph{The `Good Cross -- Bad Cross' Framework.}
At first glance, it may seem as if the above logic completely fails when $\eps < 1$.
The length of the path implied by the intersection $D_x$ and $D_y$ may actually be $2\alpha$, and if so $x_f$ may be not $\ddfirst{\Gf}{y}{P_1}{\alpha}{\eps \alpha}$.
We overcome this issue by introducing the framework of \textit{good cross and bad cross} which we use to solve multiple problems, in each of these problems this idea is used in a different way.
This is one of the main technical contributions of the paper, and we strongly believe it would find future applications.

Let $x$ and $y$ be the first and last vertices on $P_2$ that can reach $P_1$ with budget $\alpha$.
Consider the `lucky' situation in which $y$ can reach $x_f$ with budget $(1+\eps) \alpha$.
In this case, the construction we describe above would work.
Namely, for every $v\in P[x,y]$, we have  $x_f$ is $\ddfirst{\Gf}{v}{P_1}{\alpha}{\eps \alpha}$.
We call this kind of situation a {\em good cross}.
A {\em bad cross} is the complementary case in which $\dist_{\Gf}(y,x_f) > (1+\eps) \alpha$.
In this case, in particular the path $D_{x,y}$ from $y$ to $x_f$ that uses a prefix of $D_y$ and a suffix of $D_x$ is too costly.
We exploit the fact that the remainders of $D_x$ and $D_y$ that do not participate in $D_{x,y}$ form a cheap path in $\Gf$ from $x$ to $y_f$.

\begin{figure}[htb]
\begin{center}
\includegraphics[width=0.6\textwidth]{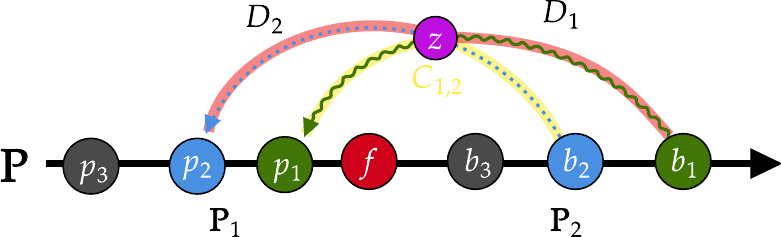}
\end{center}
\end{figure}

Let us show how the concept is used to obtain labels of size $O(\frac{1}{\eps})$.
Let $b_1$ be the last vertex in $P_2$ that can reach $P_1$ with budget $\alpha$, and let $p_1 = \dfirst{\Gf}{b_1}{P_1}{\alpha}$.
Now, let $b_2$ be the last vertex on $P_2(f,b_1]$ that has a `bad cross' with $b_1$, i.e. the last vertex earlier than $b_1$ such that $\dist_{\Gf}(b_1,p_2) > (1+\eps) \alpha$ for $p_2 = \dfirst{\Gf}{b_2}{P_1}{\alpha}$.
Denote the vertex following $b_2$ on $P$ as $b'_1$ and $\dfirst{\Gf}{b'_1}{P_1}{\alpha}= v_1$.
By the definition of $b_2$, the vertices $b_1$ and $b'_1$ good cross each other, which means that $\dist_{\Gf}(b_1,v_1) \le (1+\eps) \alpha$.
Due to $\len(P) = 0$ we also have that $v_1$ is an $\ddfirst{\Gf}{b}{P_1}{\alpha}{\eps \alpha}$ for every $b \in P_2(b_2,b_1]$.

We now discuss the implication of the bad cross between $b_1$ and $b_2$ .
Consider shortest paths $D_{1}$ and $D_2$ from $b_1$ to $p_1$ and from $b_2$ to $p_2$, respectively.
Due to planarity, the paths must intersect.
Let $z\in D_1 \cap D_2$, and denote the path $D_{1,2} = D_1[b_1,z]\cdot D_2[z,p_2]$.
Due to the bad cross, $\len(D_{1,2}) > (1+\eps)\alpha$.
Together, $D_1$ and $D_2$ cost at most $2\alpha$.
Hence, the path $C_{1,2} = D_2[b_2,z]\cdot D_1[z,p_1]$ must have length of at most $2\alpha-(1+\eps)\alpha=(1-\eps)\alpha$.

As it is, $C_{1,2}$ does not seem like a `useful' path: it allows $b_2$ to reach $p_1$ which can only be worse (i.e. later on $P_1$) than $p_2$.
It is helpful to think of the length of $C_{1,2}$ as a `budget' for bad crosses.
We will describe an iterative procedure in which every time a bad cross occurs, we get a path that resembles $C_{1,2}$ whose length is smaller by $\eps \alpha$.

\begin{figure}[htb]
\begin{center}
\includegraphics[width=0.6\textwidth]{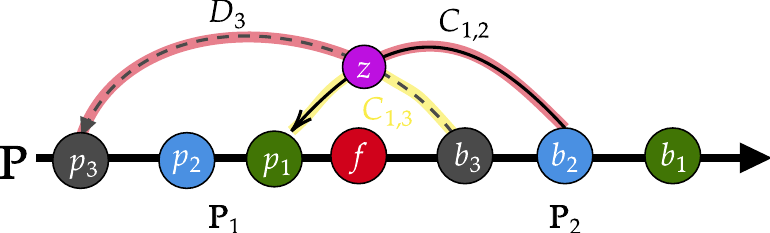}
\end{center}
\end{figure}

Let us proceed by choosing $b_3$ to be the last vertex on $P(f,b_2]$ that has a bad cross with $b_2$.
By the same reasoning, the vertex $v_2 = \dfirst{\Gf}{b'_2}{P_1}{\alpha}$, with $b'_2$ being the successor of $b_3$ on $P$, is an $\ddfirst{\Gf}{b}{P_1}{\alpha}{\eps \alpha}$ for every $b\in P(b_3,b_2]$.
Consequently, a shortest path $D_3$ from $b_3$ to $p_3=\dfirst{\Gf}{b_3}{P_1}{\alpha}$ must intersects $C_{1,2}$  at some vertex $z$.
This yields the path $D_{2,3} = C_{1,2}[b_2,z]\cdot D_3[z,p_3]$ with $\len(D_{2,3}) > (1+\eps) \alpha$ (due to the bad cross between $b_2$ and $b_3$).

Now, the sum of the lengths of $D_3$ and $C_{1,2}$ is bounded by $(2-\eps)\alpha$.
Therefore, we have that  $C_{1,3} = D_3[b_3,z]\cdot C_{1,2}[z,p_1]$ has length at most $(2-\eps)\alpha-(1+\eps)\alpha =  (1-2\eps)\alpha$.

Clearly, this process must terminate after $O(\frac{1}{\eps})$ steps.
The sequences of $b_i$'s and $v_i$'s found in the process are stored in the label of $f$, which is sufficient to classify every $b$ to the subpath of $P$ such that $b\in P(b_i,b_{i-1}]$ and report $v_i$ as an answer for $b$.
Thus, a labeling scheme with labels of size $O(\frac{1}{\eps})$ is achieved for the approximate version of the easy problem.

The labeling scheme we described above is the simplest application of the `good cross-bad cross' technique.
As it turns out, this idea proves useful in generalizing many of the reachability subroutines to the approximate distance settings.
Intuitively, the reachability variants of each subroutine depend on arguments of the form: \textit{"There is a path from $b_1$ to $p_1$ and a path from $b_2$ to $p_2$. These paths intersect, so there is a path from $b_1$ to $p_2$"}.
One would like to have the following similar argument for approximate distances: \textit{"There is a path of length $\alpha$ from $b_1$ to $p_1$ and a path of length $\alpha$ from $b_2$ to $p_2$.
These paths intersect, so there is a path of length $(1+\eps)\alpha$ from $b_1$ to $p_2$"}.
Unfortunately, the latter statement is not true in general.
The `good cross-bad cross' technique exploits the fact that if the latter statement is not true (which we can consider as a `bad event'), then \textit{some} path is cheap.
The cheap path on its own is not necessarily useful, e.g. the path from $b_2$ to $p_1$ is not an approximate shortest path.
However, the length of the cheap path can be used as a decreasing potential that cannot become negative, and therefore bounds the number of `bad' events.

The solution for the easy problem can be described as a selection procedure for 'useful' vertices that terminates quickly due to 'good cross-bad cross' arguments.
The challenge in applying the 'good cross- bad cross' approach to other sub-problems is in defining the selection mechanism.
Specifically, the paths in the easy problem have a convenient structure in the sense that every two paths must intersect.
This is not the case in other sub-problems.
For example, in other sub-problems, the cheap paths created as a result of a bad cross in each step of the algorithm seem entirely unrelated to the sub-problem at hand.
One needs to creatively find an intricate way this relation can be made.

\section{Preliminaries}

\medskip
\noindent
{\bf The decomposition tree.}
In~\cite[Lemma 2.2]{Thorup04}, Thorup proved that we can assume the graph $G$ has an undirected spanning tree $T$ (i.e., $T$ is an unrooted spanning tree in the undirected graph obtained from $G$ by ignoring the directions of edge) such that each path in $T$ is the concatenation of $O(1)$ directed paths in $G$.

This way, we can describe the process of decomposing $G$ into pieces in the undirected version of $G$. After describing the decomposition, we will replace each undirected path of $T$ defined in the process by its $O(1)$ corresponding directed paths in $G$. We therefore proceed to describe the decomposition treating $G$ as an undirected graph with a rooted spanning tree $T$.

A balanced fundamental cycle separator~\cite[Lemma 5.3.2]{planarbook} (cf. \cite{LTsep} and \cite[Lemma 2.3]{Thorup04}) is a simple cycle $C$ in $G$ whose vertices are those of a single path of the (unrooted) spanning tree $T$, such that the removal of the vertices of $C$ and their incident edges separates $G$ into two roughly equal sized subgraphs. The balance of the separator can be defined with respect to a weight function on the vertices of $G$ rather than just the number of vertices.

The recursive decomposition tree $\mathcal T$ of $G$ is defined as follows.
Each node of $\mathcal T$ corresponds to a subgraph of $G$ (called a {\em piece}). The root piece of $\mathcal T$ is the entire graph $G$.
The boundary $\partial G$ of $G$ is defined to be the empty set.
We define the children of a piece $H$ in $\mathcal T$ recursively. Let $\partial H$ be the boundary of $H$, and let $C$ be a fundamental cycle separator which balances the number of non-boundary vertices of $H$.
Let $Q$ be the set of maximal subpaths of $C$ that are internally disjoint from $\partial H$.
To reduce clutter we sometimes refer to a vertex $v \in Q$, by which we mean that $v$ belongs to some path in $Q$.
The vertices of $H$ that are enclosed by $C$ (including the vertices of $C$) belong to one child $H_1$ of $H$. The vertices of subpaths in $Q$ and the vertices of $H$ not enclosed by $C$ belong to the other child $H_2$.
Note that the vertices of $Q$ are the only vertices of $H$ that belong to both children.
The endpoints of $Q$ that belong to $\partial H$ are called the {\em apices} of $H$. The importance of apices arises from the fact that apices are the only vertices that belong to more than two pieces at the same depth of $\mathcal T$.\footnote{
The term apex was previously used in exactly the same context in \cite{AbrahamCG12}.}
We call the paths in $Q$ without the apices of $H$ the {\em separator} of $H$. We do not include the apices in the separator paths to guarantee that the separator is vertex disjoint from $\partial H$.
The boundary $\partial H_i$ of a child $H_i$ of $H$ consists of the separator of $H$ and of the subpaths of $\partial H$ induced by the vertices of $H_i$.

Since $H_2$ might not contain all the vertices of $C$, the subgraph induced on the spanning tree $T$ by $H_2$ may become disconnected.
To overcome this slight technical issue we embed in $H_2$, if necessary, an artificial root connected by artificial edges to the rootmost apex in each of the resulting components of the tree. The embedding remains planar since all these apices were on the fundamental cycle separator of the parent piece $H$. We treat the artificial root and edges as part of the spanning tree of $H_2$. To guarantee that the addition of the artificial root and edges does not affect the reachability of non-artificial vertices of $H_2$, we direct the artificial edges into the artificial root.

The leaves of $\mathcal T$ (called {\em atomic} pieces) correspond to pieces with $O(1)$ non-boundary vertices. The depth of $\mathcal T$ is $O(\log n)$. For convenience, we consider all $O(1)$ vertices of an atomic (leaf) piece that are not already boundary vertices as the separator of the piece.
It follows that the boundary $\partial H$ of any piece $H$ consists of  $\tilde O(1)$ vertex disjoint paths (the subpaths induced by the vertices of $H$ on the separators of the ancestor pieces of $H$), and each of the paths in $\partial H$ lies on a single face of $H$.
Also, since in each node of the decomposition tree only $\tilde O(1)$ apices are created, there are $\tilde O(1)$ apices along any root-to-leaf path in $\mathcal T$.

Having defined the decomposition tree $\mathcal T$ we can go back to treating $G$ as a directed graph. As we explained above, each path we had discussed in the undirected version of $G$ is the union of $O(1)$ directed paths in $G$. From now on when we refer to the separator paths of a piece $H$ (resp., paths of $\partial H$), we mean the set of directed paths comprising the undirected separator paths of $H$ (resp., the set of directed paths comprising the paths of $\partial H$).

To be able to control the size of the labels in our construction we need to be aware of the number of pieces of $\mathcal T$ to which a vertex belongs.
The only vertices of a piece $H$ that belong to both its children are the vertices of the separator path of $H$ and the $\tilde O(1)$ apices of $H$.
The fact that separator paths are disjoint from the boundary imply that every vertex belongs to the separator of at most one piece in $\mathcal T$. Hence, if a vertex is not an apex, it appears in $O(\log n)$ pieces of $\mathcal T$.
Apices, on the other hand require special attention because they may belong to many (e.g., polynomially many) pieces of $\mathcal T$; High degree vertices may be apices in many pieces, and we will need a special mechanism for dealing with such vertices.
Dealing with apices (like dealing with holes in other works on planar graphs) introduces technical complications that are not pertinent to understanding the main ideas of our work.\footnote{A reader who is not interested in those details can safely skip the parts dealing with apices and just act under the assumption that each vertex appears in 2 atomic pieces (leaves) of $\mathcal T$, and that the following definition of ancestor pieces of a vertex $v$ just degenerates to the set of $O(\log n)$ ancestors of the 2 atomic pieces containing $v$.}

\begin{figure}[h]
\begin{center}
\includegraphics[scale=0.18]{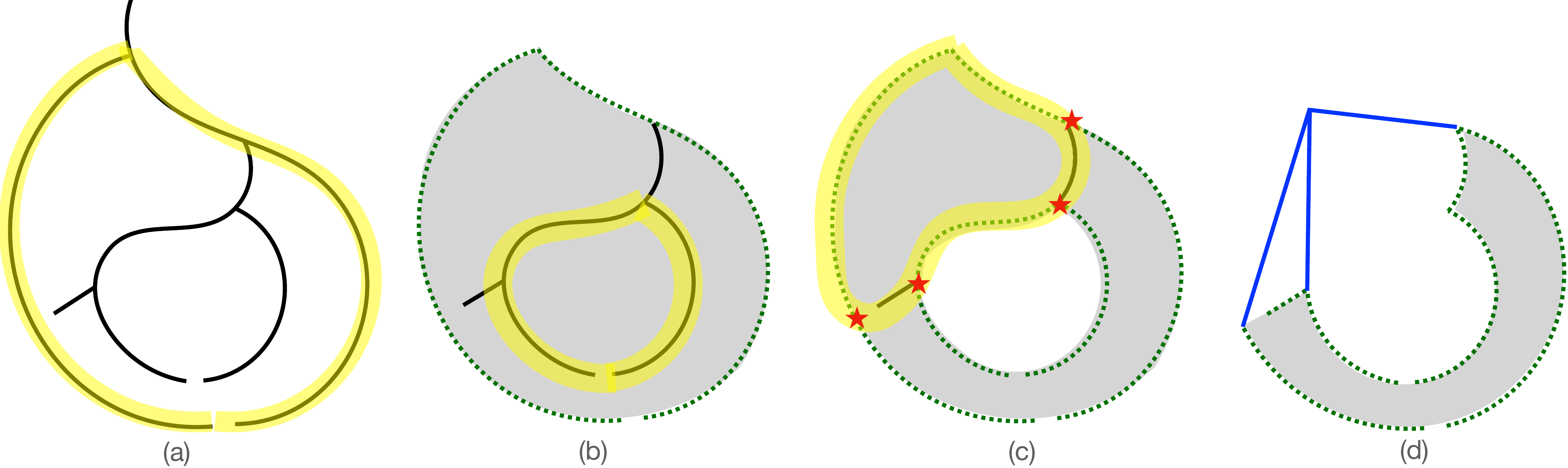}
\caption{Illustration of the recursive decomposition process.
(a) a portion of the spanning tree $T$ of $G$ is shown in black. A balanced fundamental cycle separator of $G$ is highlighted in yellow.
(b) The piece $H$ containing the vertices of $G$ enclosed by the fundamental cycle of $G$ is indicated in shaded grey. The boundary $\partial H$ (dashed) consists of the separator of $G$. Portions of the spanning tree $T$ induced by the vertices of $H$ are shown, along with a fundamental cycle separator of $H$, highlighted in yellow.
(c) The piece $H_{2}$ containing the vertices of $H$ not strictly enclosed by the fundamental cycle separator of $H$ is indicated in shaded grey. The boundary $\partial H_{2}$ (dashed) consists of $\partial H$ and of the separator of $H$. Portions of the spanning tree $T$ induced by the vertices of $H_{2}$ are shown, along with a fundamental cycle separator of $H_2$, highlighted in yellow. There are two  maximal subpaths $Q$ of the fundamental cycle separator of $H_{2}$ that do not belong to $\partial H_{2}$ (the two parts of the highlighted yellow cycle that are not dashed). The endpoints of $Q$ that belong to $\partial H_{2}$ are the apices of $H_{2}$ (the four red stars).
(d) The piece $H_{22}$ containing the vertices of $H_{2}$ not strictly enclosed by the fundamental cycle of $H_{2}$ is indicated in shaded grey. The boundary $\partial H_{22}$ (dashed) consists of the paths induced by $H_{22}$ on $\partial H_{2}$, and of the separator of $H_{2}$. Since the subgraph induced on $T$ by the vertices of $H_{22}$ is disconnected, an artificial root and edges (in blue) are added to $H_{22}$.
\label{fig:decomp}
}
\end{center}
\end{figure}

We associate with every vertex $v\in G$ the (at most 2) rootmost pieces $H$ in $\mathcal T$ in which $v$ is an apex (or the atomic pieces containing $v$ if $v$ is never an apex). We denote these pieces by $H_v$. Note that every piece that contains a vertex $v$ is either an ancestor of a piece in $H_v$ or a descendant of a piece in $H_v$.
For a vertex $v\in G$ we define the {\em ancestor pieces} of $v$  to be the set of (weak) ancestors in $\mathcal T$ of the pieces $H_v$.
By definition of $H_v$, every vertex, apex or not, has $O(\log n)$ ancestors pieces.
We similarly define the ancestor separators/paths/apices of a vertex $v\in G$ as the separators/separator-paths/apices of any ancestor piece of $v$. See Figure~\ref{fig:decomposition}.

\begin{figure}[h]
\begin{center}
\includegraphics[scale=0.18]{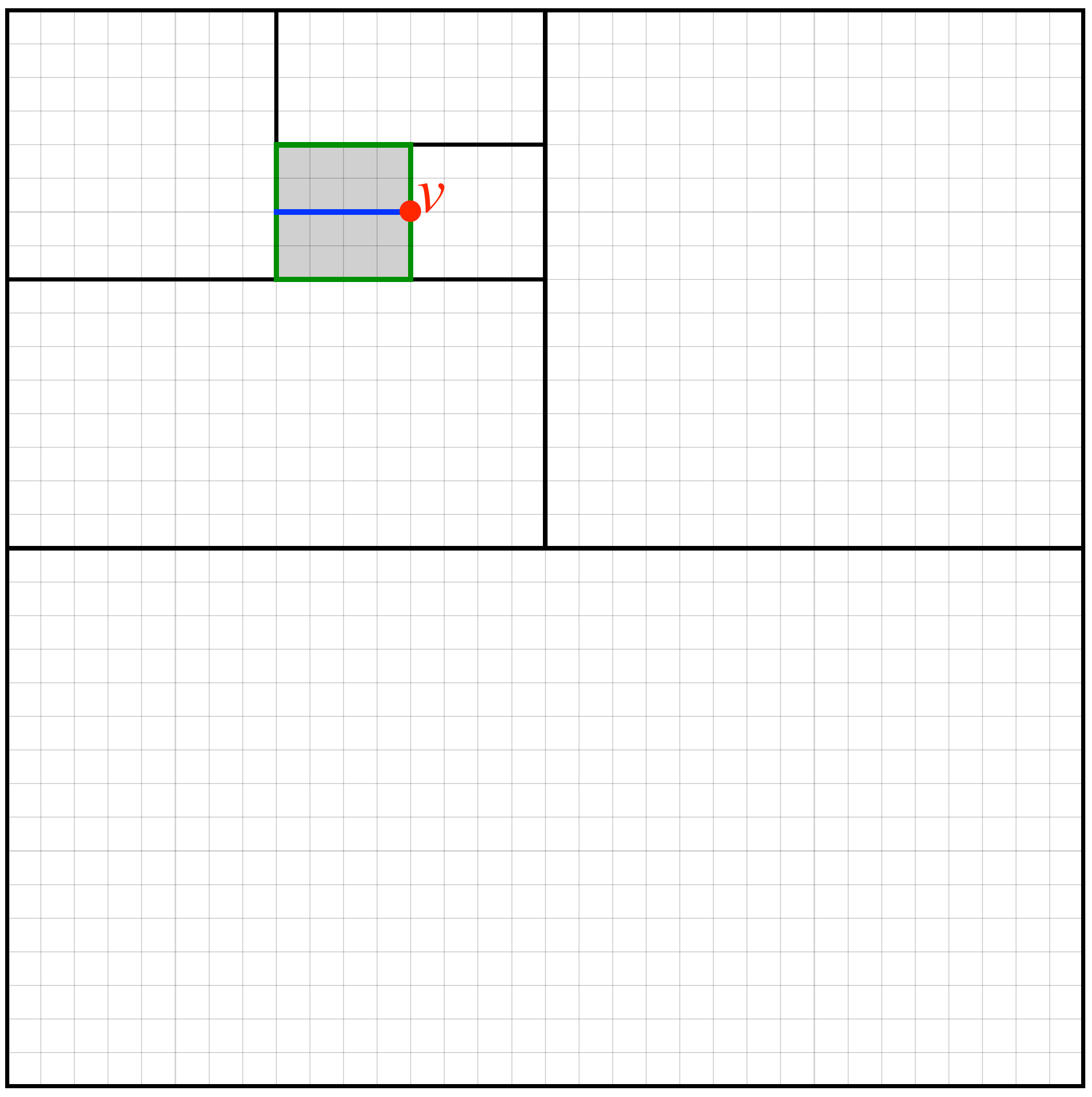}
\caption{Pieces in a recursive decomposition. The separators in this example alternate between horizontal and vertical lines.
All rectangular pieces that contain the gray piece $H$ are its ancestors.
The boundary $\partial H$ of $H$ is shown in green.
The maximal subpath of the cycle separator of $H$ that is internally disjoint from $\partial H$ is shown in blue.
The vertex $v$ is an apex of $H$ because it is an endpoint of this maximal subpath. The separator of $H$ is the blue path (without its endpoints).
\label{fig:decomposition}
}
\end{center}
\end{figure}

We say that the separator $Q$ of a piece $H$ {\em separates} two vertices $u$ and $v$ (in $H$) if any $u$-to-$v$ path in $H$ must touch the separator $Q$ of $H$ or an apex of $H$. I.e., if at least one of the following holds: (1) $u\in Q$ or $u$ is an apex of $H$, or (2) $v\in Q$ or $v$ is an apex of $H$, or (3) each of $u$ and $v$ is in one distinct child of $H$.
Note that
if $Q$ separates $u$ and $v$ in $H$ then every $u$-to-$v$ path in $G$ either touches $Q$ or touches the boundary of $H$.

For a subgraph $H$, a path $P$ and a vertex $v$, let $\first{H}{v}{P}$ denote the first vertex of $P$ that is reachable from $v$ in $H$, and let $\last{H}{v}{P}$ denote the last vertex of $P$ that can reach $v$ in $H$. If vertex $u$ appears before vertex $v$ on a path $P$ then we denote this by $u<_P v$ (or simply $u<v$ if $P$ is clear from the context).
Throughout the paper, we gradually describe the information stored in the labels along with the explanations of why this particular information is stored (and why it is polylogarithmic). To assist the reader, we highlight in gray the parts that describe the information stored. For starters, we let
\hlgray{
every vertex $v\in G$ store in its label, for every ancestor path $P$ of $v$, the identity of $P$ and, if $v \in P$, the location of $v$ in $P$} (so that given two vertices $u,v$ of $P$, we can tell if $u<v$). We denote $P[u,v]$ the subpath of $P$ between vertices $u$ and $v$. To avoid unnecessary repetitions in the text, we assume that this information is included in every labeling scheme described in this paper, and do not include it explicitly in their descriptions.

\medskip
\noindent
{\bf Thorup's non-faulty labeling.}
Using the above definitions and notations, it is now very simple to describe Thorup's non-faulty reachability labeling~\cite{Thorup04}.
Consider any vertex $v$. Let $H$ be the rootmost piece in $\mathcal T$ in which $v$ belongs to the separator.
The crucial observation is that $v$ is separated from every other vertex in $G$ either by the separator of $H$ or by the separator of some ancestor piece of $H$.
Hence, every vertex $v\in G$ stores in its label $\first{G}{v}{P}$ and $\last{G}{v}{P}$ for every path $P$ of the separator of every ancestor of the rootmost piece in which $v$ belongs to the separator.
Then, given a query pair $u,v$, there exists a $u$-to-$v$ path in $G$ if and only if $\first{G}{u}{P} < \last{G}{v}{P}$ for one of the $O(1)$ paths
$P$ of the separator of an ancestor piece of the rootmost piece whose separator separates $u$ and $v$. Both $u$ and $v$ store the relevant information for these paths in their labels. We note that in Thorup's scheme we do not need to worry about apices since each vertex $v$ only stores information in pieces above the first time $v$ appears on a separator.

\section{Reachability Labeling}\label{sec:reach_labels}

In this section, we describe our reachability labeling scheme. I.e., what to store in the labels so that given the labels of any three vertices $s,f,t$ we can infer whether $t$ is reachable from $s$ in $\Gf$. We call the $s$-to-$t$ path $R$ in $\Gf$ the {\em replacement path}.

\begin{theorem}\label{thm:reach}
    There exists a labeling scheme for reachability for $G$ that, given vertices $s,t,f$ returns whether $t$ is reachable from $s$ in $\Gf$.  The size of each label is $\Oone$.
\end{theorem}

\subsection{Reduction into simpler problems}

Instead of describing our labeling scheme that proves \cref{thm:reach} directly, we describe labeling for several specialized tasks, and show how to combine these specialized labeling schemes into the desired labeling scheme for reachability.
While this division may seem somewhat excessive for the description of the reachability labels, it serves to clearly present and explain the high level structure of the scheme, in preparation for the more complicated and elaborate scheme for approximate distances in~\cref{sec:dist_labels}.
The first reduction is to two labeling schemes from a vertex to a path in the presence of a failed vertex, one for the case that the failing vertex is not on the path, and the other for the case that it is.

\begin{restatable}{lemma}{stofpreachlem}\label{lem:stfopreach}
    There exists a labeling scheme $\LabelstfoP=\LabelstfoP_{G,P}$ where $G$ is a planar graph equipped with a decomposition tree $\mathcal T$, and $P$ is a path in $\mathcal T$.
    Let $s$ and $f\notin P$ be two vertices of $G$ that are not an ancestor apex of one another, and such that $P$ is an ancestor of both $s$ and $f$.
    Given the labels of $s$ and $f$, one can compute the index on $P$ of the vertex $b$ that is $\first{\Gf}{s}{P}$.
    In this labeling scheme, the only vertices that store a label are those that have $P$ as an ancestor. The size of each label stored by such a vertex is $\tilde O(1)$.
\end{restatable}

\begin{restatable}{lemma}{stopreachlem}\label{lem:stopreach}
    There exists a labeling scheme $\LabelstoP=\LabelstoP_{H,P}$ where $H$ is a planar graph equipped with a decomposition tree $\mathcal T$, and $P$ is a path in $\mathcal T$ such that both endpoints of $P$ lie on the same face of $H$.
    Let $s$ and $f\in P$ be two vertices of $H$ that are not an ancestor apex of one another, and such that $P$ is an ancestor of both $s$ and $f$.
    Given the labels of $s$ and $f$, one can compute the indices on $P$ of the vertices $b_1 = \first{H\setminus f}{s}{P_1}$ and $b_2 = \first{H \setminus f}{s}{P_2}$, where $P_1$ (resp. $P_2$) is the prefix (resp. suffix) of $P$ that precedes (resp. follows) $f$, excluding $f$.
    In this labeling scheme, the only vertices that store a label are those that have $P$ as an ancestor. The size of each label stored by such a vertex is $\tilde O(1)$.
\end{restatable}

\begin{proof}[Proof of \cref{thm:reach}]
Let $G_{rev}$ be the graph obtained from $G$ by reversing all edges (the reverse graph $G_{rev}$ has exactly the same decomposition tree as $G$, but the direction of each path is reversed).

The label of a vertex $v$ consists of the following:
\begin{enumerate}
\item \label{it:reach:thorup} \hlgray{For each piece $H$ in the recursive decomposition $\mathcal T$ of $G$ such that $v \in H \setminus \partial H$, $v$ stores its label in the standard (non-faulty) labeling of Thorup for $H\setminus \partial H$.}

\item \label{it:reach:last} \hlgray{For every ancestor piece $H$ in $G$, $v$ stores $\last{H}{v}{\P}$ for each of the $O(\log n)$ paths $\P$ of $\partial H$.}

\item \label{it:reach:stfop} Using of \cref{lem:stfopreach}, \hlgray{for every ancestor path $P$ of $v$ in the recursive decomposition $\mathcal T$, $v$ stores $\LabelstfoP_{G,P}(v)$}.

\item \label{it:reach:stop} Using \cref{lem:stopreach}, \hlgray{for every ancestor piece $H$ of $v$ in $\mathcal T$, for every path $P$ of the cycle separator $C$ of $H$, $v$ stores $\LabelstoP_{H^{\times P},P}(v)$}, where $H^{\times P}$ is the graph obtained from $H\setminus \partial H$ by making an incision along all the edges of $C$\footnote{We refer here to the cycle separator $C$ of $H$ and not to the separator $Q$ of $H$ because the separator $Q$ only includes the subpaths of $C$ that do not belong to $\partial H$. Here we want to make the incision along the entire cycle separator.} other than those of $P$. Note that, because of the incision, the endpoints of $P$ lie on a single face of $H^{\times P}$, so \cref{lem:stopreach} indeed applies.
\item \label{it:reach:apex} \hlgray{For every ancestor apex $a$ of $v$, for every path $\P$ that is an ancestor of $v$, $v$ stores in its label $\first{G\setminus a}{v}{\P}$, and  $\first{G\setminus v}{a}{\P}$.}
If $a$ lies on $P$ then let $P_2$ be the suffix of $P$ after $a$ (not including $f$).
$v$ also stores \hlgray{$\first{G\setminus a}{v}{P_2}$}.
Similarly, if $v$ lies on $P$ then let $P_2$ be the suffix of $P$ after $v$ (not including $v$).
$v$ also stores \hlgray{$\first{G\setminus v}{a}{P_2}$}.

\item \hlgray{$v$ also stores all the above items computed in the graph $G_{rev}$ instead of the graph $G$.}
\end{enumerate}

\paragraph{Size.} Since each vertex has only $\tilde O(1)$ ancestor pieces, paths and apices, and by \cref{lem:stfopreach,lem:stopreach}, all items above sum up to a label of size $\tilde O(1)$.

\paragraph{Decoding and Correctness.}
Let $\Htf$ be the rootmost piece in $\mathcal T$ whose separator $Q$ separates $t$ and $f$.
Let $H$ be a child piece of $\Htf$ that contains $t$ (if both children of $\Htf$ contain $t$ then, if one of the children does not contain $f$ we choose $H$ to be that child).
Note that by choice of $H$, $f \notin H \setminus \partial H$.

We assume without loss of generality that $s \in \Htf$. We handle the other case analogously to the description below, by swapping the roles of $s$ and $t$ and working in $G_{rev}$ instead of in $G$.

Observe that by definition of $H$ and of separation, $f \in \partial H$ iff $f \in \Q$.
Consider first the case when a replacement path $R$ does not touch $\partial H$. i.e., $s,t$ and $R$ are all contained in $H\setminus \partial H$, and $f$ is not contained in $H \setminus \partial H$. In this case querying Thorup's non-faulty labels for $H\setminus \partial H$ (stored in item (\ref{it:reach:thorup})) will correctly identify the existence of a replacement path.

To treat the case where the replacement path $R$ touches $\partial H$, we separately handle the cases where $f \notin \Q$ and $f \in Q$.

\paragraph{When $f \notin \Q$ (and so, $f \notin \partial H$).}

In this case, $R$ must have a suffix contained in $H$, and this suffix is unaffected by the fault $f$. More precisely, $R$ exists iff $\first{\Gf}{s}{\P} < \last{H}{t}{\P}$  for one of the paths $\P$ forming $\partial H$.
The vertex $\last{H}{t}{\P}$ can be retrieved from item (\ref{it:reach:last}) of the label of $t$. Notice that by the rootmost choice of $\hat H$, $H$ is an ancestor piece of $t$, so $t$ indeed stores $\last{H}{t}{\P}$.
It thus remains only to describe how to find $\first{\Gf}{s}{\P}$ from the labels of $s$ and $f$.
If either $s$ or $f$ store $\first{\Gf}{s}{P}$ in item (\ref{it:reach:apex}), we are done.
Otherwise neither $s$ nor $f$ is an ancestor apex of the other, and since both $s$ and $f$ are in $\Htf$, $\P$ is indeed an ancestor of both $s$ and $f$, so, by \cref{lem:stfopreach}, $\first{\Gf}{s}{\P}$ can be obtained from  $\LabelstfoP_{G,P}(s)$  (stored in item (\ref{it:reach:stfop}) of the label of $s$) and $\LabelstfoP_{G,P}(f)$  (stored in item (\ref{it:reach:stfop}) of the label of $f$).

\paragraph{When $f \in \Q$.}
Let $\P$ be the path of $\Q$ that contains $f$.
Consider first the case where the replacement path $R$ touches some path $\PP\neq P$ of $\Q$ or of the boundary of some ancestor of $H$. Since boundary paths are vertex disjoint, $f\in \P$ implies $f \notin \PP$.
Hence, we can obtain $\first{\Gf}{s}{\PP}$ in a similar manner to the case $f \notin Q$ above, with $\hat P$ taking the role of $P$. In an analogous manner, we can obtain $\last{\Gf}{t}{\PP}$ as we just found $\first{\Gf}{s}{\PP}$, but with $\hat P$ taking the role of $P$, $t$ taking the role of $s$, and  $G_{rev}$ taking the role of $G$ (we cannot use part (\ref{it:reach:last}) of the label of $t$ in this case because now $f$ does belong to $\partial H$).
Then, $s$ can reach $t$ in $\Gf$ iff $\first{\Gf}{s}{\PP} <_{\PP} \last{\Gf}{t}{\PP}$. See Figure~\ref{fig:touchesornot} (right).

Now consider the case where other than $P$, $R$ does not touch any path of $Q$ or any path of the boundary of an ancestor of $H$. In this case, we might as well use labels in $\hat H^{\times P}$ instead of in $G$ because $R$ does touch $\partial \hat H$, and only crosses $Q$ at $P$.
Let $P_1$ and $P_2$ be the prefix and suffix obtained from $P$ by deleting $f$.
If either $s$ or $f$ is an ancestor apex of one another then
$\first{\Gf}{s}{P_2}$ is stored in item (\ref{it:reach:apex}) of either $s$ or $t$, and, if $\first{\Gf}{s}{P_1}$ exists, then it is equal to $\first{\Gf}{s}{P}$, which is also stored in  item (\ref{it:reach:apex}) of either $s$ or $t$. (If $\first{\Gf}{s}{P}$ is not a vertex of $P_1$ then $\first{\Gf}{s}{P_1}$ does not exist.
If neither $s$ not $t$ is an ancestor apex of the other, then $P$ is an ancestor of both $s$ and $f$, and $\first{\Gf}{s}{P_1}$ and $\first{\Gf}{s}{P_2}$ can be obtained, by \cref{lem:stopreach},  from $\LabelstoP_{\hat H^{\times P},P}(s)$  (stored in item (\ref{it:reach:stop}) of the label of $s$) and $\LabelstoP_{\hat H^{\times P},P}(f)$  (stored in item (\ref{it:reach:stop}) of the label of $f$).
Then, $s$ can reach $t$ in $\Gf$ iff either $\first{\Gf}{s}{P_1} <_{P_1} \last{\Gf}{t}{P_1}$ or $\first{\Gf}{s}{P_2} <_{P_2} \last{\Gf}{t}{P_2}$.
\end{proof}

\subsection{The $\LabelstfoP$ labeling (\cref{lem:stfopreach})}
\label{src:stfopreach}

To show the labeling schemes $\LabelstfoP$ (i.e., prove \cref{lem:stfopreach}) we will compose the following specialized labeling schemes, which we prove in the sequel.

\begin{restatable}{lemma}{atfopreachlem}\label{lem:atfopreach}
    There exists a labeling scheme $\LabelatfoP=\LabelatfoP_{H,P,P'}$ where $H$ is a planar graph, and $P$ and $P'$ are two paths.
    Given the labels of a vertex $a$ on $P'$ and a vertex $f$ not in $P'\cup P$, one can retrieve the index on $P$ of the vertex $b = \first{H\setminus f}{a}{P}$.
    The size of each label is $\Oone$.
\end{restatable}

\begin{restatable}{lemma}{stoareachlem}\label{lem:stoareach}
    There exists a (trivial) labeling scheme    $\Labelstoa=\Labelstoa_{H, P'}$  where $H$ is a planar graph with a path $P'$.
    Given the label of a vertex $s$, one can retrieve the vertex $\first{H}{s}{P'}$.
    The size of each label is $\Oone$.
\end{restatable}

We note that the proof of \cref{lem:stoareach} is trivial, as each vertex $s$ of $H$ only needs to store the identity of $\first{H}{s}{P'}$. The proof of \cref{lem:atfopreach} appears after the following proof of \cref{lem:stfopreach}.

\begin{proof}[Proof of \cref{lem:stfopreach}]
The labeling scheme only labels vertices whose ancestor is $P$. The label of such a vertex $v$ consists of the following:
\begin{enumerate}

\item Using \cref{lem:atfopreach}, \hlgray{for every ancestor path $P'$ of $v$ that has $P$ as an ancestor, the label $\LabelatfoP_{G,P',P}(v)$.}

\item (Composition of labels) For a path $P'$, let $G^1$ be the graph $G$ labeled with the labels $\LabelatfoP_{G,P',P}$ of \cref{lem:atfopreach}.
Using \cref{lem:stoareach}, \hlgray{for every ancestor path $P'$ of $v$ such that $P$ is an ancestor of $P'$, $v$ stores the label $\Labelstoa_{G^1,P'}(v)$} (i.e., the label $\Labelstoa_{G^1,P'}(v)$ contains not only the identity of the desired vertex $a=\first{H}{v}{P'}$, but also the label $\LabelatfoP_{G,P',P}(a)$ that $a$ is labeled with in $G^1$).
\end{enumerate}

\paragraph{Size.} Since each vertex has $\tilde O(1)$ ancestor paths and by \cref{lem:stopreach,lem:atfopreach}, the size of the label of each vertex is $\tilde O(1)$.

\paragraph{Decoding and Correctness.}
Assume, per the statement of the lemma that $P$ is an ancestor of both $s$ and $f$ and that $s$ and $f$ are not an ancestor apex of one another.
Consider the set of leafmost pieces in $\mathcal T$ that contain both $s$ and $f$. Let $\Hsf'$ be such a piece. It must be that $\Hsf'$ is an ancestor piece of both $s$ and $f$ or else one of $s$ and $f$ is an ancestor apex of the other. Hence, there are only $O(1)$ leafmost pieces that contain both $s$ and $f$.
To avoid unnecessary clutter we shall assume there is a unique piece $\Hsf'$. In reality we would have to apply the same argument for all $O(1)$ such pieces.
Since $\Hsf'$ is an ancestor piece of both $s$ and $f$, we can find the piece $\Hsf'$ by traversing the list of ancestors of $s$ (stored in $s$) and of $f$ (stored in $f$) until finding the lowest common ancestor.
Let $\Hsf$ be the child piece of $\Hsf'$ that contains only $s$ (if $\Hsf'$ is an atomic piece then define $\Hsf=\Hsf'$).
Since $P$ is an ancestor of both $s$ and $f$, $P$ is a boundary path of a (possibly weak) ancestor $\Htf$ of $\Hsf'$.

Consider a path in $\Gf$ from $s$ to $\first{\Gf}{s}{\P}$.
Such a path begins with a prefix that is contained in $\Hsf$ from $s$ to $a=\first{\Hsf}{s}{\Psf}$ where $\Psf$ is some path of $\partial\Hsf$.
Since $f \notin \Hsf$, the candidate $a$ can be retrieved from the label $\Labelstoa$ of $s$, stored in item (3), along with the label of $a$ in $\LabelatfoP_{G,P',P}$. From this label, and from the label $\LabelatfoP_{G,P',P}$ of $f$ stored in item (2) we can obtain $\first{\Gf}{a}{\P}$ (which is equal to $\first{\Gf}{s}{\P}$). Note that the conditions of \cref{lem:atfopreach} apply since $f \notin H'$ (and hence $f \notin P'$), and $f \notin P$.
\end{proof}

\begin{figure}[htb]
  \begin{center}
 \includegraphics[scale=0.18]{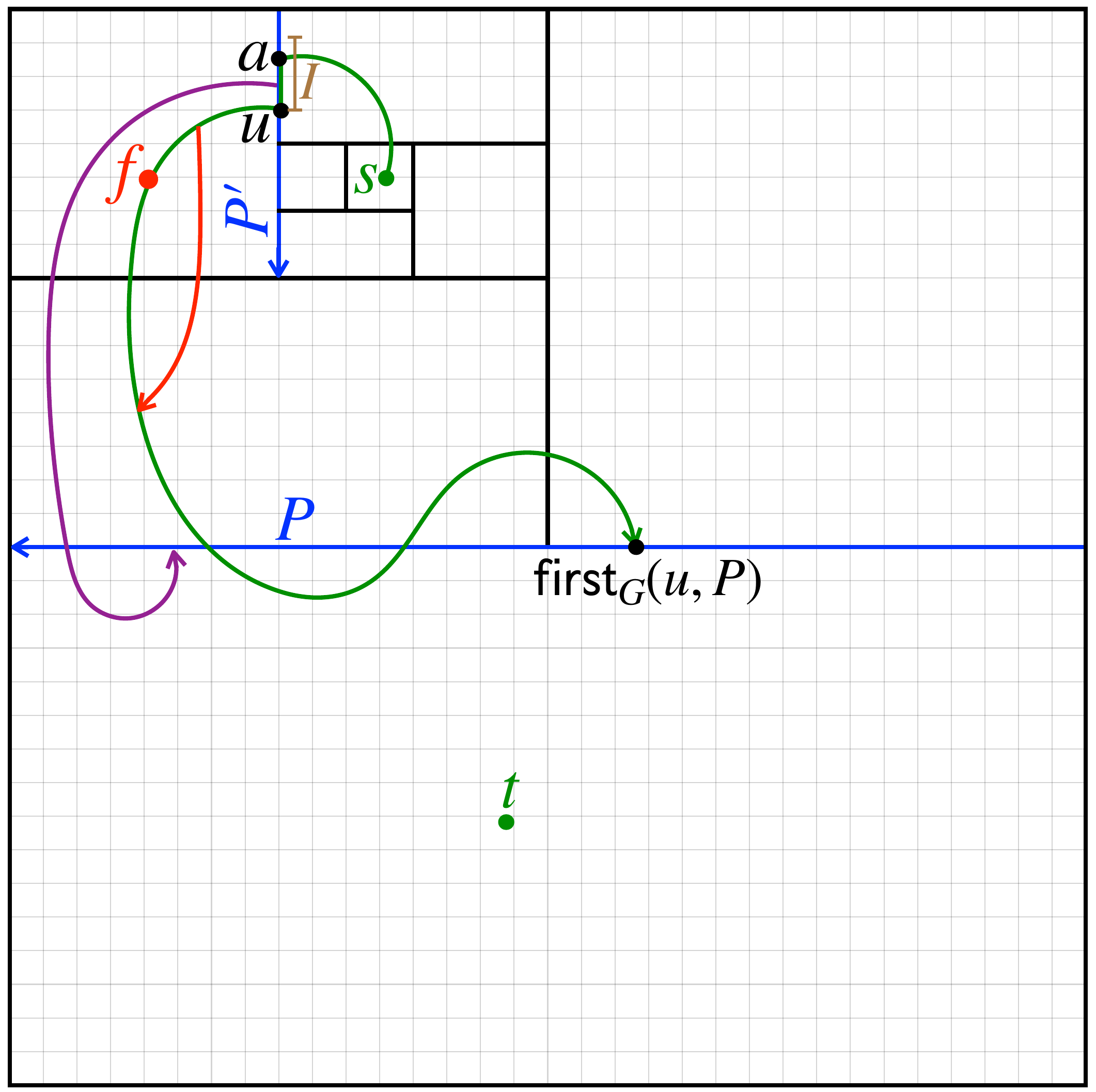}
  \caption{When $f \notin \Q$: The (green) $s$-to-$t$ path in $G$ first touches the (blue) separator path $\Psf\in \partial\Hsf$ at vertex $a$ (that is found using \cref{lem:stoareach}), then continues along $\Psf$ to $\u$, then goes through $f$ to $\first{G}{u}{\P}$ (that is found using \cref{lem:atfopreach}) on the (blue) separator path $\P \in \partial H$. The (brown) interval $I$ contains all vertices $p$ that can reach $\first{G}{p}{\P}$ through $u$. When $f$ fails, the replacement $s$-to-$t$ path in $\Gf$ either takes a (red) detour around $f$ (and then $f$ stores this information) or is (purple) completely disjoint from the original (green) $u$-to-$\first{G}{u}{\P}$ part (and then $s$ stores this information).\label{fig:fnotinQ}}
   \end{center}
\end{figure}

\begin{proof}[Proof of \cref{lem:atfopreach}.]

Consider the Pair of paths $\Psf$ and $\P$ in the statement of the lemma.
We shall define a set $C = C_{\Psf,\P}$ of canonical paths from $\Psf$ to $\P$ in the graph $H$.
Notice that for two vertices $p \leq_{\Psf} q$ on $\Psf$, $\first{G}{p}{\P} \leq_{\P} \first{G}{q}{\P}$ (this is because any vertex of $\Psf$ earlier than $q$ can reach $\first{G}{q}{\P}$ through $q$).
It follows that $\Psf$ can be partitioned into maximal contiguous intervals of vertices that have the same $\first{G}{q}{\P}$.
For every such interval $\I$, we can think of all $p \in \I$ as first going along $\Psf$ until the last vertex $\u$ of $\I$, and then reaching $\first{G}{\u}{\P}$ in $G$ using the same path $S_\I$. We call $S_\I$ the canonical path of interval $\I$, and let $C$ be the set of all canonical paths for all the intervals in the partition of $\Psf$.
For a vertex $a\in I$, we call $S_I$ the canonical interval (of $\Psf,\P$) used by $a$.

Observe that the paths in $C$ are mutually disjoint, since otherwise the corresponding intervals would be the same interval.
It follows that $f$ can be on at most one of those canonical paths.
The canonical paths and intervals are useful because their definition only depends on $\Psf$ and $\P$, but not on the identity of $s$ and $f$.

The label of a vertex $v$ consists of the following:

\begin{enumerate}
        \item \label{it:reach:first} \hlgray{$\first{H}{v}{\P}$.}

    \item \hlgray{if $v$ lies on $P'$, $v$ stores:}
    \begin{enumerate}
        \item \label{it:reach:index} \hlgray{its index on $\Psf$.}
        \item \label{it:reach:avoidcanon} \hlgray{$\first{H\setminus S_I}{v}{\P}$, where $S_I$ is the canonical path used by $v$.}
    \end{enumerate}

    \item if $v$ lies on some canonical path $S_\I$ of $C_{\Psf,\P}$, we let $v$ store:
    $v$ also stores:
    \begin{enumerate}
    \item \label{it:reach:canon} \hlgray{the indices in $\Psf$ of the endpoints of $\I$.}
        \item \label{it:reach:canon2} \hlgray{the index of the last vertex of the prefix of $I$ that can reach $\first{H}{f}{\P}$ in $H\setminus v$.}
        \item \label{it:reach:canon3} \hlgray{$\first{H\setminus v}{u}{\P}$, where $u$ is the last vertex of $I$.}
    \end{enumerate}

\end{enumerate}

\paragraph{Size.} Clearly, each vertex stores $\tilde O(1)$ bits.

\paragraph{Decoding and Correctness.}
Given the labels of $a$ and $f$, we can check if the index of $a$ on $\Psf$ (stored in item (\ref{it:reach:index}) is in the interval $\I$ stored in item (\ref{it:reach:canon}) of the label of $f$.
If $a$ is not in the interval $\I$, then the path in $H$ from $a$ to $\first{H}{a}{\P}$ does not go through $f$ and therefore $\first{H\setminus f}{a}{\P}=\first{H}{a}{\P}$, which is available in item (\ref{it:reach:first}) of the label of $a$.
Otherwise, $a$ is in the interval $\I$ stored by $f$, and the path in $H$ from $a$ to $\first{H}{a}{\P}$ does go through $f$.
In particular $\first{H}{a}{\P} = \first{H}{f}{\P}$.
Let $S_1$ be the prefix of $S_I$ ending just before $f$, and let $S_2$ be the suffix of $S_I$ starting immediately after $f$. Then, there are two options regarding the path in $H \setminus f$ from $a$ to $\first{H \setminus f}{a}{\P}$:

\begin{enumerate}

\item The path intersects $S_2$. In this case, the path can continue along $S_2$ until reaching $\first{H}{f}{\P}$, so $\first{H \setminus f}{a}{\P}= \first{H}{f}{\P}$ and we have it stored in item (\ref{it:reach:first}) of the label of $f$. It only remains to check if this is indeed the case. Consider all vertices $p\in \I$ that can reach $\first{H}{f}{\P}$ in $H \setminus f$. They must constitute a (possibly empty) prefix of $\I$ (since any such vertex $p\in \I$ can reach any later vertex of $\I$ by going along $\I$). Item (\ref{it:reach:canon2}) of the label of $f$ stores the index of the last vertex of the prefix of $\I$ that can reach $\first{H}{f}{\P}$ in $H \setminus f$. Therefore, we only need to check if $a$ is earlier on $P'$ than this vertex.

\item The path is disjoint from $S_2$. To handle this case, we will identify two valid candidates for $\first{H\setminus f}{a}{\P}$ and take the earlier in $\P$ of these two candidates:

\begin{enumerate}

\item The path is disjoint from $S$. To handle this case, we use item (\ref{it:reach:avoidcanon}) of the label of $a$, which stores the first vertex of $\P$ that is reachable from $a$ in $H$ using a path that is disjoint from the canonical path $S_\I$) used by $a$.
Since $f \in S_I$ then $f \notin  H\setminus S_I$ and so $\first{H\setminus S_I}{a}{\P}$ is a valid candidate for $\first{H \setminus f}{s}{\P}$.

\item  The path intersects $S_1$, so it might as well go through the first vertex $\u$ of $S_1$ (the last vertex of the interval $\I$ stored by $f$). To handle this case, Hence, $\first{H \setminus f}{\u}{\P}$ (stored in item (\ref{it:reach:canon3}) of the label of $f$) is also a valid candidate for $\first{H \setminus f}{a}{\P}$. \qedhere
\end{enumerate}
\end{enumerate}
\end{proof}

\begin{figure}[htb]
  \begin{center}
 \includegraphics[scale=0.205]{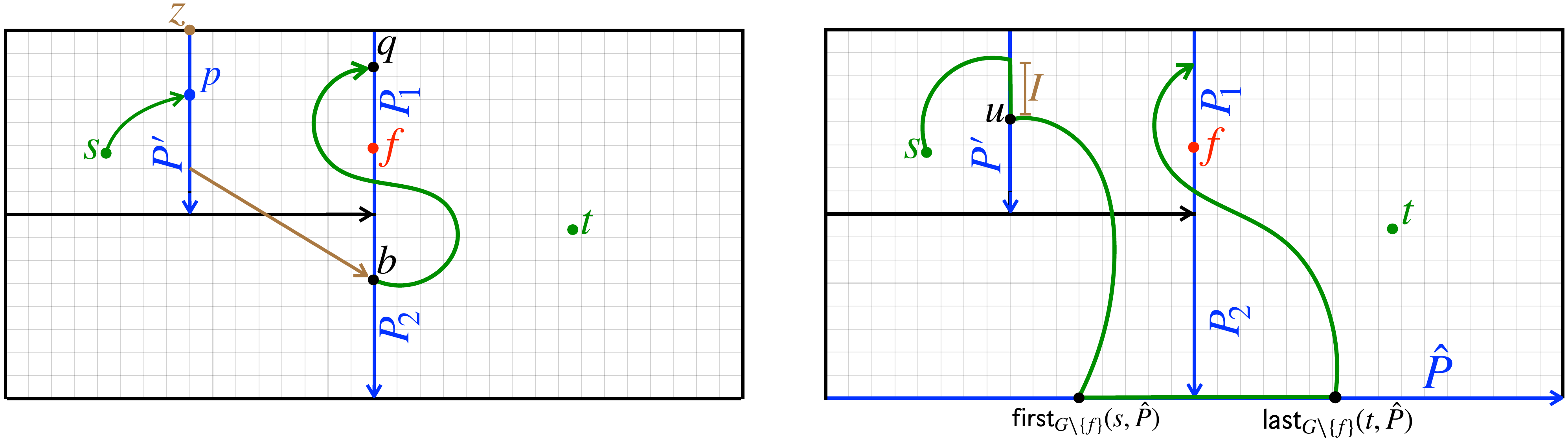}
   \caption{When $f\in \Q$: In the left example, the (green) replacement path $R$ touches only $\P\in \Q$ (the path on which $f$ lies).
   In this case, $b$ is the first vertex of $\Pb$ that is reachable from $s$ in a path internally disjoint from $\Q$. The path from $b$ to $q=\first{\Gf}{b}{\Pa}$ is then found using the mechanism of Section~\ref{sec:singlePath}.
   In the right example, the (green) replacement path $R$ touches some other path $\hat P \neq P$ (on which $f$ does not lie). In this case, using  \cref{lem:stfopreach} twice we check if
   $\first{\Gf}{s}{\PP} <_{\PP} \last{\Gf}{t}{\PP}$.
    \label{fig:touchesornot}}
   \end{center}
\end{figure}

\subsection{The  $\LabelstoP$ labeling (\cref{lem:stopreach})}

To show the labeling scheme $\LabelstoP$ (i.e., prove \cref{lem:stopreach}) we will compose the following specialized labeling schemes, which we prove in the sequel.

\begin{restatable}{lemma}{atopfreachlem}\label{lem:atopfreach}
    There exists a labeling scheme
    $\LabelatoP=\LabelatoP_{H,P,P'}$ where $H$ is a planar graph with paths $P'$ and $P$. Such that $P\cap P'=\emptyset$, $P$ has no outgoing edges, and $P$ lies on a single face of $H$.
    For $f\in P$ let $P_1$ and $P_2$ be the prefix and suffix of $P$ before and after $f$ (without $f$), respectively.
    Given the labels of two vertices $a\in P'$ and $f\in P$, one can retrieve the two vertices $b_1 = \first{H\setminus\{f\}}{a}{P_1}$ and $b_2 = \first{H\setminus\{f\}}{a}{P_2}$.
    The size of each label is $\Oone$.
\end{restatable}

\begin{restatable}{lemma}{ptopreachlem}\label{lem:ptopreach}
    There exists a labeling scheme $\LabelPtoP=\LabelPtoP_{G,P}$ where $G$ is a planar graph, and $P$ is a path of $G$ whose two endpoints lie on the same face.
    Given the labels of two vertices $b$ and $f$ on $P$ let $P_1$ and $P_2$ be the prefix and suffix of $P$ before and after $f$ (without $f$), respectively.
    One can compute the indices on $P$ of some vertices $b_1 = \first{\Gf}{b}{P_1}$ and $b_2=\first{\Gf}{b}{P_2}$.
    The size of each label is $\Oone$.
\end{restatable}

We first prove \cref{lem:stopreach} assuming \cref{lem:ptopreach,lem:atopfreach}, and then prove the two latter lemmas.
\begin{proof}[Proof of \cref{lem:stopreach}]
Let $H$ and $P$ be as in the statement of the lemma.
The labeling is obtained by composition of labels via \cref{lem:stfopreach,lem:atopfreach,lem:ptopreach}.
To this end we define auxiliary labeled graphs.
Let $H_P$ be the graph $H$ with the vertices of $P$ labeled with their labels in $\LabelPtoP_{H,P}$ of \cref{lem:ptopreach}.
Let $H_P^0$ be the graph obtained from $H_P$ by deleting all the edges outgoing from the vertices of $P$.
For a path $P'$ such that $P$ is an ancestor path of $P$, Let $H_{P,P'}$ be the graph $H^0_P$ with the vertices of $P$ and $P'$ labeled with their labels in $\LabelatoP_{H^0_P,P',P}$ of \cref{lem:atopfreach}.

The label of a vertex $v$ consists of:
\begin{enumerate}
    \item \label{it:reach:stoa-stopf}\hlgray{For every ancestor piece $H'$ of $v$ in $H_{P,P'}$ such that $P$ is an ancestor path of $H'$, for every boundary path $P'$ of $H'$, $v$ stores $\Labelstoa_{H',P'}(v)$} using \cref{lem:stoareach}.
    (i.e., the label $\Labelstoa_{H',P'}(v)$ contains not only the identity of the desired vertex $a=\first{H'}{v}{P'}$, but also the label of $a$ in $H_{P,P'}$).

    \item if $v$ is a vertex of $P$, $v$ stores:
    \begin{enumerate}
        \item \label{it:reach:atop-stopf} \hlgray{For every ancestor path $P'$ of $v$ such that $P$ is an ancestor of $P'$, the $\LabelatoP_{H^0_P,P',P}(v)$.}
        \item \label{it:reach:ptop-ptopf} \hlgray{$\LabelPtoP_{H,P}(v)$.}
    \end{enumerate}
\end{enumerate}

\paragraph{Size.} Since each vertex has $\tilde O(1)$ ancestor paths and by \cref{lem:stopreach,lem:atopfreach,lem:ptopreach}, the size of the label off each vertex is $\tilde O(1)$.

\paragraph{Decoding and Correctness}
As in the setup of the proof of \cref{lem:stfopreach}, assume, per the statement of the lemma that $P$ is an ancestor of both $s$ and $f$ and that $s$ and $f$ are not an ancestor apex of one another.
Consider the set of leafmost pieces in $\mathcal T$ that contain both $s$ and $f$. Let $\Hsf'$ be such a piece. It must be that $\Hsf'$ is an ancestor piece of both $s$ and $f$ or else one of $s$ and $f$ is an ancestor apex of the other. Hence, there are only $O(1)$ leafmost pieces that contain both $s$ and $f$.
To avoid unnecessary clutter we shall assume there is a unique piece $\Hsf'$. In reality we would have to apply the same argument for all $O(1)$ such pieces.
Since $\Hsf'$ is an ancestor piece of both $s$ and $f$, we can find the piece $\Hsf'$ by traversing the list of ancestors of $s$ (stored in $s$) and of $f$ (stored in $f$) until finding the lowest common ancestor.
Let $\Hsf$ be the child piece of $\Hsf'$ that contains only $s$ (if $\Hsf'$ is an atomic piece then define $\Hsf=\Hsf'$).
Since $P$ is an ancestor of both $s$ and $f$, $P$ is a boundary path of a (possibly weak) ancestor $\Htf$ of $\Hsf'$.

We give the algorithm for identifying $b_1$. The argument for $b_2$ is identical (just replace $b_1$ by $b_2$).
Consider a path $R$ from $s$ to $b_1$ in $H$.
Let $b$ be the first vertex of $R$ that belongs to $P$.
Note that the prefix $R[s,b]$ does not contain any edge outgoing from any vertex of $P$. Hence, identifying this prefix can be done in the graphs $H^0_P$ and $H_{P,P'}$ in which all such edges were removed.
In other words, we may assume without loss of generality that $b = \first{H^0_P}{s}{P}$.
Consider the maximal prefix of $R[s,b]$ that is entirely contained in $H'$. Note that since $f \notin H'$, $f$ is not on this maximal prefix.
Let $P'$ be the path of $\partial H'$ at which this maximal prefix of $R$ terminates.
Let $a$ be the vertex of $P'$ at which this prefix terminates.
We may assume without loss of generality that $a = \first{H'}{s}{P'}$. Furthermore, since $R[s,b]$ visits $a$, $b = \first{H^0_P}{s}{P}$ is also $b = \first{H^0_P}{a}{P}$.

We now show that the vertices $a,b,b_1$ can be retrieved from the labels of $s$ and of $f$.
The vertex $a$ can be retrieved, by \cref{lem:stoareach} from the label $\Labelstoa_{H',P'}(s)$ stored in item (\ref{it:reach:stoa-stopf}) of the label of $s$, along with the label $\LabelatoP_{H^0_P,P',P}(a)$.
By \cref{lem:atopfreach}, the label of $b$ in $H^0_P$ can be obtained from $\LabelatoP_{H^0_P,P',P}(a)$ and from $\LabelatoP_{H^0_P,P',P}(f)$, which is stored in item (\ref{it:reach:atop-stopf}) of the label of $f$.
Note that, by definition of $H^0_P$, the label of $b$ in $H^0_P$ is $\LabelPtoP_{H,P}(b)$.
Finally, from $\LabelPtoP_{H,P}(b)$, and $\LabelPtoP_{H,P}(f)$ (stored in item (\ref{it:reach:ptop-ptopf}) of the lable of $f$), we can obtain, by \cref{lem:ptopreach}, the identity of $b_1$ in $H$.
\end{proof}

We next prove \cref{lem:atopfreach}, restated here for convenince.
\atopfreachlem*
\begin{proof}[Proof of \cref{lem:atopfreach}.]
Let $a$ be a vertex of $P'$ and $f$ be a vertex of $P$ as in the statement of the lemma.
Let $z$ be the first vertex of $\Psf$.
For a vertex $v$, Let $\Pz(v)$ denote the set of vertices of $\P$ that are reachable from $v$ in $H$.
Since anything reachable from $a$ is also reachable from $z$, we have that $\Pz(a)\subseteq \Pz(z)$. In fact, planarity dictates the following:
\begin{claim}\label{claim:z}
	The set $\Pz(a)$ consists of at most two intervals of consecutive vertices of $\Pz(z)$.
\end{claim}
\begin{proof}
Consider two vertices $u,v \in \Pz(a)$ that belong to two disjoint intervals of $\Pz(z)$. Let $C$ be the (undirected) cycle formed by the $a$-to-$v$ path, the $a$-to-$u$ path, and $\P[u,v]$ (see Figure \ref{fig:z}).
Since $P$ lies on a single face, no vertices of $P$ other than those of $P[u,v]$ are (even weakly) enclosed by $C$.
The vertex $z$ must be enclosed by $C$, otherwise, the path from $z$ to any vertex of $\P[u,v]$ must intersect the $a$-to-$v$ path or the $a$-to-$u$ path (and is thus reachable from $p$ as well in contradiction to the two intervals being disjoint). Since $z$ is enclosed by $C$, any vertex of $\P$ that is reachable from $z$ and does not belong to $\P[u,v]$ is also reachable from $z$ (since the path to it from $z$ must intersect the $a$-to-$v$ path or the $a$-to-$u$ path).
\end{proof}

Equipped with the structure described in \cref{claim:z}, the label of a vertex $v$ consists of the following:
\begin{enumerate}
    \item If $v$ is a vertex of $P'$, $v$ stores:
    \begin{enumerate}
        \item \label{it:reach:b1-z} \hlgray{$\first{H}{v}{P}$.}
        \item \label{it:reach:Pz-z}\hlgray{ the identities of the endpoints of the two intervals of $\Pz(v)$ in $\Pz(z)$ (where $p$ and $z$ are as defined above).}
    \end{enumerate}

    \item \label{it:reach:f-z} If $v$ is a vertex of $P$, \hlgray{$v$ stores the identity of the first vertex $u$ of $\Pz(z)$ that is after $v$ on $\Psf$ (where $z$ is the first vertex of $\Psf$).}
\end{enumerate}

The size of the label is clearly $O(1)$.

To obtain $b_1$ we simply check whether $\first{H}{a}{P}$ (stored in item (\ref{it:reach:b1-z}) of the label of $a$ appears before $f$ on $P$. If so, since $P$ has no outgoing edges, $\first{H}{a}{P} = \first{H \setminus f}{a}{P_1}$. Otherwise, $P_1$ is not reachable from $a$ in $H \setminus f$.

Next we describe how to find $b_2 = \first{\Htf^\circ\setminus \Q}{s}{\Pb}$. If the vertex $u$ stored in item (\ref{it:reach:f-z}) of the label of $f$ falls inside one of the two intervals stored in item (\ref{it:reach:Pz-z}) of the label of $a$ then $b_2=u$.
Otherwise, $b_2$ is the earliest starting endpoint that is later than $f$ among the two endpoints of the two intervals stored in the label of $a$ (if both intervals are before $f$ on $\P$ then $b_2$ does not exist).
\end{proof}

\begin{figure}[htb]
  \begin{center}
 \includegraphics[scale=0.205]{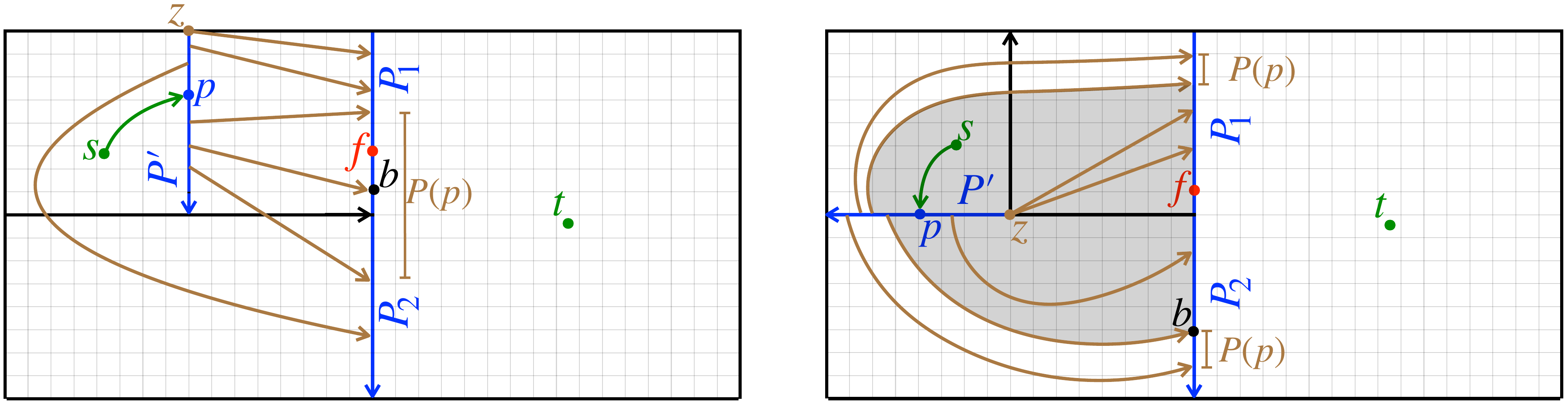}
   \caption{When $f\in \Q$: The vertical (blue) separator path $\P$ (in this example the separator $\Q$ consists of just a single path $\P$) is partitioned by $f$ into $\Pa$ and $\Pb$. The vertex $z$ is the first vertex of $p$'s (blue) path $\Psf$. The brown paths are the ways that $z$ can reach $\P$ in $\Htf^\circ\setminus \Q$. The vertex $b$ is the first vertex of $\Pb$ that is reachable from $p$ in $\Htf^\circ\setminus\Q$. In the left image, $f$ lies inside the interval of vertices $\Pz(p)$ that are reachable from $p$ ($b$ is therefore stored in $f$). In the right image, $f$ lies outside the two intervals $\Pz(p)$ (and $b$ is therefore stored in $s$). The shaded area is the cycle $C$ in the proof of Claim~\ref{claim:z}.  \label{fig:z}}
   \end{center}
\end{figure}

\subsection{The $\LabelPtoP$ labeling (\cref{lem:ptopreach})} \label{sec:singlePath}
In this section we present the secondary labeling scheme  (with polylogarithmic-size labels), which addresses the following problem: We are given a directed planar graph $G$ and a single path $\P$ in $G$ whose endpoints lie on the same face. We need to label the vertices of $\P$ such that given the labels of any two vertices $b,f$ of $\P$ we can find the first vertex $p < f$ of $\P$ that is reachable from $b$ in $\Gf$ and the first vertex $p > f$ of $\P$ that is reachable from $b$ in $\Gf$.

\subsubsection*{An auxiliary procedure}\label{sec:auxiliary}
We begin with an auxiliary procedure that will be useful for our labeling.
In this procedure, we wish to label the vertices of $\P$, such that given the labels of any two vertices $b,f$ such that $f$ is before $b$ on $P$ (i.e., $f < b$), we can find the first vertex $p>f$ of $P$ that is reachable in $G$ from $b$ using a path that does not touch $f$ or any vertex before $f$ (i.e., does not touch any vertex $v \leq f$ of $\P$).

Let $H_P$ be the graph composed of the path $\P$ and the following additional edges: for every pair of vertices $u,v \in \P$ where $u>v$, we add an edge $(u,v)$ iff (1) there exists a $u$-to-$v$ path in $G$ that does not touch $\P$ before $v$, and (2) there is no such $w$-to-$v$ path for any $w>u$. The following claim shows that instead of working with $G$, we can work with $H_\P$:

\begin{claim}\label{claim:HP}
	Given any $f < b$, the first vertex $p>f$ that is reachable from $b$ using a path that does not touch $f$ or any vertex of $\P$ before $f$, is the same in $G$ and in $H_P$.
\end{claim}
\begin{proof}
Let $S$ be the corresponding $b$-to-$p$ path in $G$ for $f<p<b$.
The path $S$ visits no vertex of $\P$ before $f$ (by definition of $S$) and no vertex of $\P$ between $f$ and $p$ (by definition of $p$). Hence, in $H_P$ there is an edge $(u,p)$ for some $u$ where $u=b$ or $u>b$, and so $S$ is represented in $H_P$ (by a path composed of a $b$-to-$u$ prefix along $\P$ followed by the single edge $(u,p)$). In the other direction, let $R$ be the corresponding $b$-to-$p$ path in $H_P$. The path $R$ visits no vertex of $\P$ before $f$ (by definition of $R$) and every edge $(u,v)$ of $R$ corresponds to a path in $G$ that visits no vertex of $\P$ before $v$ (and therefore visits no vertex of $\P$ before $f$), hence $R$ is appropriately represented in $G$.
\end{proof}

We call the edges of $H_P$ that are not edges of $\P$ {\em detours}.
We will use the fact that detours do not cross (i.e., they form a laminar family):

\begin{claim}\label{claim:nested}
If there is a detour $(u,v)$ then there is no detour  $(w,x)$ with $v<x<u<w$.
\end{claim}
\begin{proof}
By concatenating the $(w,x)$ detour, the $x$-to-$u$ subpath of $\P$, and the $(u,v)$ detour, we get a path in $G$ that does not touch $\P$ before $v$ but starts at a vertex $w>u$ contradicting condition (2) of the $H_P$ edges.
\end{proof}

We now explain what needs to be stored to facilitate the auxiliary procedure.
We define the {\em size} $|d|$ of a detour $d=(u,v)$ as the number of vertices of $\P$ between $v$ and $u$.
We say that a vertex of $\P$ is {\em contained} in a detour $d=(u,v)$ if it lies on $\P$ between $v$ and $u$. We say that a detour $d'$ is {\em contained} in detour $d$ if all vertices that are contained in $d'$ are also contained in $d$.
Notice that, from Claim~\ref{claim:HP}, the vertex $p$ sought by the auxiliary procedure is an endpoint of the largest detour $d$ that contains $b$ and does not contain $f$. To find $d$, \hlgray{every vertex $v$ of $\P$ stores a set of $O(\log n)$ nested detours $d^v_1,d^v_2,\dots$ where $d^v_1$ is the largest detour containing $v$, and $d^v_{i+1}$ is the largest detour of size $|d^v_{i+1}| \le |d^v_i|/2$ containing $v$}. By Claim~\ref{claim:nested}, this set of nested detours is well defined.
\hlgray{Additionally, for each $d^v_i$, $v$ stores the largest detour $\hat{d}{^v_i}$ that is strictly contained in $d^v_i$ but does not contain $d^v_{i+1}$.}
See Figure \ref{fig:nested}.

\begin{figure}[htb]
  \begin{center}
 \includegraphics[scale=0.3]{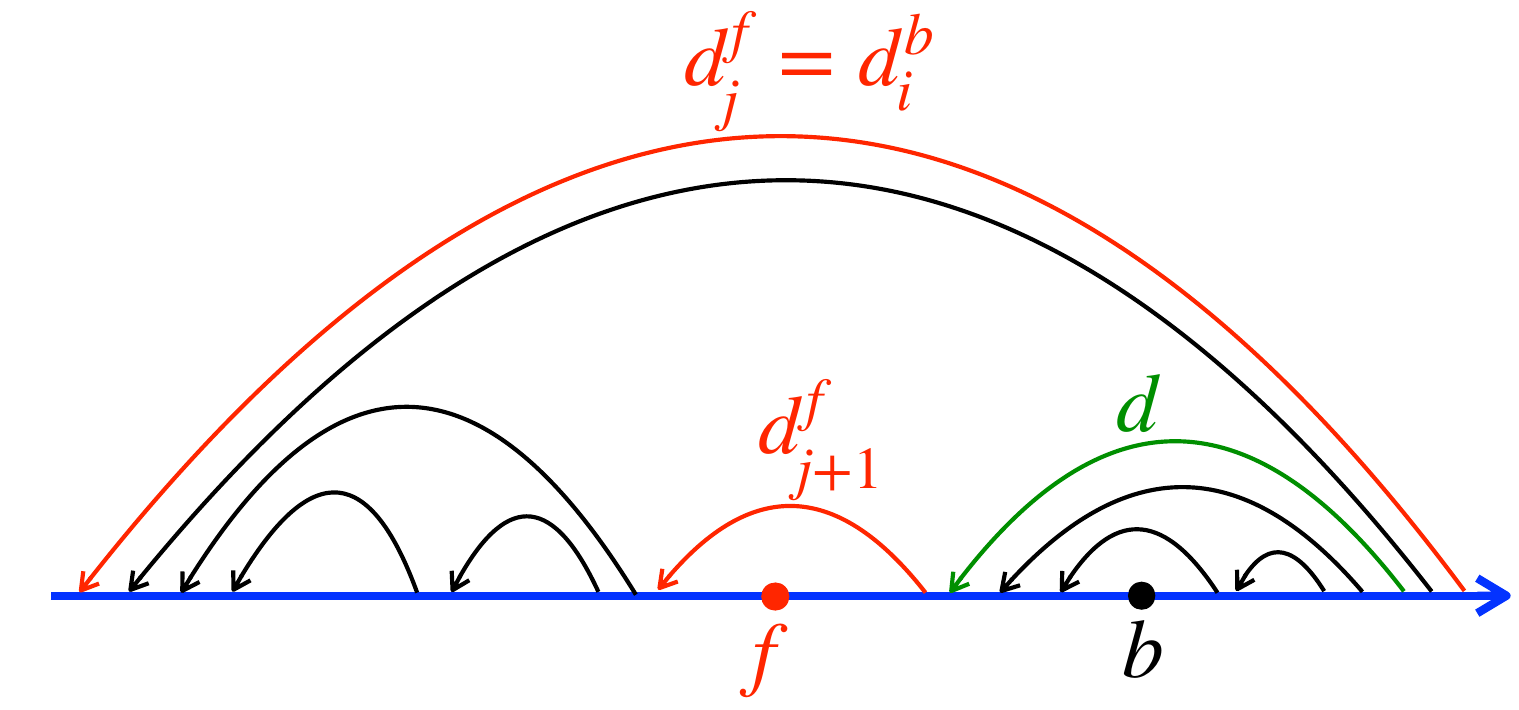}
  \caption{A nested system of $O(\log n)$ detours  $d^f_1,d^f_2,\dots$ of $f$ (in red) used in the auxiliary procedure. The (green) detour $d$ is the largest detour  that contains $b$ and does not contain $f$. Either $d=\hat{d}{^f_j}$  and is stored in $f$, or $d=d^b_{i+1}$ and is stored in $b$.
  \label{fig:nested}}
   \end{center}
\end{figure}

Consider the smallest detour $d^f_j=d^b_i$ that is saved by both $b$ and $f$. If such a detour does not exist then there is no detour that contains both $b$ and $f$ and so  $d= d^b_1$ is stored in $b$. If $|d|>|d^f_j|/2$ then $d$ must be the largest detour that is contained in $d^f_j$ but does not contain $d^f_{j+1}$. In other words, $d=\hat{d}{^f_j}$ and is stored in $f$. If on the other hand $|d|\le |d^f_j|/2$, then $|d|\le |d^b_i|/2$ and by the choice of $d^b_i$ we have that $d^b_{i+1}$ does not contain $f$ and so $d=d^b_{i+1}$ and is stored in $b$. This completes the description of the auxiliary procedure.

\subsubsection*{Description of the $\LabelPtoP$ labeling}
Equipped with the above auxiliary procedure, we are now ready to describe the secondary labeling scheme. Recall that there are four cases to consider: Given $b,f \in \P$ where $b$ can be before/after $f$ we wish to find the first vertex $p$ before/after $f$ that is reachable from $b$ in $\Gf$. There are four cases to consider:

\medskip
\noindent
{\bf Given \boldmath$f < b$, find the first vertex \boldmath$p<f$ that is reachable from \boldmath$b$ in \boldmath$\Gf$.} Consider the $b$-to-$p$ path $R$ in $\Gf$. Let $r$ be the first vertex of $R$ that belongs to $\P$ and $r<f$. Let $r'$ be the vertex of $\P$ that precedes $r$ on $R$. Note that $r<f$, that $r'>f$, and that the $r'$-to-$r$ subpath of $R$ (which we call a {\em bypass} of $f$ and denote $R[r',r]$) is internally disjoint from $\P$ and it either emanates to the left or to the right of $\P$. Moreover, since the endpoints of $\P$ lie on the same face, then if $R$ emanates at $r'$ to the left (resp. right) of $\P$ it must enter $\P$ at $r$ from the left (resp. right) of $\P$. This means that we can assume w.l.o.g that the bypass $R[r',r]$ is the largest such bypass (in terms of the number of vertices of $\P$ between $r$ and $r'$). This is because any smaller bypass that is on the same side of $\P$ as $R[r',r]$ is either contained in $R[r',r]$ (in which case we might as well use $R[r',r]$) or intersects $R[r',r]$ implying (in contradiction) that there is a larger bypass.  See Figure~\ref{fig:bypass}.

 \begin{figure}[htb]
  \begin{center}
 \includegraphics[scale=0.3]{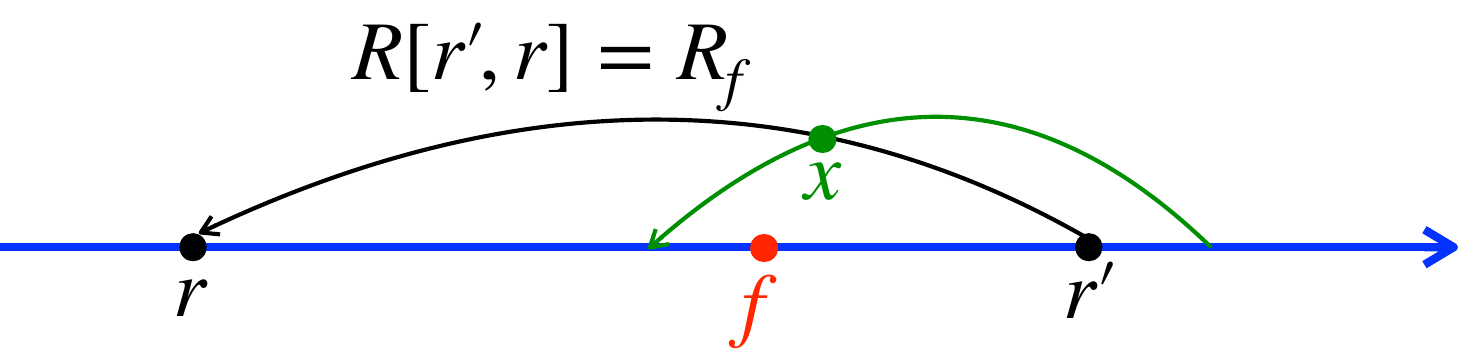}
  \caption{A bypass $R[r',r]$ (in black) of $f$.
  If there was an intersecting bypass (in green) that is larger than $R[r',r]$ then the green subpath before $x$ and the black subpath after $x$ would constitute a larger bypass.  \label{fig:bypass}}
   \end{center}
\end{figure}

 We let \hlgray{each vertex $f \in \P$ store the endpoints of the largest bypass $L_f$ (resp. $R_f$) that is internally disjoint from $\P$, emanates left (resp. right) of $\P$ at a vertex that is after $f$, and enters $\P$ from the left (resp. right) at a vertex that is before $f$. The vertex $f \in \P$ also stores the first vertex $L_f^+$ (resp. $R_f^+$) of $\P$ that is before $f$ and is reachable in $\Gf$ from the endpoint of $L_f$ (resp. $R_f$) that is after $f$.}

In order to find $p$, we first use the auxiliary procedure of Section~\ref{sec:auxiliary} to find the first vertex $p'>f$ that is reachable in $G$ from $b$ using a path that does not touch any vertex of $\P$ before $f$. Then, we check whether $p'$ is contained in $L_f$ and if so we consider $L_f^+$ as a candidate for $p$. Similarly, we check whether $p'$ is contained in $R_f$ and if so we consider $R_f^+$ as a candidate for $p$. Finally, we return the earlier of the (at most two) candidates.

 \medskip
\noindent
{\bf Given \boldmath$f < b$, find the first vertex \boldmath$p>f$ that is reachable from \boldmath$b$ in \boldmath$\Gf$.}
This case is very similar to the previous case. The only differences are: (1) we let \hlgray{every vertex $f \in \P$ store the first vertex $L_f^-$ (resp. $R_f^-$) of $\P$ that is after $f$ and is reachable in $\Gf$ from the endpoints of $L_f$ (resp. $R_f$)}, and (2) we add $p'$ itself as a third possible candidate for $p$.

 \medskip
\noindent
{\bf Given \boldmath$b < f$, find the first vertex \boldmath$p>f$ that is reachable from \boldmath$b$ in \boldmath$\Gf$.}
 This case and the next one are handled by small (but not symmetric) modifications of the previous two cases.
 Consider the $b$-to-$p$ path $R$ in $\Gf$. Let $r$ be the first vertex of $R$ that belongs to $\P$ and $r>f$. Let $r'$ be the vertex of $\P$ that precedes $r$ on $R$. The subpath $R[r',r]$ (which we call a {\em byway} of $f$) is internally disjoint from $\P$, and if it emanates at $r'$ to the left (resp. right) of $\P$ then it must enter $r$ from the left (resp. right) of $\P$. We can assume w.l.o.g that  $R[r',r]$ is the smallest such byway (because any larger byway that is on the same side of $\P$ as $R[r',r]$ either contains $R[r',r]$ (in which case we might as well use $R[r',r]$) or intersects $R[r',r]$ implying (in contradiction) that there is a smaller byway. See Figure~\ref{fig:byway}.

 We let \hlgray{each vertex $f \in \P$ store the endpoints of the smallest byway $L_f$ (resp. $R_f$) containing $f$ that is to the left (resp. right) of $\P$. The vertex $f \in \P$ also stores the first vertex $L_f^-$ (resp. $R_f^-$) of $\P$ that is after $f$ and is reachable in $\Gf$ from the endpoint of $L_f$ (resp. $R_f$) that is after $f$.}

In order to find $p$, we begin by finding the first vertex $p'<f$ that is reachable in $G$ from $b$ using a path that does not touch any vertex of $\P$ after $f$. This is done by using a variant of the auxiliary procedure of Section~\ref{sec:auxiliary} with the following modification. In $H_P$ we add the detour $(u,v)$ iff (1) there exists a $u$-to-$v$ path in $G$ that does not touch $\P$ before $v$, and (2) there is no such $u$-to-$w$ path for any $w<v$.
After finding $p'$ using this mechanism, we check whether $p'$ appears on $P$ before the starting point of the byway $L_f$, and if so we consider $L_f^-$ as a candidate for $p$. Similarly, we check whether $p'$ appears on $P$ before the starting point of the byway $R_f$ and if so we consider $R_f^-$ as a candidate for $p$. Finally, we return the earlier of the (at most two) candidates.

 \begin{figure}[htb]
  \begin{center}
 \includegraphics[scale=0.3]{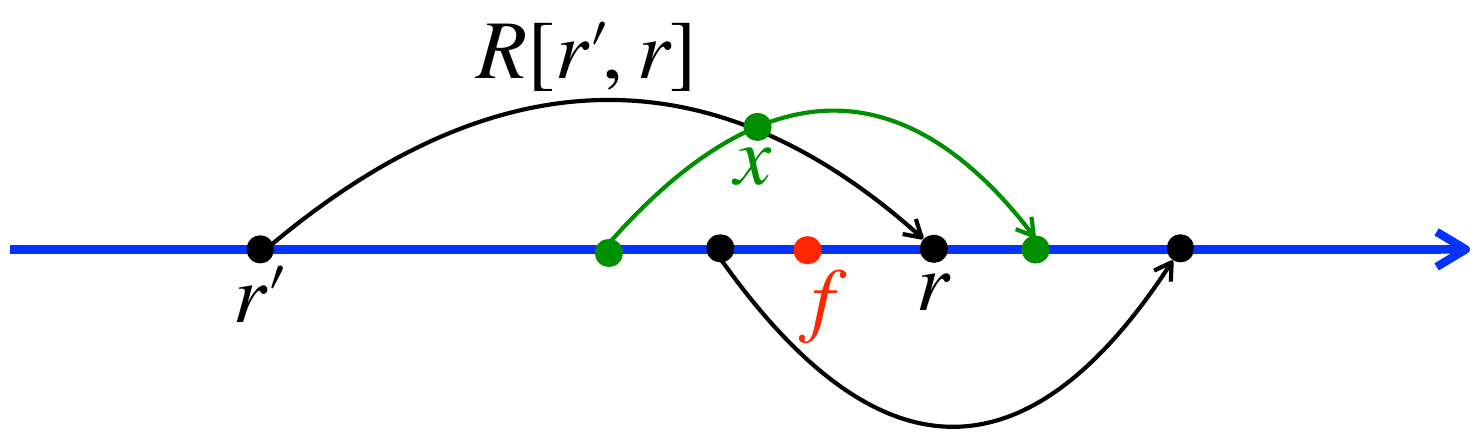}
  \caption{Two byways of $f$ (in black), the top one is $R[r',r]$. If there was an intersecting byway (in green) that is smaller than $R[r',r]$ then the green subpath before $x$ and the black subpath after $x$ would constitute a smaller byway. \label{fig:byway}}
   \end{center}
\end{figure}

 \medskip
\noindent
{\bf Given \boldmath$b < f$, find the first vertex \boldmath$p<f$ that is reachable from \boldmath$b$ in \boldmath$\Gf$.}
 This case is very similar to the previous case. The only differences are: (1) we \hlgray{let every vertex $f \in \P$ store the first vertex $L_f^+$ (resp. $R_f^+$) of $\P$ that is before $f$ and is reachable in $\Gf$ from the endpoints of $L_f$ (resp. $R_f$),} and (2) we add $p'$ itself as a third possible candidate for $p$.

\section{Approximate Distance Labeling}\label{sec:dist_labels}
In the remainder of this paper, we extend our labeling scheme from reachability to approximate distances.
We use standard notation for weighted graphs.
For a path $P$ we use $\len(P)$ to denote the total length of $P$, i.e. the sum of $P$'s edge weights.
We also use $\dist_H(x,y)$ to denote the $x$-to$y$ distance in the graph $H$.

We begin by describing the scaling approach of Thorup~\cite{Thorup04}.
A $(3,r)$-layered spanning tree $T$ in a digraph $H$ is an unoriented rooted spanning tree (i.e., a rooted spanning tree when ignoring the directions of edges) such that any path in $T$ from the root is the concatenation of at most $3$ shortest directed paths in $H$, each of length at most $r$.
We say that $H$ is a $(3,r)$-layered graph if it has such a spanning tree.
\begin{lemma}{\cite[Lemma 3.2 \footnote{The statement of item (3) in~\cite{Thorup04} only concerns shortest paths, but the proof there actually applies to any path of length at most $r$, as stated here.}]{Thorup04}} \label{lem:ThorupScale}
Given a directed graph $G$ and a scale $r$, we can construct in linear time a series of directed graphs $G_1^r, \dots , G_k^r$, such that:
\begin{enumerate}
    \item The total number of edges and vertices in all $G_i^r$'s is linear in the number of edges and vertices of $G$.
    \item Each vertex $v$ of $G$ has an index $i(v)$, such that $v$ does not belong to any $G_i^r$ other than $G_{i(v)}^r$ $G_{i(v)-1}^r$, and $G_{i(v)-2}^r$. Moreover, for any path $P$ of length at most $r$ that starts at $v$, $P$ is a path in $G$ if and only if $P$ is a path in at least one of the three $G_i^r$'s containing $v$.
    \item Each $G_i^r$ is a $(3,r)$-layered graph.
    \item $G_i$ is a minor of $G$. That is, $G_i$ is obtained from $G$ by deletion of edges and vertices, and contraction of edges.
\end{enumerate}
\end{lemma}
We invoke \cref{lem:ThorupScale} at exponentially growing scales $r$ (i.e., all powers of $2$ up to $2nM$ where $M$ is the largest edge-weight in the graph).
To each of the resulting graphs $G_i^r$, we will apply a labeling scheme for reporting the length of an $s$-to-$t$ path in $\Gf$ that is with additive error at most $\frac{\eps}{2} \cdot r$ with respect to the length of the shortest such path.
By item (2) of the Lemma, the distance reported for one of the three graphs containing $s$ at the scale $r$ satisfying $2\dist_{\Gf}(s,t) \geq r \geq \dist_{\Gf}(s,t)$, is a $(1+\eps)$ multiplicative approximation of the correct answer, and no distance reported for any graph containing $s$ will be smaller than the correct answer.

Hence, from now on we focus on describing the labeling scheme with additive error $\eps r$ for a $(3,r)$-layered directed planar graph.
The distance labeling scheme generalizes the reachability labeling scheme, but not in a black box manner.
We shall follow the overall structure of the reachability labels. In particular, the graph will be recursively decomposed using fundamental cycle separators, but now we use the the $(3,r)$-layered spanning tree $T$ guaranteed by item (3) in \cref{lem:ThorupScale}.
Each separator $Q$ still consists of a constant number of directed paths, but now this constant is larger (the spanning tree we use now is $3$-layered rather than $2$-layered for reachability), and now these paths are shortest paths, each of length at most $r$.
We further break each of these paths into $O(1/\eps)$ subpaths, each of length at most $\eps r$. Hence, from now on we treat every separator $Q$ as a set of $O(1/\eps)$ directed shortest paths of length $\eps r$ each.

To explain the high level idea of our approach,  let us describe how to generalize from reachability to approximate distances in the {\em non-faulty} case.
Assume we just want to approximate the length of a shortest path from $s$ to $t$ under the assumption that this shortest path intersects a particular separator path $P$.
In the case of reachability, we could deduce the answer by comparing $\first{G}{s}{P}$ and $\last{G}{P}{t}$.
To generalize to approximate distances, let $\dfirst{G}{s}{P}{\alpha}$ (resp., $\dlast{G}{t}{P}{\alpha}$) denote the first (resp., last) vertex of $P$ that is reachable from $s$ (resp., can reach $t$) in $G$ via a path of length at most $\alpha$. We store $\dfirst{G}{s}{P}{i \eps r}$ and $\dlast{G}{t}{P}{i\eps r}$ for every integer $0 \leq i \leq 1/\eps$.
To answer the query, we find $i,j$ minimizing $(i+j)$ such that $\dfirst{G}{s}{P}{i \eps r} \leq_P \dlast{G}{t}{P}{j\eps r}$, and return $(i+j+1)\eps r$ as the estimated distance.

This is correct since $\dfirst{G}{s}{P}{i \eps r} \leq_P \dlast{G}{t}{P}{j\eps r}$ implies there exists an $s$-to-$t$ path of length at most $(i+j+1)\eps r$ (the $+1$ term is for going along the path $P$ whose length is at most $\eps r$).
To see the approximation guarantee, consider the shortest $s$-to-$t$ path $\Gamma$ that intersects $P$. Let $\alpha$ be the length of the prefix of $\Gamma$ ending at the earliest vertex of $P$ visited by $\Gamma$. Let $\beta$ be the length of the maximal suffix of $\Gamma$ that is internally disjoint from $P$. Then, for $i=\lceil\frac{\alpha}{\eps r} \rceil$, and $j=\lceil\frac{\beta}{\eps r} \rceil$,  $\dfirst{G}{s}{P}{i \eps r} \leq_P \dlast{G}{t}{P}{j \eps r}$, so we will return at most $(i+j+1)\eps r \leq \alpha + \beta + 3\eps r$, while the length of $\Gamma$ is at least $\alpha+\beta$.

We note that Thorup's distance labels use the idea of {\em portals} (a.k.a $\eps$-covers) which results in smaller labels and more efficient query algorithm.
Our more brute force approach incurs polynomial factors in $1/\eps$, but allows us to use the ideas of the fault-tolerant reachability mechanism from the previous sections.

Note, however, that our scheme for fault-tolerant reachability requires in some cases to break the $s$-to-$t$ path into more than two segments (yet still a constant number of segments), and requires additional modifications as we explain in the rest of this section. In the following sections, some of the text is a straightforward adaptation of the corresponding text for reachability implementing the high-level ideas above. However, in several places (which we will highlight) significant changes and additional ideas are required.

\subsection{Fault tolerant approximate distance labeling}

In this section, we describe our approximate distances labeling scheme. I.e., what to store in the labels so that given the labels of any three vertices $s,f,t$ we can approximate $\dist_{\Gf}(s,t)$.
We follow the structure of the reachability labels from \cref{sec:reach_labels}, and adjust it to the more complicated task of approximate distances. Our main result is the following theorem.

\begin{theorem}\label{thm:dist}
    There exists a labeling scheme for a planar graph $G$ that, given vertices $s,t,f$ returns a $(1+\eps)$-approximation of $\dist_{\Gf}(s,t)$.
    The size of each label is $\Opoly$.
\end{theorem}

Before defining the analogous of \cref{lem:stopreach,lem:stfopreach} from \cref{sec:reach_labels}, we introduce the analog of $\first{F}{v}{P}$.
For a graph $F$, a vertex $v$, a path $P$, and a number $\alpha$, let $\dfirst{F}{v}{P}{\alpha}$ denote the first vertex on $P$ that is reachable from $v$ in $F$ via a path of length at most $\alpha$.
We also present the following relaxation of $\dfirst{F}{u}{P}{\alpha}$.

\begin{definition}[$\delta$-$\mathsf{first}$]\label{def:deltafirst}
Let $F$ be a directed  graph with edge weights from $\mathbb R^+$, and let $P$ be a path.
A vertex $v\in F$ is $\ddfirst{F}{u}{P}{\alpha}{\delta}$ if
\begin{enumerate}
    \item $v \leq_\P \dfirst{F}{u}{\P}{\alpha}$ and,
    \item $\dist_{F}(u,v)\le \alpha+\delta$.
\end{enumerate}
\end{definition}

We follow the reduction into two labeling schemes from a vertex to a path in
the presence of a failed vertex, one for the case that the failing vertex is not on the path, and the
other for the case that it is.

\begin{restatable}{lemma}{stofplabellem}\label{lem:stfoplabel}
    There exists a labeling scheme $\LabelstfoP=\LabelstfoP_{G,P,r,\alpha,\eps}$ where $G$ is a planar graph equipped with a decomposition tree $\mathcal T$,  $P$ is a path in $\mathcal T$ with $\len(P)=0$, and $r,\alpha,\eps\in \mathbb R^+$ such that $\alpha\le r$, and the length of every separator in $\mathcal T$ is bounded by $r$.
    Let $s$ and $f\notin P$ be two vertices of $G$ that are not an ancestor apex of one another, and such that $P$ is an ancestor path of both $s$ and $f$.
    Given the labels of $s$ and $f$, one can compute the index on $P$ of some vertex $b$ that is $\ddfirst{\Gf}{s}{P}{\alpha}{\eps r}$.
    In this labeling scheme, the only vertices that store a label are those that have $P$ as an ancestor. The size of each label stored by such a vertex is $\Opoly$.
\end{restatable}

\begin{restatable}{lemma}{stoplabellem}\label{lem:stoplabel}
    There exists a labeling scheme $\LabelstoP=\LabelstoP_{G,P,r,\alpha,\eps}$ where $G$ is a planar graph equipped with a decomposition tree $\mathcal T$, and $P$ is a path in $\mathcal T$ such that both endpoints of $P$ lie on the same face of $G$ and $\len(P)=0$ and $r,\alpha,\eps\in \mathbb R^+$ such that $\alpha\le r$, and the length of every separator in $\mathcal T$ is bounded by $r$.
    Let $s$ and $f\in P$ be two vertices of $G$ that are not an ancestor apex of one another, and such that $P$ is an ancestor path of both $s$ and $f$.
    Given the labels of $s$ and $f$, one can compute the indices on $P$ of some vertices $b_1$ and $b_2$ that are $\ddfirst{\Gf}{s}{P_1}{\alpha}{\eps r}$ and $\ddfirst{\Gf}{s}{P_2}{\alpha}{\eps r}$, respectively, where $P_1$ (resp. $P_2$) is the prefix (resp. suffix) of $P$ that precedes (resp. follows) $f$, excluding $f$.
    In this labeling scheme, the only vertices that store a label are those that have $P$ as an ancestor path. The size of each label stored by such a vertex is $\Opoly$.
\end{restatable}

We prove the above two lemmas in \cref{sec:stfoP,sec:stoPf}.
Given these lemmas, we now prove our main theorem.

\begin{proof}[Proof of \cref{thm:dist}]
Let $G_{rev}$ be the graph obtained from $G$ by reversing all edges ($G_{rev}$ has exactly the same decomposition tree as $G$, but the direction of each path is reversed).
Let $\eps' = \frac{\eps}{10}$ be an approximation parameter, and for every $r\in\mathbb R^+$ let $\Gamma_r=\{ i\eps' r \mid i\in [\ceil{\frac{1}{\eps'}}] \}$.

Let $v$ be a vertex.
For each $r$ that is a power of $2$ between $1$ and $2nM$, for every $i$ such that $v\in G^r_i$, and for every $\alpha,\beta\in\Gamma_r$ the label of $v$ stores the following:
\begin{enumerate}
\item \label{it:dist:thorup} \hlgray{For each piece $H$ in the recursive decomposition $\mathcal T$ of $G^r_i$ such that $v \in H \setminus \partial H$, $v$ stores its label in the standard (non-faulty) approximate distance labeling of Thorup for $H\setminus \partial H$.}

\item \label{it:dist:last} \hlgray{For every ancestor piece $H$ in $G^r_i$, $v$ stores $\dfirst{H}{v}{\P}{\beta}$ for each of the $\Opoly$ paths $\P$ of $\partial H$.}

\item \label{it:dist:stfop} Using \cref{lem:stfoplabel}, \hlgray{for every ancestor path $P$ of $v$ in the recursive decomposition $\mathcal T$, $v$ stores $\LabelstfoP_{G',P,r,\alpha,\eps'}(v)$ where $G'$ is obtained from $G^r_i$ by setting $\len(P)=0$}.

\item \label{it:dist:stop} Using \cref{lem:stoplabel}, \hlgray{for every ancestor piece $H$ of $v$ in $\mathcal T$, for every path $P$ of the separator $C$  of $H$, $v$ stores $\LabelstoP_{H^{\times P},P,r,\alpha,\eps'}(v)$},  where $H^{\times P}$ is the graph obtained from $H\setminus \partial H$ by making an incision along all the edges of $C$ other than those of $P$ and setting $\len(P)=0$.
Note that, because of the incision, the endpoints of $P$ lie on a single face of $H^{\times P}$, so \cref{lem:stoplabel} indeed applies.

\item \label{it:dist:apex} \hlgray{For every ancestor apex $a$ of $v$, for every ancestor path $\P$ of $v$, $v$ stores in its label $\dfirst{G^r_i\setminus \{a\}}{v}{\P}{\alpha}$, and  $\dfirst{G^r_i\setminus \{v\}}{a}{\P}{\alpha}$.}
If $a$ lies on $P$ then let $P_2$ be the suffix of $P$ after $a$ (not including $f$).
$v$ also stores \hlgray{$\dfirst{G^r_i\setminus \{a\}}{v}{P_2}{\alpha}$}.
Similarly, if $v$ lies on $P$ then let $P_2$ be the suffix of $P$ after $v$ (not including $v$).
$v$ also stores \hlgray{$\dfirst{G^r_i\setminus \{v\}}{a}{P_2}{\alpha}$}.

\item \hlgray{$v$ additionally stores all the above items in the graph $(G^r_i)_{rev}$ instead of the graph $G^r_i$.}
\end{enumerate}

\paragraph{Size.} Each vertex participates in $O(\log M n)$ graphs $G^r_i$, and $|\Gamma_r|=O(1/\eps)$.
Thus, for each $G^r_i$ such that $v\in G^r_i$ and for every $\alpha,\beta\in \Gamma_r$, $v$ has only $\Opoly$ ancestor pieces, paths and apices, and by \cref{lem:stfoplabel,lem:stoplabel}, all items above sum up to a label of size $\Opoly$.

\paragraph{Decoding and Correctness.}
Let $R$ be a shortest path from $s$ to $t$ in $\Gf$.
By the second property of \cref{lem:ThorupScale}, there exists a graph $G^r_i$ such that $R$ is contained in $G^r_i$.
If $f\notin G^r_i$ we query Thorup's non-faulty approximate distances labels for $G^r_i$ (stored in item (\ref{it:dist:thorup})).
This will output a $(1+\eps')$-approximation of $\dist_{G^r_i}(s,t)= \dist_{\Gf}(s,t)$.
Otherwise, let $\Htf$ be the rootmost piece in $\mathcal T$ whose separator $Q$ separates $t$ and $f$.
Let $H$ be a child piece of $\Htf$ that contains $t$ (if both children of $\Htf$ contain $t$ then, if one of the children does not contain $f$ we choose $H$ to be that child).
Note that by choice of $H$, $f \notin H \setminus \partial H$.

We assume without loss of generality that $s \in \Htf$. We handle the other case symmetrically to the description below, by swapping the roles of $s$ and $t$ and working in $(G^r_i)_{rev}$ instead of in $G^r_i$.

Observe that by definition of $H$ and of separation, $f \in \partial H$ iff $f \in \Q$.
Consider first the case when $R$ does not touch $\partial H$. i.e., $s,t$ and $R$ are all contained in $H\setminus \partial H$, and $f$ is not contained in $H \setminus \partial H$. In this case, querying Thorup's non-faulty labels for $H\setminus \partial H$ (stored in item (\ref{it:dist:thorup})) will output a $(1+\eps')$-approximation of $\dist_{H\setminus \partial H}(s,t)=\dist_{\Gf}(s,t)$.

To treat the case where the replacement path $R$ touches $\partial H$, we separately handle the cases where $f \notin \Q$ and $f \in Q$.

\paragraph{When $f \notin \Q$ (and so, $f \notin \partial H$).}
In this case, $R$ must have a suffix contained in $H$, and this suffix is unaffected by the fault $f$.
Let $a$ be the last vertex on $R$ which is on $Q$, and let $\alpha$ and $\beta$ be the smallest numbers in $\Gamma_r$ with $\alpha \ge \len(R[s,a])$ and $\beta\ge \len(R[a,t])$.
Let $P$ be the subpath of $Q$ that contains $a$.
If $R$ exists then $\dfirst{G^r_i\setminus\{f\}}{s}{\P}{\alpha} <_P \dfirst{H_{rev}}{t}{\P}{\beta}$.
The vertex $\dfirst{H_{rev}}{t}{\P}{\beta}$ is stored in item (\ref{it:dist:last}) of the label of $t$ in $(G^r_i)_{rev}$.
Notice that by the rootmost choice of $\hat H$, $H$ is an ancestor piece of $t$, so $t$ indeed stores $\dfirst{H_{rev}}{t}{\P}{\beta}$.
It thus remains only to describe how to find $\dfirst{G^r_i\setminus\{f\}}{s}{\P}{\alpha}$ from the labels of $s$ and $f$.
If either $s$ or $f$ store $\dfirst{G^r_i\setminus\{f\}}{s}{P}{\alpha}$ in item (\ref{it:dist:apex}), we are done.
Otherwise neither $s$ nor $f$ is an ancestor apex of the other, and since both $s$ and $f$ are in $\Htf$, $\P$ is indeed an ancestor path of both $s$ and $f$, so, by \cref{lem:stfoplabel}, a vertex $a'$ that is an $\ddfirst{G^r_i\setminus\{f\}}{s}{\P}{\alpha}{\eps'r}$ can be obtained from  $\LabelstfoP_{G',P,r,\alpha,\eps'}(s)$  (stored in item (\ref{it:dist:stfop}) of the label of $s$) and $\LabelstfoP_{G',P,r,\alpha,\eps'}(f)$  (stored in item (\ref{it:dist:stfop}) of the label of $f$).
The decoding algorithm will recognize that indeed $a' \le_P \dfirst{H_{rev}}{t}{\P}{\beta}$ and deduce that there exists a path of length at most $(\alpha+\eps'r)+\eps'r+\beta$ from $s$ to $t$ in $\Gf$.
Thus, the distance we output is
\begin{align*}
(\alpha+\eps'r)+\eps'r+\beta&\le (\alpha+\beta)+2\eps'r
\\&\le|R|+4\eps'r
\\&\le|R|+8\eps'|R|
\\&\le(1+\eps)|R|= (1+\eps)\dist_{\Gf}(s,t).
\end{align*}

\paragraph{When $f \in \Q$.}
Let $\P$ be the path of $\Q$ that contains $f$.
Consider first the case where the path $R$ touches some path $\PP\neq P$ of $\Q$ or of the boundary of some ancestor of $H$. Since boundary paths are vertex disjoint, $f\in \P$ implies $f \notin \PP$.
Hence, we can obtain a vertex $a'$ that is an $\ddfirst{G^r_i\setminus\{f\}}{s}{\PP}{\alpha}{\eps'r}$ in a similar manner to the case $f \notin Q$ above, with $\hat P$ taking the role of $P$.
In an analogous manner, we can obtain a vertex $b'$ that is a $\ddfirst{(G^r_i)_{rev}\setminus\{f\}}{t}{\PP}{\alpha}{\eps'r}$ as we just found $a'$, but with $\hat P$ taking the role of $P$, $t$ taking the role of $s$, and  $(G^r_i)_{rev}$ taking the role of $G^r_i$ (we cannot use part (\ref{it:dist:last}) of the label of $s$ or $t$ in this case because now $f$ does belong to $\partial H$).

Now consider the case where other than $P$, $R$ does not touch any path of $Q$ or any path of the boundary of an ancestor of $H$. In this case, it is valid to use labels in $\hat H^{\times P}$ instead of in $G^r_i$ because $R$ does touch $\partial \hat H$, and only crosses $Q$ at $P$.
Let $P_1$ and $P_2$ be the prefix and suffix obtained from $P$ by deleting $f$.
If either $s$ or $f$ is an ancestor apex of one another then
$\dfirst{G^r_i\setminus\{f\}}{s}{P_2}{\alpha}$ is stored in item (\ref{it:dist:apex}) of either $s$ or $t$, and, if $\dfirst{G^r_i\setminus\{f\}}{s}{P_1}{\alpha}$ exists, then it is equal to $\dfirst{G^r_i\setminus\{f\}}{s}{P}{\alpha}$, which is also stored in  item (\ref{it:dist:apex}) of either $s$ or $t$.
(If $\dfirst{G^r_i\setminus\{f\}}{s}{P}{\alpha}$ is not a vertex of $P_1$ then $\dfirst{G^r_i\setminus\{f\}}{s}{P_1}{\alpha}$ does not exist.)
If neither $s$ nor $t$ is an ancestor apex of the other, then $P$ is an ancestor path of both $s,t$ and $f$.
Let $a$ be the last vertex on $R$ which is on $P$, and let $\alpha$ and $\beta$ be the smallest numbers in $\Gamma_r$ where $\alpha \ge \len(R[s,a])$ and $\beta \ge \len(R[a,t])$.
We use the labels $\LabelstoP_{\hat H^{\times P},P,r,\alpha,\eps'}(s), \LabelstoP_{\hat H^{\times P},P,r,\alpha,\eps'}(f), \LabelstoP_{(\hat H^{\times P})_{rev},P,r,\beta,\eps'}(f)$ and $\LabelstoP_{(\hat H^{\times P})_{rev},P,r,\beta,\eps'}(t)$ to obtain the following vertices:
\begin{enumerate}
    \item $a_1$ which is an $\ddfirst{\hat H^{\times P}}{s}{P_1}{\alpha}{\eps'r}$
    \item $a_2$ which is an $\ddfirst{\hat H^{\times P}}{s}{P_2}{\alpha}{\eps'r}$,

    \item $b_1$ which is an $\ddfirst{(\hat H^{\times P})_{rev}}{s}{P_1}{\beta}{\eps'r}$
    \item $b_2$ which is an $\ddfirst{(\hat H^{\times P})_{rev}}{s}{P_2}{\beta}{\eps'r}$
\end{enumerate}
The decoding algorithm will recognize that indeed we have  $a_1 \le_P b_1$ (if $a\in P_1$) or $a_2\le_P b_2$  (if $a\in P_2$) and deduce that there exists a path of length at most $(\alpha+\eps'r)+\eps'r+(\beta+\eps'r)$ from $s$ to $t$ in $\Gf$.
Thus, the distance we output is
\begin{align*}
(\alpha+\eps'r)+\eps'r+(\beta+\eps'r)&\le (\alpha+\beta)+3\eps'r
\\&\le|R|+5\eps'r
\\&\le|R|+10\eps'|R|
\\&\le(1+\eps)|R|= (1+\eps)\dist_{\Gf}(s,t).
\end{align*}
Notice that the decoding algorithm does not know a-priori the correct values of $G^r_i$, $\alpha$ and $\beta$.
Thus, the algorithm iterates over all possible values and returns the minimum distance found.
Clearly, every distance found for some values implies a path in $\Gf$, thus the algorithm will always return at least $\dist_{\Gf}(s,t)$, and as we proved it will always be at most $(1+\eps)\dist_{\Gf}(s,t)$.
\end{proof}

\subsection{The $\LabelstfoP$ labeling}\label{sec:stfoP}

In this section we provide a labeling scheme $\LabelstfoP$, proving \cref{lem:stfoplabel}.
Our labeling scheme makes use of the following two auxiliary labeling schemes:

\begin{restatable}{lemma}{stoalabellem}\label{lem:stoalabel}
    There exists a (trivial) labeling scheme    $\Labelstoa=\Labelstoa_{H, P,\alpha}$  where $H$ is a planar graph and $P$ is a $0$-length path.
    Given the label of a vertex $a$, one can retrieve the vertex $\dfirst{H}{a}{P}{\alpha}$.
    The size of each label is $\Opoly$.
\end{restatable}

\begin{restatable}{lemma}{atfoplabellem}\label{lem:atfoplabel}
    There exists a labeling scheme $\LabelatfoP=\LabelatfoP_{H,P,P',\alpha,\eps}$ where $H$ is a planar graph, $P$ and $P'$ are two $0$-length paths and $\alpha,\eps\in\mathbb R^+$.
    Given the labels of a vertex $a$ on $P'$ and a vertex $f$ not in $P'\cup P$, one can retrieve the index on $P$ of some vertex $b\in P$ such that $b$ is an $\ddfirst{H\setminus f}{a}{P}{\alpha}{\eps\alpha}$.
    The size of each label is $\Opoly$.
\end{restatable}

The proof of \cref{lem:stoalabel} is trivial and the proof of \cref{lem:atfoplabel} appears in \cref{sec:atfop}.

Conceptually, \cref{lem:stfoplabel} is obtained by composing \cref{lem:atfoplabel,lem:stoalabel}, as we describe in the following overview.
However, the concept of composing two label schemes introduces several technical details, making the proof presented below appear more intricate.

\begin{figure}[h]
\begin{center}
\includegraphics[width=0.3\textwidth]{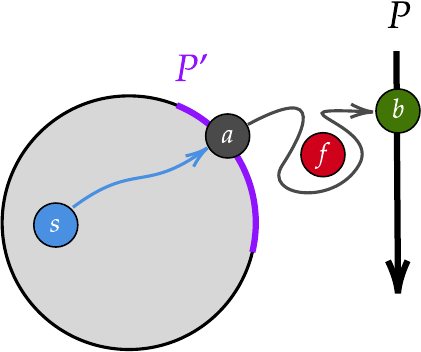}
\caption{An lustration of a path from $s$ to $b=\dfirst{\Gf}{s}{P}{\alpha}$.
The blue subpath is from $s$ to $a\in P'$, and the gray subpath is from $a$ to $b \in P$.
We use $\Labelstoa$ to find a good middle point for $a$ on $P'$, and $\LabelatfoP$ to find a good middle point for $b$ on $P$.
\label{fig:vtfoP}
}
\end{center}
\end{figure}

\paragraph{Overview (see \cref{fig:vtfoP}).}
Let $R$ be a shortest path from $s$ to $b=\dfirst{\Gf}{s}{P}{\alpha}$ in $\Gf$, and let $Q$ be a separator that separates $s$ and $f$.
We are interested in a special vertex on $R$: the first vertex $a^*$ on $R$ that is in $Q$, and specifically on some $P'$ which is an $\eps r$-subpath of $Q$.
We 'guess' estimations $\beta,\gamma$ for the lengths of $R[s,a^*]$ and $R[a^*,b]$, respectively.
In this overview, it is useful to consider a vertex $m$ that is $\ddfirst{\Gf}{s}{P}{\beta}{\eps r}$ as a sufficiently good 'middle point.
Specifically, if one aims to reach an early vertex on $P$ from $s$ within a budget $\alpha$, requiring the path to pass through $m$ and allowing the budget to exceed $\alpha$ by $\eps r$ does not worsen the result on $P$.
Using the labels $\Labelstoa$ with budget $\beta$, we find a sufficiently good 'middle point' $a$ for replacing $a^*$.
Starting from $a$, we use the labels $\LabelatfoP$ with budget $\gamma$ to find a vertex earlier on $P$ than $\dfirst{\Gf}{a}{P}{\gamma}\le_P \dfirst{\Gf}{a^*}{P}{\gamma}$.

Since each middle point is reachable with a budget increase of $O(\eps r)$, the final point is reachable with a budget increase of $O(\eps r)$.
Since every middle point is 'reasonably good', and the budget used in each phase is an estimation of the length of the corresponding  subpath of $R$, it is guaranteed that the final destination is at least as good as $b$.

We are now ready to provide the formal proof of \cref{lem:stfoplabel}.

\stofplabellem*
\begin{proof}
Let $\eps' = \frac{\eps}{4} \in \Theta (\eps)$ be an approximation parameter and let $\Gamma=\{ i\eps' \alpha \mid i\in [\ceil{\frac{1}{\eps'}}] \}$.
For a subpath $P'$ of a separator in the fully recursive decomposition of $G$, we denote as $G^0_{P'}$ the graph obtained from $G$ by setting the weight of every edge in $P'$ to be $0$.
For every $\gamma\in\Gamma$ let $G^1_{P',\gamma}$ be the graph $G$ where the label of each vertex $v$ is set to be $\LabelatfoP_{G^0_{P'},P,P',\gamma,\eps'}(v)$.
We also define $G^2_{P',\gamma}$ as the subgraph of $G^1_{P',\gamma}$ induced only on vertices below the separator of $P'$ in the fully recursive decomposition of $G$.

\paragraph{The Labeling.}
For every vertex $v$ such that $P$ is an ancestor of $v$ in $\mathcal T$, the label of $v$ stores
\hlgray{
for every triplet $(P',\beta,\gamma)$ such that $P'$ is an $\eps' r$-subpath ancestor of $v$  and $\beta,\gamma\in \Gamma$, the label $\Labelstoa_{G^2_{P',\gamma},P',\beta}(v)$ and the label of $v$ in $G^1_{P',\gamma}$.
Moreover, the label of $v$ stores the identifiers of the ancestor pieces of $v$ in the full recursive decomposition of $G$.
Finally, for every ancestor separator $Q$ of $v$ below $P$, the label of $v$ stores $b^\beta_Q= \dfirst{G'_{Q}}{v}{P}{\beta}$ where $G'_{Q}$ is the subgraph below the separator of $Q$, and the label $\LabelatfoP_{G,P,P,\gamma,\eps'}(b^\beta_Q)$.
\footnote{Notice that $\dfirst{G'_{Q}}{v}{P}{\alpha}$ is a slight abuse of notation, since $P$ is not contained in $G'_{Q}$. Note that the original definition of this notation can be applied similarly even if $P$ is only partially contained in the graph. In this context, $\dfirst{G'_{Q}}{v}{P}{\alpha}$ is the first vertex of $P$ obtainable from $v$ via a path in $G'_{Q}$ with length at most $\alpha$}
}

\paragraph{Size.}
There are $O(\log n)$ ancestor separators of $v$ in the fully recursive decomposition.
Each separator is partitioned into $O(\frac{1}{\eps})$ subpaths, so overall there are $\Opoly$ options for $P'$ for each vertex $v$.
It follows from $|\Gamma|= O(\frac{1}{\eps})$ that there are $\Opoly$ triplets $(P',\beta,\gamma)$.
For every such triplet, the label stores an $\Labelstoa$-type label in $G^2_{P',\gamma}$ and the label of a vertex in $G^2_{P',\gamma}$.
Each vertex in $G^2_{P',\gamma}$ is labeled with a $\LabelatfoP$-type label, which is of size $\Opoly$ due to \cref{lem:atfoplabel}. 
Due to \cref{lem:stoalabel}, the size of $\Labelstoa$ is $\Opoly$, which should be multiplied by another $\Opoly$ factor due to vertices of $G_{P',\gamma}$ each having labels of size $\Opoly$.
Finally, the height of the recursive decomposition is logarithmic, so storing the identifiers of ancestor pieces of $v$ and storing $b^\beta_Q$ and a $\LabelatfoP$-type label of $b^\beta_Q$ per ancestor separator $Q$ of $v$ increases the size of the label by $\Opoly$.

\paragraph{Decoding.}
Given the labels of two vertices $s$ and $f$ such that $P$ is an ancestor of $s$ and $f$, and none of $s$ and $f$ is an apex ancestor of the other in $\mathcal{T}$.
The algorithm finds an $\ddfirst{\Gf}{s}{P}{\alpha}{\eps r}$ vertex as follows.
The algorithm starts by finding the highest separator $Q$ in the recursive decomposition that has $s$ in one side, and $f$ strictly in the other side.
The label of $s$ stores $b^\beta_{Q} = \dfirst{G'_{Q}}{s}{P}{\alpha}$ and $\LabelatfoP_{G,P,P,\gamma,\eps}(b^\beta_Q)$ for every $\beta, \gamma \in \Gamma$. 
The label of $f$ stores $\LabelatfoP_{G,P,P,\gamma,\eps}(f)$.
We use the two labels to obtain $b^{\beta,\gamma}_Q$ that is a $\ddfirst{\Gf}{b^\beta_Q}{P}{\gamma}{\eps' \gamma}$.

Recall that $Q$ is partitioned into $O(\frac{1}{\eps})$ subpaths with length at most $\eps' r$ each.
For each such subpath $P'$, and for every $\beta,\gamma \in \Gamma$ such that $\beta+\gamma \le (1+2\eps')\alpha$, the algorithm uses $\Labelstoa_{G^2_{P',\gamma},P',\beta}(s)$ to find $a_{P',\beta}= \dfirst{G'_Q}{s}{P'}{\beta}$ where $G'_Q$ is the subgraph induced by vertices below $Q$ in the fully recursive decomposition.
More precisely, the label $\Labelstoa_{G^2_{P',\gamma},P',\beta}(s)$ allows the algorithm to retrieve $\ell_a = \LabelatfoP_{G^0_{P'},P,P',\gamma,\eps'}(a_{P',\beta})$.
Note that, the label of $f$ contains $\ell_f = \LabelatfoP_{G^0_{P'},P,P',\gamma,\eps'}(f)$.
The algorithm uses the labels $\ell_a$ and $\ell_f$ to retrieve a vertex $b_{P',\beta,\gamma}$ on $P$ that is an $\ddfirst{G^0_{P'}\setminus \{f\} }{a_{P',\beta}}{P}{\gamma}{\eps' \gamma}$.
Finally, the algorithm returns the vertex $b = \min_{\le P}\{ b_{P',\beta,\gamma},b^{\beta,\gamma}_Q \mid \beta,\gamma \in \Gamma , P' \textit{ is a subpath of }Q, \beta + \gamma \le (1+2 \eps')\alpha \}$.
That is, the first vertex on $P$ that was found over all choices of $P'$, $\beta$ and $\gamma$, or the first $b^{\beta,\gamma}_{Q}$ found on $P$.

\paragraph{Correctness.}
We show that when decoding the labels of two vertices $s$ and $f$, we indeed return a vertex $b$ that is an $\ddfirst{\Gf}{s}{P}{\alpha}{\eps r}$.
We start by showing that for every candidate $x$ found by the decoding algorithm for $b$, we have $\dist_{\Gf}(s,x) \le \alpha + \eps r$.
Let $Q$ be the lowest separator in the decomposition such that $s$ is in one side of $Q$ and $f$ is strictly in the other side.
Considered a vertex $b^{\beta,\gamma}_Q$ that is a candidate for being $b$ returned by the decoding algorithm.
Recall that $b^{\beta,\gamma}_Q$ is an $\ddfirst{\Gf}{b^\beta_Q}{P}{\gamma}{\eps' \gamma}$ with $b^\beta_Q=\dfirst{G'_Q}{s}{P}{\beta}$.
By definition, we have $\dist_{G'_{Q}}(s,b^\beta_{Q}) \le \beta$ and since $f\notin G'_Q$ we have $\dist_{\Gf}(s,b^\beta_Q) \le \beta$.
By definition, $\dist_{\Gf}(b^\beta_Q ,b^{\beta,\gamma}_Q) \le (1+\eps') \gamma$.
By triangle inequality we have $\dist_{\Gf}(s,b^{\beta,\gamma}_Q) \le \beta + \gamma + \eps' \gamma \le \alpha + 2\eps' \alpha + \eps' \alpha \le (1+\eps) \alpha$.

We now deal with candidates obtained as $b_{P',\beta,\gamma}$.
Recall that $b_{P',\beta,\gamma}$ is an $\ddfirst{G^0_{P'} \setminus \{f\}}{a_{P',\beta}}{P}{\gamma}{\eps' \gamma}$ with $a_{P',\beta} = \dfirst{G'_Q}{s}{P'}{\beta}$.
Since $f \notin G'_Q$ we have $\dist_{\Gf}(s,a_{P',\beta}) \le \dist_{G'_Q}(s,a_{P',\beta}) \le \beta$.
Since $b_{P',\beta,\gamma}$ is $\ddfirst{G^0_{P'}\setminus \{f\}}{a_{P',\beta}}{P}{\gamma^*}{\eps' \gamma}$, we have $\dist_{G^0_{P'}\setminus\{f\}}(a_{P',\beta}, b_{P',\beta,\gamma}) \le \gamma + \eps' \gamma$.
Since $G^0_{P'}$ is obtained from $G$ by reducing the weight of all edges of $P'$ to $0$, and since the total length of $P'$ in $G$ is at most $\eps' r$, we have that $\dist_{\Gf}(a_{P',\beta}, b_{P',\beta,\gamma}) \le \gamma + \eps' \gamma + \eps' r$.
From triangle inequality we get $\dist_{\Gf}(s,b_{P',\beta,\gamma}) \le \beta + \gamma + \eps' \gamma + \eps' r \le \alpha + \eps' (2\alpha + \gamma + r) \le \alpha + 4\eps' r = \alpha + \eps r$ as required.

We are now left with the task of proving $b \le_P \dfirst{\Gf}{s}{P}{\alpha}$.
Let $b^*= \dfirst{\Gf}{s}{P}{\alpha}$ and let $R$ be a shortest path from $s$ to $b^*$ in $\Gf$.
Let $a^*$ be the first vertex on $R$ that is also in $Q$.
If there is no $a^*$, let $b'$ be the first vertex on $R$ that is on $P$ and let $\beta^* = \min \{x\in \Gamma \mid x \ge \len(R[s,b'] \}$ and $\gamma^* = \min \{x\in \Gamma \mid x \ge \len(R[b',b^*] \}$.
Note that $\beta^* + \gamma^* \le (1+2\eps') \alpha$ and therefore the decoding algorithm computed some vertex $b^{\beta^*,\gamma^*}_Q$ as a candidate for $b$.
Recall that $b^{\beta^*,\gamma^*}_Q$ is an $\ddfirst{\Gf}{b^{\beta^*}_Q}{P}{\gamma}{\eps' \gamma}$ vertex for $b^{\beta^*}_Q = \dfirst{G'_Q}{s}{P}{\beta^*}$.
Note that $R[s,b']$ is a path with length at most $\beta^*$ that is completely contained in $G'_{Q}$.
Therefore, $b^{\beta^*}_{Q} =\dfirst{G'_{Q}}{s}{P}{\beta^*} \le_P b'$.
Since $P[b^{\beta^*}_Q,b']\cdot R[b',b^*]$ is a path of length at most $\gamma^*$ in $\Gf$, we also have $b^{\beta^*,\gamma^*}_Q \le_P \dfirst{\Gf}{b'}{P}{\gamma^*} \le_P b^*$.
This concludes the case where $a^*$ does not exist, since we have $b \le_P b^*$ due to $b$ being the $\le_P$ minimum in a set containing $b^{\beta^*,\gamma^*}_Q$.

We now treat the case where $a^*$ exists.
Let $P'$ be the subpath of $Q$ that contains $a^*$.
Let $\beta^* = \min\{x \in \Gamma \mid x \ge \len(R[s,a^*])\}$ and $\gamma^* = \min\{x \in \Gamma \mid x \ge \len(R[a^*,b^*])\}$, and notice that $\beta^* + \gamma^* \le \len(R) + 2\eps' \alpha \le (1+ 2\eps') \alpha$.
Therefore, the decoding algorithm computed some $b_{P',\beta^*,\gamma^*}$ that is $\ddfirst{G^0_{P'}\setminus \{f\} }{a_{P',\beta^*}}{P}{\gamma^*}{\eps' \gamma^*}$ with $a_{P',\beta^*} = \dfirst{G'_Q}{s}{P'}{\beta^*}$.

Since $R[s,a^*]$ is a path from $s$ to $P'$ in $G'_Q$ of length at most $\beta^*$, we have $a_{P',\beta^*} \le_{P'} a^*$.
Since $P'[a_{P',\beta^*},a^*] \cdot R[a^*,b^*]$ is a path from $a_{P',\beta^*}$ to $b^*$ in $G^0_{P'}\setminus \{f\}$ with length at most $\gamma^*$ (recall that the length of $P'$ in $G^0_{P'}$ is zero), we have $b_{P',\beta^*,\gamma^*} \le_P b^*$.
Due to the minimality of the returned vertex $b$ on $P$ across all choices of $P'$, $\beta$ and $\gamma$, we have that $b \le_P b_{P',\beta^*,\gamma^*} \le_P b^*$ as required.
\end{proof}

\subsection{The $\LabelstoP$ labeling}\label{sec:stoPf}

In this section we provide a labeling scheme $\LabelstoP$, proving \cref{lem:stoplabel}.
Our labeling scheme makes use of \cref{lem:stoalabel} and the following two auxiliary labeling schemes:

\begin{restatable}{lemma}{ptoplabellem}\label{lem:ptoplabel}
    There exists a labeling scheme $\LabelPtoP=\LabelPtoP_{G,P,\alpha,\eps}$ where $G$ is a planar graph, $P$ is a $0$-length path of $G$ whose two endpoints lie on the same face and $\alpha,\eps\in \mathbb R^+$.
    Given the labels of two vertices $b$ and $f$ on $P$ let $P_1$ and $P_2$ be the prefix and suffix of $P$ before and after $f$ (without $f$), respectively.
    One can compute the indices on $P$ of some vertices $b_1$ and $b_2$ that are $\ddfirst{\Gf}{b}{P_1}{\alpha}{\eps\alpha}$ and $\ddfirst{\Gf}{b}{P_2}{\alpha}{\eps\alpha}$, respectively.
    The size of each label is $\Opoly$.
\end{restatable}

\begin{restatable}{lemma}{atopflabellem}\label{lem:atopflabel}
    There exists a labeling scheme
    $\LabelatoP=\LabelatoP_{H,P,P',\alpha,\eps}$ where $H$ is a planar graphs $H$ with a $0$-length path $P'$ and a path $P$ without outgoing edges which lies on a single face, such that $P\cap P'=\emptyset$, and $\alpha,\eps\in\mathbb R^+$.
    For $f\in P$ let $P_1$ and $P_2$ be the prefix and suffix of $P$ before and after $f$ (without $f$), respectively.
    Given the labels of two vertices $a\in P'$ and $f\in P$, one can retrieve two vertices $b_1$ and $b_2$ which are $\ddfirst{H\setminus\{f\}}{a}{P_1}{\alpha}{\eps\alpha}$ and $\ddfirst{H\setminus\{f\}}{a}{P_2}{\alpha}{\eps\alpha}$, respectively.
    The size of each label is $\Opoly$.
\end{restatable}
The proof of \cref{lem:ptoplabel} appears in \cref{sec:ptop} and the proof of \cref{lem:atopflabel} appears in \cref{sec:atfop}.
We note that given \cref{lem:stoalabel,lem:atopflabel,lem:ptoplabel} the label of $\LabelstoP$ is conceptually simple as we explain in the following high-level overview.
Due to certain technical details, the complete proof presented below appears more intricate.

\begin{figure}[h]
\begin{center}
\includegraphics[width=0.3\textwidth]{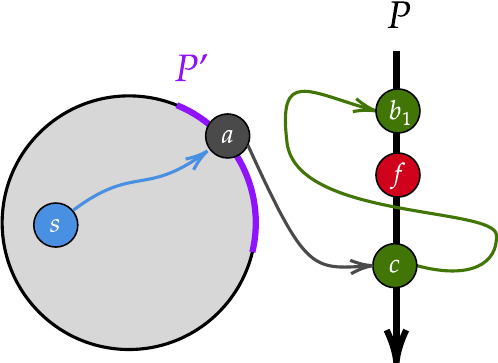}
\caption{An lustration of a path from $s$ to $b_1=\dfirst{\Gf}{s}{P_1}{\alpha}$.
The blue subpath is from $s$ to $a\in P'$, and the gray subpath is from $a$ to $c \in P_2$ and is internally disjoint from $P$, and the green subpath is from $c$ to $b_1\in P_1$.
We use $\Labelstoa$ to find a good middle point for $a$ on $P'$, and $\LabelatoP$ to find a good middle point for $c$ on $P_2$ and $\LabelPtoP$ to find a good middle point for $b_1$ on $P_1$.
\label{fig:vtoPf}
}
\end{center}
\end{figure}

\paragraph{Overview (see \cref{fig:vtoPf}).}
Let $R$ be a shortest path from $s$ to $b_1=\dfirst{\Gf}{s}{P_1}{\alpha}$ in $\Gf$, and let $Q$ be a separator that separates $s$ and $f$.
We are interested in two special vertices on $R$: the first vertex $a^*$ on $R$ that is in $Q$, and specifically on some $P'$ which is an $\eps r$-subpath of $Q$, and the first vertex $c^*$ on $R$ after $a^*$ that is on $P$ (say, on $P_2$).
We 'guess' estimations $\beta,\gamma,\delta$ for the lengths of $R[s,a^*]$, $R[a^*,c^*]$ and $R[c^*,b_1]$, respectively.
In this overview, it is useful to consider a vertex $m$ that is $\ddfirst{\Gf}{s}{P}{\beta}{\eps r}$ as a sufficiently good 'middle point.
Specifically, if one aims to reach an early vertex on $P$ from $s$ within a budget $\alpha$, requiring the path to pass through $m$ and allowing the budget to exceed $\alpha$ by $\eps r$ does not worsen the result on $P$.
Using the labels $\Labelstoa$ with budget $\beta$, we find a sufficiently good 'middle point' $a$ for replacing $a^*$.
Starting from $a$, we use the labels $\LabelatoP$ with budget $\gamma$ to find a sufficiently good 'middle point' replacing $c^*$.
Finally, we use $\LabelPtoP$ with budget $\delta$ to find a vertex earlier on $P$ than $\dfirst{\Gf}{c^*}{P_1}{\gamma}$.

Since each middle point is reachable with a budget increase of $O(\eps r)$, the final point is reachable with a budget increase of $O(\eps r)$.
Since every middle point is 'reasonably good', and the budget used in each phase is an estimation of the length of the corresponding  subpath of $R$, it is guaranteed that the final destination is at least as good as $b_1$.

We are now ready to provide the formal proof of \cref{lem:stoplabel}.

\stoplabellem*
\begin{proof}
Let $\eps' = \frac{\eps}{7} \in \Theta (\eps)$  be an approximation parameter and let $\Gamma=\{ i\eps' \alpha \mid i\in [\ceil{\frac{1}{\eps'}}] \}$.
For a subpath $P'$ of a separator in $\mathcal T$ below $P$, and for $\beta,\gamma, \delta\in\Gamma$ we define the following graphs.
\begin{enumerate}
    \item $G^P_{P',\delta}$ is the graph $G$ where the weights of all edges of $P'$ are set to $0$, and the label of each vertex $v$ is set to be $\LabelPtoP_{G,P,\delta,\eps'}(v)$ obtained by \cref{lem:ptoplabel}.
    \item $G^0_{P',\delta}$ is the graph obtained from $G^P_{P',\delta}$ by removing all outgoing edges of $P$ (the edges of $P$ itself are not removed).
    \item $G^1_{P',\delta,\gamma}$ is the graph obtained by setting the label of each vertex $v$ in $G^0_{P',\delta}$ to be $\LabelatoP_{G^0_{P',\delta},P,P',\gamma,\eps'}(v)$.
    \item $G^2_{P',\delta,\gamma}$ is the induced graph of $G^1_{P',\delta,\gamma}$ only on vertices that are below the separator of $P'$ in the fully recursive decomposition.
\end{enumerate}

\paragraph{The Labeling.}
For every vertex $v$ below the separator of $P$ in the fully recursive decomposition the label of $v$ stores
\hlgray{
the index of $v$ in $P$, if $v\in P$.
For every tuple $(P',\beta,\gamma,\delta)$ such that $P'$ is an $\eps' r$-subpath ancestor of $v$ and $\beta,\gamma,\delta\in \Gamma$, the label $\Labelstoa_{G^2_{P',\delta,\gamma},P',\beta}(v)$ and the labels of $v$ in $G^2_{P',\delta,\gamma}$ and in $G^P_{P',\delta}$.
Moreover, the label of $v$ stores the identifiers of the ancestor pieces of $v$ in $\mathcal{T}$.
Finally, for every maximal consecutive subpath $P^*$ of $P$ in $G'_{P'}$, and $\gamma,\delta \in \Gamma$ the label of $v$ stores $b^\gamma_{Q,P^*} = \dfirst{G'_{P'}}{v}{P^*}{\gamma}$ where  $Q$ is the separator from which $P'$ originates and $G'_{P'}$ is the subgraph below $Q$, and the label $\LabelPtoP_{G,P,\gamma,\eps'}(b^\gamma_{Q,P^*})$.
}

\paragraph{Size.}
There are $O(\log n)$ ancestor separators of $v$ in $\mathcal T$.
Each separator is partitioned into $O(\frac{1}{\eps})$ subpaths, so overall there are $\Opoly$ options for $P'$ for each vertex $v$.
It follows from $|\Gamma|= O(\frac{1}{\eps})$ that there are $\Opoly$ tuples $(P',\beta,\gamma,\delta)$.
For every such tuple, the label stores an $\Labelstoa$-type label in $G^2_{P',\delta,\gamma}$ and the label of $v$ in $G^2_{P',\delta,\gamma}$ and in $G^P_\delta$.
Due to \cref{lem:stoalabel}, the size of $\Labelstoa$ is $\Opoly$. In our construction, the size should be multiplied by another two $\Opoly$ factors since the $\Labelstoa$ label is applied to the graph $G^1_{P',\delta,\gamma}$ in which every vertex is labeled with a $\LabelatoP$ type label (which is of size $\Opoly$ due to \cref{lem:atopflabel}), and these $\LabelatoP$ labels are applied to the graph $G^P_{P',\delta}$ in which the label of each vertex is a $\LabelPtoP$-type label (which is of size $\Opoly$ due to \cref{lem:ptoplabel}).
Finally, the height of the recursive decomposition is logarithmic, and in every ancestor piece of $v$ below $P$, the number of maximal consecutive subpaths of $P$ is bounded by the number of apices ancestors of $v$ which is $\tilde O(1)$.
It follows that storing for every ancestor piece of $v$, for every maximal consecutive subpath $P^*$ of $P$ in the piece and for every $\gamma,\delta \in \Gamma$ the vertex $b^\gamma_{Q,P^*}$ and an $\LabelPtoP$-type label of $b^\gamma_{Q,P^*}$ increases the size of the label by $\Opoly$.

\paragraph{Decoding.}
Given the labels of two vertices $s$ and $f\in P$ that are not an ancestor apex of one another, and such that $P$ is an ancestor path of both $s$ and $f$, the algorithm finds a vertex $b_1$ that is $\ddfirst{\Gf}{s}{P_1}{\alpha}{\eps r}$ and a vertex $b_2$ that is $\ddfirst{\Gf}{s}{P_2}{\alpha}{\eps r}$ as follows.
The algorithm starts by finding the highest separator $Q$ in the recursive decomposition that has $s$ in one side, and $f$ strictly in the other side (this can be done since both $s$ and $f$ store all ancestor pieces).
Recall that $Q$ is partitioned into $O(\frac{1}{\eps})$ subpaths of length at most $\eps' r$ each.
For each such subpath $P'$, and for every $\beta,\gamma,\delta \in \Gamma$ such that $\beta+\gamma+\delta \le (1+3\eps')\alpha$, the algorithm uses $\Labelstoa_{G^2_{P',\delta,\gamma},P',\beta}(s)$ to find $a_{P',\beta}= \dfirst{G^2_{P',\delta,\gamma}}{s}{P'}{\beta}$.
More precisely, the label $\Labelstoa_{G^2_{P',\delta,\gamma},P',\beta}(s)$ allows the algorithm to retrieve $\ell_a = \LabelatoP_{G^0_{P',\delta},P,P',\gamma,\eps'}(a_{P',\beta})$.
Note that the label of $f$ contains $\ell_f = \LabelatfoP_{G^0_{P',\delta},P,P',\gamma,\eps'}(f)$.
The algorithm uses the labels $\ell_a$ and $\ell_f$ to retrieve vertices $b^1_{P',\beta,\gamma}$ and $b^2_{P',\beta,\gamma}$ on $P_1$ (resp. on $P_2$) such that $b^1_{P',\beta,\gamma}$ is an $\ddfirst{G^0_{P',\delta}\setminus \{f\} }{a_{P',\beta}}{P_1}{\gamma}{\eps' \gamma}$ and $b^2_{P',\beta,\gamma}$ is an $\ddfirst{G^0_{P',\delta}\setminus \{f\} }{a_{P',\beta}}{P_2}{\gamma}{\eps' \gamma}$.
Again, since the vertices of the graph $G^0_{P',\delta}$ are labeled, the algorithm actually obtains $\LabelPtoP_{G,P,\delta,\eps'}$ labels of $b^1_{P',\beta,\gamma}$ and of $b^2_{P',\beta,\gamma}$, denote these labels as $\ell^1_b$ and $\ell^2_b$ respectively.
Recall that the label of $f$ stores $\ell'_f = \LabelPtoP_{G,P,\delta,\eps'}(f)$.
The algorithm uses the labels $\ell^1_b$, $\ell^2_b$ and $\ell'_f$ to obtain:
\begin{itemize}
    \item a vertex $b^{1,1}_{P',\beta,\gamma,\delta}$ that is $\ddfirst{\Gf}{b^1_{P',\beta,\gamma}}{P_1}{\delta}{\eps' \delta}$,
    \item a vertex $b^{1,2}_{P',\beta,\gamma,\delta}$ that is $\ddfirst{\Gf}{b^1_{P',\beta,\gamma}}{P_2}{\delta}{\eps' \delta}$,
    \item a vertex $b^{2,1}_{P',\beta,\gamma,\delta}$ that is $\ddfirst{\Gf}{b^2_{P',\beta,\gamma}}{P_1}{\delta}{\eps' \delta}$,
    \item a vertex $b^{2,2}_{P',\beta,\gamma,\delta}$ that is $\ddfirst{\Gf}{b^2_{P',\beta,\gamma}}{P_2}{\delta}{\eps' \delta}$.
\end{itemize}

Let $b^\gamma_{Q,1} = \min_{\le P} \{b^\gamma_{Q,P^*} \mid P^* \textbf{ is a maximal consecutive subpath of }P\textbf{ in } G'_{Q} \textit{, }b^\gamma_{Q,P^*} \in P_1\}$ and  $b^\gamma_{Q,2} = \min_{\le P} \{b^\gamma_{Q,P^*} \mid P^* \textbf{ is a maximal consecutive subpath of }P\textbf{ in } G'_{Q} \textit{, }b^\gamma_{Q,P^*} \in P_2\}$.
Notice that $b^\gamma_{Q,1}$ and $b^\gamma_{Q,2}$ can be computed using the index of $f$ in $P$ to classify each $b^\gamma_{Q,P^*}$ either to $P_1$ or to $P_2$.
The algorithm uses the labels $\LabelPtoP_{G,P,\delta,\eps'}(b^\gamma_{Q,1})$, $\LabelPtoP_{G,P,\delta,\eps'}(b^\gamma_{Q,1})$, and $\LabelPtoP_{G,P,\delta,\eps'}(f)$ to obtain the vertices:
\begin{enumerate}
    \item $b^{\gamma,\delta}_{Q,1,1}$ that is an $\ddfirst{\Gf}{b^\gamma_{Q,1}}{P_1}{\delta}{\eps' \delta}$.
    \item $b^{\gamma,\delta}_{Q,1,2}$ that is an $\ddfirst{\Gf}{b^\gamma_{Q,1}}{P_2}{\delta}{\eps' \delta}$.
    \item $b^{\gamma,\delta}_{Q,2,1}$ that is an $\ddfirst{\Gf}{b^\gamma_{Q,2}}{P_1}{\delta}{\eps' \delta}$.
    \item $b^{\gamma,\delta}_{Q,2,2}$ that is an $\ddfirst{\Gf}{b^\gamma_{Q,2}}{P_2}{\delta}{\eps' \delta}$.
\end{enumerate}

The algorithm sets the vertex $b_1 = \min_{\le P}\big(\{ b^{x,1}_{P',\beta,\gamma,\delta} \mid \beta,\gamma,\delta \in \Gamma ,P'\textit{ is a subpath of }Q ,\beta + \gamma +\delta \le (1+3 \eps')\alpha, x\in \{1,2\} \} \cup \{ b^{\gamma,\delta}_{Q,x,1} \mid \gamma + \delta \le (1+2\eps')\alpha , x\in \{1,2\} \} \big)$ i.e. the first vertex on $P_1$ that was found across all choices of $\beta,\gamma,\delta$ and a subpath $P'$, or some $b^{\gamma,\delta}_{Q,x,1}$.
The algorithm also sets $b_2 = \min_{\le P}\big(\{ b^{x,2}_{\beta,\gamma,\delta} \mid \beta,\gamma,\delta \in \Gamma ,P'\textit{ is a subpath of }Q, \beta + \gamma +\delta \le (1+3 \eps')\alpha, x\in \{1,2\} \} \cup \{ b^{\gamma,\delta}_{Q,x,2} \mid \gamma + \delta \le (1+2\eps')\alpha , x\in \{1,2\} \}  \big)$.
That is, the first vertex on $P_2$ that was found across all guesses of $\beta,\gamma,\delta$ and a subpath $P'$ of $Q$, or $b^{\gamma,\delta}_{Q,x,2}$.
The algorithm returns $b_1$ and $b_2$.
\paragraph{Correctness.}
We show that when decoding the labels of two vertices $s$ and $f$, we indeed return a vertex $b_1$ that is $\ddfirst{\Gf}{s}{P_1}{\alpha}{\eps r}$.
The proof that $b_2$ is $\ddfirst{\Gf}{s}{P_2}{\alpha}{\eps r}$ is similar.
We start by showing that for every candidate $b'$ the algorithm considers for being $b_1$, we have $\dist_{\Gf}(s,b') \le \alpha + \eps r$.

Let $Q$ be the lowest separator in the decomposition such that $s$ is in one side of $Q$ and $f$ is strictly in the other side, and let $P^*$ be a maximal consecutive subpath of $P$ contained in $G'_Q$.
We start by focusing of a candidate $b'$ of the form $b^{\gamma,\delta}_{Q,x,1}$.
Recall that $\gamma + \delta \le (1+2\eps')\alpha$ and that $b^{\gamma,\delta}_{Q,x,1}$ is an $\ddfirst{\Gf}{b^{\gamma}_{Q,x}}{P_1}{\gamma}{\eps' \gamma}$ vertex for $b^{\gamma}_{Q,x}$ which is the first vertex of $P$ that can be reached on $P_1$ in $G'_Q$ with budget $\gamma$.
By definition, we have $\dist_{G'_{Q}}(s,b^\gamma_{Q,1}) \le \gamma$, and since $f\notin G'_Q$ we have $\dist_{\Gf}(s,b^\gamma_{Q,x}) \le \gamma$.
We also have $\dist_{\Gf}(b^{\gamma}_{Q,x},b^{\gamma,\delta}_{Q,x,1}) \le (1+\eps') \delta$.
It follows from triangle inequality that $\dist_{\Gf}(s,b^{\gamma,\delta}_{Q,x,1})\le \gamma + (1+\eps') \delta \le \alpha + 3\eps' \alpha \le (1+\eps) \alpha$ as required.

We now treat candidates of the form $b^{x,1}_{P',\beta,\gamma,\delta}$ for some $\beta,\gamma,\delta \in \Gamma$ and $x\in \{ 1,2\}$.
Recall that:
\begin{enumerate}
    \item $b^{x,1}_{P',\beta,\gamma,\delta}$ is a $\ddfirst{\Gf}{b^x_{P',\beta,\gamma}}{P_1}{\delta}{\eps' \delta}$.
    \item $b^x_{P',\beta,\gamma}$ is a $\ddfirst{G^0_{P',\delta} \setminus \{ f\}}{a_{P',\beta}}{P_x}{\gamma}{\eps' \gamma}$.
    \item $a_{P',\beta}$ is $\dfirst{G^2_{P',\delta,\gamma}}{s}{P'}{\beta}$
\end{enumerate}
Recall that $f\notin G^2_{P',\delta,\gamma}$ since $G^2_{P',\delta,\gamma}$ only contains vertices in the side of $s$ of the separator of $Q$ in the full recursive decomposition.
Also notice that in $G^2_{P',\delta,\gamma}$ the length of the path $P'$ is $0$, and in $G$ the length of $P'$ is at most $\eps' r$.
It follows that $\dist_{\Gf}(s,a_{P',\beta}) \le \dist_{G^2_{P',\delta,\gamma}}(s,a_{P',\beta}) + \eps' r \le \beta + \eps' r $.
Due to the same reasoning, we also have $\dist_{\Gf}(a_{P',\beta},b^x_{P',\beta,\gamma}) \le \dist_{G^0_{P',\delta} \setminus \{ f\}}(a_{P,\beta},b^x_{P',\beta,\gamma}) + \eps' r \le \gamma + \eps' \gamma + \eps' r$.
Directly from the definition of $\ddfirst{\Gf}{b^x_{P',\beta,\gamma}}{P_1}{\delta}{\eps' \delta}$ we have $\dist_{\Gf}(b^x_{P',\beta,\gamma},b^{x,1}_{P',\beta,\gamma,\delta}) \le \delta + \eps' \delta$.

From triangle inequality and $\beta+\gamma + \delta \le (1+3\eps') \alpha$ we get $\dist_{\Gf}(s,b^{x,1}_{P',\beta,\gamma,\delta}) \le (\beta +\eps' r) + (\gamma + \eps' \gamma + \eps' r) + (\delta + \eps' \delta) \le (\beta + \gamma + \delta) + \eps' (2r + \gamma + \delta) \le \alpha + \eps' (2r + 3\alpha + \gamma + \delta) \le \alpha + 7\eps' r = \alpha + \eps r$ as required.

We are now left with the task of proving $b_1 \le_P \dfirst{\Gf}{s}{P_1}{\alpha}$.
Notice that we have $b_1\in P_1$ since $b_1$ can either be $b_{Q,1}$ which is by definition on $P_1$, or $b^{x,1}_{\beta,\gamma,\delta}$ which is a $P_1$ output of a $\LabelPtoP$-type label.

Let $b^*= \dfirst{\Gf}{s}{P_1}{\alpha}$ and let $R$ be a shortest path from $s$ to $b^*$ in $\Gf$.
Let $a^*$ be the first vertex on $R$ that is also in $Q$, and let $c^*$ be the first vertex on $R[a^*..b^*]$ that is on $P$.
Notice that $a^*$ (and therefore, $c^*$) do not necessarily exist, but if $a^*$ exists, so does $c^*$.
\paragraph{Case 1: $a^*$ does not exist}
Let $b'$ be the first vertex on $R$ that is on $P$.
Let $\gamma^* = \min \{ x\in \Gamma  \mid x \ge \len(R[s,b'] \} $ and $\delta^* = \min \{ x\in \Gamma  \mid x \ge \len(R[b',b^*] \} $ and $b'\in P_x$ for $x\in \{1,2\}$.
Notice that $b'$ must be on one of the maximal consecutive subpaths of $P$ that are contained in $G'_Q$, we denote this maximal subpath as $P^*$.
Since $R[s,b']$ is a path with length at most $\gamma$ in $G'_Q$ from $s$ to $b'$, we have $b^{\gamma^*}_{Q,P^*} =\dfirst{G'_{Q}}{s}{P^*}{\gamma} \le_P b'$.
In particular, $b^{\gamma^*}_{Q,x} \le_P b'$ as an $\le_P$ minimum in a set containing $P^*$.
Notice that $\gamma^*+\delta^* \le (1+2\eps')\alpha$, so the algorithm considered a candidate $b^{\gamma^*,\delta^*}_{Q,x,1}$ that is an $\ddfirst{\Gf}{b^{\gamma^*}_{Q,x}}{P_1}{\delta^*}{\eps' \delta^*}$.
Since $P_1[b^{\gamma^*}_{Q,x},b']R[b',b^*]$ is a path in $\Gf$ from $b^{\gamma^*}_{Q,x}$ to $b^*$ of length at most $\delta^*$, we have that $b^{\gamma^*,\delta^*}_{Q,x,1} \le_P \dfirst{\Gf}{b^{\gamma^*}_{Q,x}}{P_1}{\gamma^*} \le_P b^*$.

This concludes the proof of this case, as $b_1 \le_P b^{\gamma^*,\delta^*}_{Q,x,1}$ as the $\le_P$ minimum of a set containing the latter.

\paragraph{Case 2: $a^*$ exists, and $c^* \in P_1$ }
Let $P'$ be the subpath of $Q$ that contains $a^*$.
Let:
\begin{enumerate}
    \item$\beta^* = \min\{x \in \Gamma \mid x \ge \len(R[s,a^*])\}$,
    \item  $\gamma^* = \min\{x \in \Gamma \mid x \ge \len(R[a^*,c^*])\}$, and
    \item $\delta^* =  \min\{x \in \Gamma \mid x \ge \len(R[c^*,b^*])\}$.
\end{enumerate}
Notice that $\beta^* + \gamma^* + \delta^* \le \len(R) + 3\eps' \alpha$.
Therefore, the decoding algorithm computed some $b^{1,1}_{\beta^*,\gamma^*,\delta^*}$ such that:
\begin{enumerate}
    \item $b^{1,1}_{\beta^*,\gamma^*,\delta^*}$ is a $\ddfirst{\Gf}{b^1_{\beta^*,\gamma^*}}{P_1}{\delta^*}{\eps' \delta^*}$.
    \item $b^1_{\beta^*,\gamma^*}$ is a $\ddfirst{G^0_{P',\delta^*} \setminus \{ f\}}{a_{\beta^*}}{P_1}{\gamma^*}{\eps' \gamma^*}$.
    \item $a_{\beta^*}$ is $\dfirst{G^2_{P',\delta^*,\gamma^*}}{s}{P'}{\beta^*}$.
\end{enumerate}
Notice that we omit $P'$ from the subscript of the vertices listed above (i.e. $b^{1,1}_{P',\beta^*,\gamma^*,\delta^*}$ is written as $b^{1,1}_{\beta^*,\gamma^*,\delta^*}$). This is done for convenience, as this notation will not be used in the current context with a subpath other than $P'$.

Recall that $G^2_{P',\gamma^*,\delta^*}$ is a vertex labeled version of $G'_Q$, with the weights of the edges of $P'$ set to $0$.
Since $R[s,a^*]$ is a path from $s$ to $P'$ in $G'_Q \setminus \{ f\} = G'_Q$ with length at most $\beta^*$, it is in particular a path from $s$ to $a^*\in P'$ in $G^2_{P',\gamma^*,\delta^*}$ with length at most $\beta^*$ and therefore we have $a_{\beta^*} \le_{P'} a^*$.

Recall that $G^0_{P',\delta^*}$ is a vertex labeled version of $G$ with the weights of all edges of $P'$ set to $0$ without outgoing edges from $P$.
Since $R[a^*,c^*]$ is a path from $a^*$ to $c^*$ in $\Gf$ with length at most $\gamma^*$ that do not use outgoing edges of $P$ (as it is internally disjoint from $P$), we have that $P'[a_{\beta^*},a^*] \cdot R[a^*,c^*]$ is a path from $a_{\beta^*}$ to $c^* \in P_1$ in $G^0_{P',\delta^*}\setminus \{f\}$ with length at most $\gamma^*$ (also recall that $f\notin P'$).
We therefore have $b^1_{\beta^*,\gamma^*} \le_P c^*$.

Furthermore, consider the path $P_1[b^1_{\beta^*,\gamma^*},c^*] \cdot R[c^*,b^*]$.
It is a path with length at most $\delta^*$ (recall that $\len(P) =0$) in $\Gf$.
Therefore, we have $b^{1,1}_{\beta^*,\gamma^*,\delta^*} \le_P b^*$.
Due to the minimality of the returned vertex $b_1$ on $P$ across all values of $\beta,\gamma,\delta$ and all subpaths $P'$, we have that $b_1 \le_P b^{1,1}_{\beta^*,\gamma^*,\delta^*} \le_P b^*$ as required.

\paragraph{Case 3: $a^*$ exists and $c^* \in P_2$}
This proof for this case is completely identical to the proof of the previous case, replacing $b^{1,1}_{\beta^*,\gamma^*,\delta^*}$ with $b^{2,1}_{\beta^*,\gamma^*,\delta^*}$.
\end{proof}

\section{The $\LabelPtoP$ Labeling (Proof of \cref{lem:ptoplabel})}\label{sec:ptop}

In this section we prove \cref{lem:ptoplabel} by extending the ideas of \cref{sec:singlePath} from reachability to approximate distances.
We recall the settings of \cref{lem:ptoplabel}.
$G$ is a planar graph, $P$ is a $0$-length path of $G$ whose two endpoints lie on the same face and $\alpha,\eps\in \mathbb R^+$. Our goal is to assign $\Opoly$-sized labels to the vertices of $P$ such that, given the labels of two vertices $b$ and $f$ of $P$, one can compute the indices on $P$ of some vertices $b_1$ and $b_2$ that are $\ddfirst{\Gf}{b}{P_1}{\alpha}{\eps\alpha}$ and $\ddfirst{\Gf}{b}{P_2}{\alpha}{\eps\alpha}$, respectively. Throughout, $P_1$ and $P_2$ denote the prefix and suffix of $P$ before and after vertex $f$ (without $f$) respectively.

\begin{figure}[htb]
\begin{center}
\includegraphics[width=0.5\textwidth]{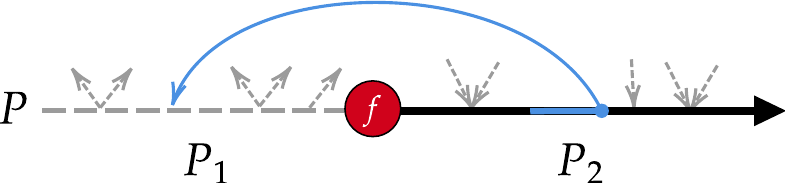}
\caption{An illustration of $G_1$.
We are interested in $P_2$ to $P_1$ paths  (like the blue path) that may start with some edges of $P_2$ and then continue with a subpath which is internally disjoint from $P$.
Formally, this is achieved by removing all in-going edges to $P_2$ and all out-going edges from $P_1$ (the removed edges are displayed in gray in the figure).
\label{fig:G1path}
}
\end{center}
\end{figure}

We define $G_1$ (resp. $G_2$) to be the graph obtained from $\Gf$ by first removing all in-going edges to vertices of $P_2$ (resp. $P_1$), except for the edges of $P_2$ (resp. $P_1$) itself, and then removing all out-going edges from vertices of $P_1$ (resp. $P_2$), including the edges of $P_1$ (resp. $P_2$), see \cref{fig:G1path}.
Intuitively, $G_1$ (resp. $G_2$) is a version of $\Gf$ in which every path from $P_2$ (resp. $P_1$) to $P_1$ (resp. $P_2$)  starts with some subpath of $P_2$ (resp. $P_1$) and is then internally disjoint from $P$.

We divide our task into four auxiliary labeling schemes.

\begin{enumerate}
	\item $\LabelPztoPi_{\beta,\gamma}$ - Given the labels of two vertices $f<_P b$, one can obtain $\dfirst{\Gf}{x}{P_1}{\gamma}$ and  $\dfirst{\Gf}{x}{P_2}{\gamma}$ for a vertex $x$ that is $\ddfirst{G_1}{b}{P_1}{\beta}{\eps\beta}$.
We introduce a labeling scheme for this problem in \cref{sec:easy}.

	\item $\LabelPitoPz_{\beta,\gamma}$ - Given the labels of two vertices $b<_P f$, one can obtain $\dfirst{\Gf}{x}{P_1}{\gamma}$ and  $\dfirst{\Gf}{x}{P_2}{\gamma}$ for a vertex $x$ that is $\ddfirst{G_2}{b}{P_2}{\beta}{\eps\beta}$.
The labeling scheme of $\LabelPztoPi_{\beta,\gamma}$ (\cref{sec:easy}) can be adjusted to solve this problem.

	\item $\LabelPztoPz_\beta$ - Given the labels of two vertices $f<_P b$ one can compute $\ddfirst{G\setminus (P_1\cup\{f\})}{b}{P_2}{\beta}{\eps\beta}$.
We introduce a labeling scheme for this problem in \cref{sec:hard_prob}.

	\item $\LabelPitoPi_\beta$- Given the labels of two vertices $b<_P f$ one can compute $\ddfirst{G\setminus (P_2\cup\{f\})}{b}{P_1}{\beta}{\eps\beta}$.
The labeling scheme of $\LabelPztoPz_\beta$ (\cref{sec:hard_prob}) can be adjusted to solve this problem.
\end{enumerate}

Exploiting the labeling schemes for the auxiliary tasks, we prove the following lemma (see \cref{fig:PtoP}).

\begin{figure}[htb]
\begin{center}
\includegraphics[width=0.5\textwidth]{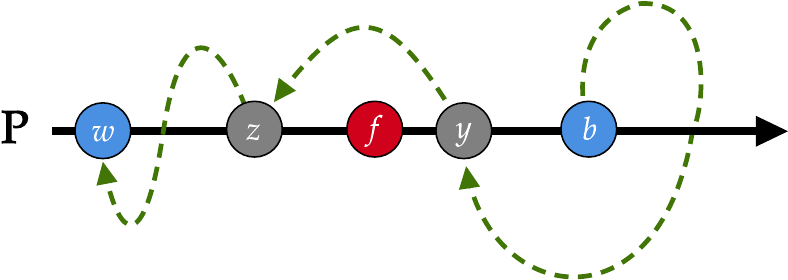}
\caption{An illustration of a path $R$ from $b$ to $w=\dfirst{\Gf}{b}{P_1}{\alpha}$.
The first subpath is from $b$ to $y\in P_2$ (without using $P_1$),
the second subpath is from $y$ to $z\in P_1$ using a path that is internally disjoint from $P$ except for a prefix.
Finally, the third subpath from $z$ to $w\in P_1$ is a path in $\Gf$.
The first part is approximated using $\LabelPztoPz$, the second and third parts are approximated using $\LabelPztoPi$.
\label{fig:PtoP}
}
\end{center}
\end{figure}

\ptoplabellem*
\begin{proof}
Let $\eps'=\eps/6$ and let $\Gamma=\{i\eps'\alpha\mid i\in[0,\ceil{1/\eps'}]\}$ be the set of multiples of $\eps'\alpha$.
For every $(\gamma,\delta)\in\Gamma^2$ we define the labeled graph $G_{\gamma,\delta}$ as a copy of $G$ in which every vertex $v\in V$ is labeled with $\LabelPztoPi_{\gamma,\delta}(v)$ and $\LabelPitoPz_{\gamma,\delta}(v)$, where the labeles are computed with approximation factor $\eps'$.
\hlgray{For every $v\in V$, the label $\LabelPtoP(v)$ stores for every $(\beta,\gamma,\delta)\in\Gamma^3$
the labels of $\LabelPztoPz_{\beta}(v)$ and $\LabelPitoPi_{\beta}(v)$ computed with approximation factor $\eps'$ on the graph $G_{\gamma,\delta}$, and also the label of $v$ in $G_{\gamma,\delta}$.} 

\paragraph{Size.}
By \cref{lem:labelPztoPi,lem:labelPztoPz} we have that the size of each label of $\LabelPztoPi,\LabelPitoPz, \LabelPztoPz$ and $\LabelPitoPi$ is $\Opoly$. Thus, the size of $\LabelPtoP(v)$ is $\Opoly$.
Note that $\LabelPztoPz$ and $\LabelPitoPi$ labels on the graph $G_{\beta,\delta}$, in which vertices have labels of size $\Opoly$, still have size $\Opoly$.

\paragraph{Decoding.}
Given the labels of $b$ and $f$, if $f<_P b$ the computation of the indices on $P$ of some vertices $b_1$ and $b_2$ that are $\ddfirst{\Gf}{b}{P_1}{\alpha}{\eps\alpha}$ and $\ddfirst{\Gf}{b}{P_2}{\alpha}{\eps\alpha}$, respectively, is done as follows.
Iterate over all triplets $(\beta,\gamma,\delta)\in\Gamma^3$ such that $\beta+\gamma+\delta\le (1+3\eps')\alpha$.
The algorithm uses  $\LabelPztoPz_{\beta}(b)$ and $\LabelPztoPz_{\beta}(f)$ to compute some vertex $b'_{\beta,\gamma,\delta}\in G_{\gamma,\delta}$ which is $\ddfirst{G\setminus (P_1\cup\{f\})}{b}{P_2}{\beta}{\eps'\beta}$.
Then, using $\LabelPztoPi_{\gamma,\delta}(f)$ and $\LabelPztoPi_{\gamma,\delta}(b'_{\beta,\gamma,\delta})$ we compute $b^1_{\beta,\gamma,\delta}=\dfirst{\Gf}{x_{\beta,\gamma,\delta}}{P_1}{\delta}$ and  $b^2_{\beta,\gamma,\delta}=\dfirst{\Gf}{x_{\beta,\gamma,\delta}}{P_2}{\delta}$ for a vertex $x_{\beta,\gamma,\delta}$ that is $\ddfirst{G_1}{b'_{\beta,\gamma,\delta}}{P_1}{\gamma}{\eps'\gamma}$.
Finally, $b_1$ and $b_2$ are the earliest vertices of $P$ among all $b^1_{\beta,\gamma,\delta}$  and $b^2_{\beta,\gamma,\delta}$, respectively.
The case where $b<_P f$ is decoded in a similar way, using $\LabelPitoPi$ and $\LabelPitoPz$.

\paragraph{Correctness.}
First notice that for every $(\beta,\gamma,\delta)$ we have $\dist_{\Gf}(b,b^1_{\beta,\gamma,\delta})\le\dist_{\Gf}(b,b'_{\beta,\gamma,\delta})+\dist_{\Gf}(b'_{\beta,\gamma,\delta},x_{\beta,\gamma,\delta})+\dist_{\Gf}(x_{\beta,\gamma,\delta},b^1_{\beta,\gamma,\delta})$.
By definition of $\LabelPztoPz$ we have $\dist_{\Gf}(b,b'_{\beta,\gamma,\delta})\le (1+\eps')\beta$ and by definition of $\LabelPztoPi$ we have $\dist_{\Gf}(b'_{\beta,\gamma,\delta},x_{\beta,\gamma,\delta})\le (1+\eps')\gamma$ and $\dist_{\Gf}(x_{\beta,\gamma,\delta},b^1_{\beta,\gamma,\delta})\le (1+\eps')\delta$.
Thus, $\dist_{\Gf}(b,b^1_{\beta,\gamma,\delta})\le(1+3\eps')\alpha+3\eps'\alpha=(1+6\eps')\alpha=(1+\eps)\alpha$.
By similar arguments, $\dist_{\Gf}(b,b^2_{\beta,\gamma,\delta})\le(1+6\eps')\alpha=(1+\eps)\alpha$, as required.

It remains to prove that $b_1\le_P\dfirst{\Gf}{b}{P_1}{\alpha}$ and $b_2\le_P\dfirst{\Gf}{b}{P_2}{\alpha}$.
We prove the former, the proof of the latter is similar.
Let $R$ be a shortest path from $b$ to $w=\dfirst{\Gf}{b}{P_1}{\alpha}$ in $\Gf$ (see \cref{fig:PtoP}.
Let $z$ be the first vertex on $R$ that is on $P_1$ and let $y$ be the last vertex on $R[b,z]$ which is on $P_2$.
Let $\beta=\min \{q\in\Gamma\mid q\ge\len(R[b,y])\}$, let $\gamma=\min \{ q\in\Gamma\mid q\ge\len(R[y,z])\}$ and let $\delta=\min \{ q\in\Gamma\mid q\ge\len(R[z,w])\}$.
Notice that $\beta+\gamma+\delta\le\alpha+3\eps'\alpha$.
We will show that $b^1_{\beta,\gamma,\delta}\le_P w$, the claim then follows from the minimality (w.r.t $\le_P$) of $b_1$ on $P$ across all choices of $\beta,\gamma,\delta$.

First, $b'_{\beta,\gamma,\delta}\le_P \dfirst{G\setminus (P_1\cup\{f\})}{b}{P_2}{\beta}\le_P y$ since $y$ can be reached in $G\setminus (P_1\cup\{f\})$ from $b$ with a path ($R[b,y]$) of length at most $\beta$.
Second, $x_{\beta,\gamma,\delta}\le_P \dfirst{G_1}{b'_{\beta,\gamma,\delta}}{P_1}{\gamma}\le_P z$ since $z$ can be reached in $G_1$ from $b'_{\beta,\gamma,\delta}$ with a path ($P[b'_{\beta,\gamma,\delta},y]\cdot R[y,z]$) of length at most $\gamma$.
Finally, $b^1_{\beta,\gamma,\delta}\le_P \dfirst{\Gf}{x_{\beta,\gamma,\delta}}{P_1}{\delta}\le_P w$ since $w$ can be reached in $\Gf$ from $x_{\beta,\gamma,\delta}$ with a path ($P[x_{\beta,\gamma,\delta},z]\cdot R[z,w]$) of length at most $\delta$.
\end{proof}

\subsection{The $\LabelPztoPi$ labeling}\label{sec:easy}

In this section we focus on a vertex $f\in P$.
We denote $G'=G^1$.
For a vertex $b\in P_2$ let $D_b$ be a shortest path from $b$ to $\dfirst{G'}{b}{P_1}{\beta}$. If $\dfirst{G'}{b}{P_1}{\beta}=null$ we say that $b$ {\em cannot bypass} $f$.
This section is mostly dedicated to proving the following lemma, which leads to a simple labeling scheme for $\LabelPztoPi$.

\begin{lemma}\label{lem:sequenceseasy}
    There are sequences $I = ( b_1 <_{P_2} b_2 <_{P_2} \ldots <_{P_2} b_t) \subseteq P_2$ and $V =(v_1, v_2, \ldots , v_{t})\subseteq P_1$ such that:
    \begin{enumerate}
        \item For every vertex $b\in P_2$ we have $\min_{\le_{P_1}} \{v_i\mid b_i\ge_{P_2} b\}$ is an $\ddfirst{G'}{b}{P_1}{\beta}{\eps \beta}$.

        \item Every vertex $b >_{P_2} b_t$ cannot bypass $f$.
        \item $t \in O(1/\eps)$.
    \end{enumerate}
\end{lemma}

For $v\in P_2$ we denote $v_{first}=\dfirst{G'}{v}{P_1}{\beta}$.
Let $B_f$ be the set of all pairs $(v,v_{first})$, that are maximal with respect to inclusion (i.e. $B_f$ does not include a pair $(v,v_{first})$ if there exists  $(u,u_{first})$ with $v<_P u$ and $u_{first}\le_P v_{first}$).
By the following claim, we can assume that for every $b\in P_2$, the path $D_b$ starts with a prefix $P[b..u]$ and procedes with $D_u$ such that $(u,u_{first}) \in B_f$.

\begin{claim}\label{claim:maxmial_bypass}
    Let $b>_P f$ be a vertex that can bypass $f$.
    There is a path from $b$ to $b_{first}$ beginning with a subpath of $P$ followed by some $D_u$ such that $(u,u_{first})\in B_f$.
\end{claim}
\begin{proof}
    Since $G'$ does not contain in-going edges into vertices of $P_2$ (except for the edges of $P_2$) and does not contain out-going edges from vertices of $P_1$ it must be that $D_b=P[b,u]\cdot D_b[u,b_{first}]$ for some vertex $u\ge_P b$ such that $D_b[u,b_{first}]$ is internally disjoint from $P$.
    Notice that $\len(D_b[u,b_{first}])=\len(D_b)\le \beta$ since $\len(P)=0$, therefore $u_{first}\le_P b_{first}$.
    Moreover, $\len(P[b,u]\cdot D_u)=\len(D_u)\le \beta$  and therefore, it must be that $u_{first}\ge_P b_{first}$, thus $b_{first}=u_{first}$.

    Now, let us fix $D_b$ as a $b$-to-$b_{first}$ path of length at most $\beta$ that maximizes $u$ (with respect to the $\le_P$ order).
    (Notice that $D_b$ is not necessarily a shortest path.)
    Assume by contradiction that $(u,u_{first})\notin B_f$, so there is a vertex $v\in P_2$  such that $(v,v_{first})\in B_f$ and $v_{first}\le_P u_{first}$ and $u<_P v$.
    However, in this case $\len(P[b,v]\cdot D_v)\le\len(D_v)\le \beta$ which implies $v_{first}\ge_P b_{first}=u_{first}$, and therefore $v_{first}=b_{first}$. The path $P[b,v]\cdot D_v$ contradicts the maximality of $u$.
\end{proof}

Since the endpoints of $P$ lie on the same face, a path $D_b$ that emanates $P$ to the right (resp. left) of $P$ must enter $P$  from the right (resp. left) of $P$.
Let $B_f^{left}$ and $B_f^{right}$ be the set of left and right pairs in $B_f$, respectively.
Let $G'_{right}$ ($G'_{left}$) be the induced graph from $P$, restricted only to vertices $v$ such that there is a path from $P$, that is internally disjoint from $P$, to $v$  that emanates $P$ to the right (resp. left).

We prove \cref{lem:sequenceseasy} for the graph $G'_{right}$ (instead of $G'$).
Applying the same argument for $G'_{left}$ would produce a second pair of sequences $I, V$.
It is straightforward to merge the two $I$'s sequences and the two $V$'s sequences obtained for both cases to complete the proof of \cref{lem:sequenceseasy}, since any $P_2$ to $P_1$ path is contained in either $G'_{right}$ or $G'_{left}$.

\paragraph{An algorithm for \boldmath$\eps=1$.}
We start with a warm-up solution for $\eps=1$ (i.e., a 2-approximation).
In this case, we show sequences $I=(b_1)$ and $V=(v_1)$.
We define $v_1$ and $b_1$ as follows.
$v_1$ is the first vertex of $P_1$ such that there exists a $(b_{v_1},v_1)\in B_f^{right}$. 
$b_1$ is the last vertex of $P_2$ such that there is a pair $(b_1,(b_{1})_{first})\in B_f^{right}$.
We show that the sequences $I$ and $V$ satisfy the lemma for $\eps = 1$.
Notice that the demand $t\in O(1 / \eps)$ is trivial in the current context as $\eps$ and $t$ are both constants.
Additionally, notice that from the definition of $b_1$, every vertex $b>_{P_2} b_1$ cannot bypass $f$.
\begin{claim}
    If $f<_P b\le_P b_1$ then $v_1$ is a $\ddfirst{G'}{b}{P_1}{\beta}{\beta}$.
\end{claim}
\begin{proof}
Notice that $b_{first} >_{P_1} v_1$ by \cref{claim:maxmial_bypass} and the definition of $v_1$.
We will show that there is a path of length at most $2\beta$ from $b_1$ to $v_1$, which concludes the claim as every $b$ can reach $v$ via $b_1$ with a path of length $2\beta$, making $v$ a $\ddfirst{G'}{b}{P_1}{\beta}{\beta}$ vertex.

Consider the paths $D_{b_1}$ and $D_{b_{v_1}}$.
Since both paths emanate $P$ to the right, and since $ v_1\le_P (b_1)_{first} <_P b_{v_1} \le_{P} b_1$ (due to both pairs being in $B_f$), it must be the case that $D_{b_1}$ intersects with $D_{b_{v_1}}$, say at vertex $z$.
It follows that the path $D_{b_1}[b_1,z] \cdot D_{b_v}[z,v_1]$ is a path of length at most $2\beta$ from $b_1$ to $v_1$, as required.
\end{proof}

\paragraph{An algorithm for any \boldmath$\eps$.}
We now extend the above solution from $\eps=1$ to any $\eps>0$.
We will partition $B_f^{right}$ into $O(1/\eps)$ subsets.
For ease of presentation let us denote $B=B_f^{right}$.

Let $A_1=(b_1,p_1)$ and $A_2=(b_2,p_2)$ with $b_1>_{P_2}b_2$ be two different pairs in $B$.
We first note that due to the maximality with respect to inclusion $p_2<_P p_1<_P f<_P b_2<_P b_1$.
Thus, we say that $A_1$ and $A_2$ {\em cross} each other, and denote it by $A_1\cross A_2$.
If there is a path $S$ from $b_1$ to $p_2$ that is internally disjoint from $P$ and is to the right of $P$ (by definition of $G'_{right}$) of length at most $(1+\eps)r$ then we say that $A_1$ and $A_2$ {\em good-cross} each other, and denote it as $A_1\goodcross A_2$. Otherwise, we say that they bad-cross each other and denote it by $A_1\badcross A_2$.

\begin{lemma}\label{lem:partition_B}
    $B$ can be partitioned into $O(1/\eps)$ subsets $B_1, B_2, \dots, B_t$ such that every two pairs in the  same $B_i$ good-cross each other.
\end{lemma}
\begin{proof}
We present an algorithm that partitions $B$.
The algorithm runs in phases, creating the  subset $B_i$ in the $i$'th phase.
At the beginning of every phase, we initialize a set $B_{good}= \emptyset$.
During a phase, we examine the pairs of $B$ in decreasing $<_P$ order of their first vertex.
When examining a pair $A$, we check if $A$ good-crosses all pairs in $B_{good}$, and if so we add $A$ into $B_{good}$ and remove $A$ from $B$.
At the end of the $i$'th phase, we set $B_i=B_{good}$ and append $B_i$ to the partition and if $B=\emptyset$ the algorithm terminates.
Clearly $|B_{good}|\ge 1$, so the algorithm halts with a partition of the initial $B$.
It should also be clear from the construction of $B_{good}$ at every iteration that two pairs in the same $B_i$ good-cross each other.

\begin{figure}[t]
\begin{center}
 \includegraphics[width=0.5\textwidth]{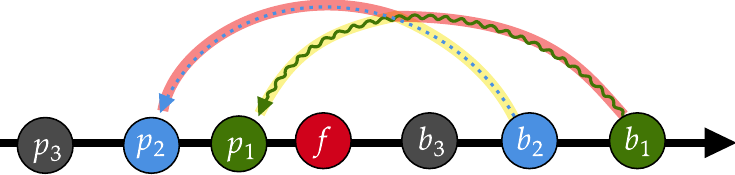}

\includegraphics[width=0.5\textwidth]{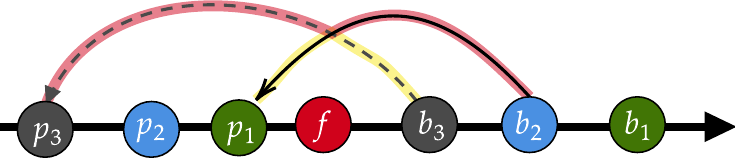}
\caption{An illustration for the proof of \cref{lem:partition_B}.
In the top figure, the green wavy path corresponds to $A_1$ and the blue dotted path corresponds to $A_2$.
The sum of the green and blue paths is bounded by $2\beta$.
Since $A_1\badcross A_2$, the highlighted red path from $b_1$ to $p_2$ is expensive: its length is at least $(1+\eps)\beta$.
Since the yellow and red paths sum up to $\len(A_1) + \len(A_2)$, it follows that the highlighted yellow path from $b_2$ to $p_1$ is cheap:  its length is at most $(1-\eps)\beta$.
\\
The bottom figure demonstrates the next inductive step.
Now, the cheap $b_2$-to-$p_1$ path from the previous step is shown in solid black.
The dashed path from $b_3$ to $p_3$ represents $A_3$.
As in the top picture, we decompose $A_3$ and the cheap path into a highlighted red $b_2$-to-$p_3$ path and a highlighted yellow $b_3$-to-$p_1$ path.
It follows from $A_1 \badcross A_2$ that the red path is of length more than $(1+\eps)\beta$.
Since we started with a cheap $b_2$-to-$p_1$ path, the sum of red and yellow is now $2\beta - \eps \beta$, which leads to the length of yellow being at most $\beta - 2\eps \beta$.
\label{fig:easy}
}
\end{center}
\end{figure}

To conclude the proof, we show that the number of phases $t\le 1/\eps+1$.
    By definition, for every $i \in [2..t]$ there is a pair $A_i=(b_i,p_i)$ in $B_i$ and  a pair $A_{i-1}=(b_{i-1},p_{i-1})$ in $B_{i-1}$ with $b_{i-1}>_{P_2} b_i$ such that $A_{i-1}\badcross A_i$.
    Let $A_t$ be some pair in $B_t$ and $A_{t-1},A_{t-2} , \ldots , A_1$ be a sequence of pairs with $A_{i-1} \in B_{i-1}$ being an arbitrary pair satisfying the above with respect to $A_i \in B_i$, starting with $A_t$.

    We claim that for every $i\in[1..t]$ there is a path from $b_i$ to $p_1$ in $G'_{right}$ of length at most $(1-(i-1)\eps)\beta$ (see \cref{fig:easy}).
    The claim is trivial for $i=1$.
    We thus assume the claim holds for $i$, and prove it for $i+1$.
    Let $P_1$ be the path from $b_i$ to $p_1$ and let $P_2 = D_{b_{i+1}}$ be the path corresponding to the pair $(b_{i+1},p_{i+1})$.
    Since $p_{i+1} \le_P p_1 <_P b_{i+1} <_P b_{i}$, and since paths only emanate to the right of $P$ in $G'_{right}$, we must have that $P_1$ and $P_2$ intersect at some vertex $z$.
    Since $A_i\badcross A_{i+1}$, any path from $b_i$ to $p_{i+1}$ in $G'_{right}$ is of length more than $(1+\eps)\beta$.
    In particular the path $S=P_1[b_i,z]\cdot P_2[z,p_{i+1}]$ is of length more than $(1+\eps)\beta$.
    Let $S'=P_2[b_{i+1},z]\cdot P_1[z,p_1]$.
    We show that $\len(S')\le (1-i\eps)\beta$.
    By the induction hypothesis, $\len(P_1)\le (1-(i-1)\eps)\beta$.
    By the definition of $B$ we have $\len(P_2)\le \beta$.
    Notice that
    $\len(P_1)+\len(P_2)=\len(S)+\len(S')$.
    Since $\len(S)>(1+\eps)\beta$ we have $\len(S')\le \beta+(1-(i-1)\eps)\beta - (1+\eps)\beta= (1-i\eps)\beta$, as required.

To conclude, if $t>1/\eps+1$ then there is a path from $b_t$ to $p_1$ of negative length, a contradiction.
\end{proof}

We proceed to define the sequences $I$ and $V$.
Let $B_1,B_2, \ldots B_t$ be the partition obtained by \cref{lem:partition_B}.
For every $i\in[t]$, let $p'_i$ be the first vertex of $P$ such that there exists a pair $(\cdot,p'_i)\in B_i$.
Let $b'_i$ be the last vertex of $P$ such that there is a $(b'_i,\cdot)\in B_i$.
Assume that the $B_i$'s are indexed such that $b'_1<_P b'_2 <_P \ldots <_P b'_t$ \footnote{The algorithm described in the proof of \cref{lem:partition_B} would output $B_i$'s exactly in the reverse of this order}.
We define $I = (f,b'_1,b'_{2} \ldots ,b'_t)$ and $V= (p'_{1},p'_{2}, \ldots p'_t)$.

It follows immediately from \cref{lem:partition_B} that the lengths of the sequences are $O(\frac{1}{\eps})$.
Additionally, note that each $b'_i$ is a last vertex on $P$ in its respective $B_i$, and $b'_t$ is the last vertex on $P$ among all $b'_i$'s.
Since the sets $B_i$'s form a partition of $B$, it holds that $b'_t$ is the last vertex of $P$ such that there is a pair $(b'_t,\cdot)$ in $P$.
From the definition of $B$, it means that every $b>_{P_2} b'_t$ cannot bypass $f$.
We prove via the following claim that $I$ and $V$ satisfy the remaining property of \cref{lem:sequenceseasy}.
\begin{claim}
     For every vertex $b\in P_2$ we have that $\min_{\le_{P_1}} \{v'_i\mid b'_i\ge_{P_2} b\}$ is a $\ddfirst{G'}{b}{P_1}{\beta}{\eps \beta}$
\end{claim}
\begin{proof}
      First, we claim that for every $b'_i \ge_{P_2} b$, we have $\dist_{G'}(b,v'_{i}) \le (1+\eps)\beta$.
      Recall that $b'_i$ and $v'_i$ are obtained as components of two pairs $(b'_i,\cdot)$ and $(\cdot,v'_i)$ in $B_i$.
      Since both pairs are in $B_i$, they good-cross each other.
      Therefore, $\dist_{G'}(b'_i,v'_i) \le (1+\eps) \beta$ via some shortest path $D_i$.
     We thus have that $P[b,b'_i]\cdot D_i$ is a path from $b$ to $v'_i$ of length at most $(1+\eps)\beta$.

      By \cref{claim:maxmial_bypass} there is a pair $(u,b_{first})\in B$ such that the path $P[b,u]\cdot D_u$ is of length at most $\beta$.
      Let $B_j$ be the set that contains the pair $(u,b_{first})$.
      Recall that $(b_j,\cdot) \in B_j$ is the pair with maximal vertex $b_j$ (with respect to $\le_{P_2}$) in $B_j$, and $(\cdot, v_j)$ is the pair with minimal (with respect to $\le_{P_1}$) vertex $v_j$ in $B_j$.
      It follows that $b_j \ge_{P_2} u \ge_{P_2} b$ and $v_j \le_{P_1} b_{first}$.
      We have therefore shown that there is $b_j \in I$ such that $b_j \ge_{P_2} b$ and $v_{j} \le_{P_1} b_{first}$.
      Therefore, in particular, the vertex $b_i \ge_{P_2} b \in I$ that minimizes $v_i$ also has $v_i \le_{P_1} b_{first}$.

      In conclusion, for $v=\min_{\le_{P_1}} \{v_i\mid b_i\ge_{P_2} b\}$, we have shown that $\dist_{G'}(b,v) \le (1+\eps)\beta$, and that $v\le_{P_1} b_{first}$.
      Therefore, $v$ is an $\ddfirst{G'}{b}{P_1}{\beta}{\eps \beta}$ as required.
\end{proof}

This concludes the proof of \cref{lem:sequenceseasy}, thus implying a simple labeling scheme for $\LabelPztoPi_{\beta,\gamma}$:
\begin{lemma}\label{lem:labelPztoPi}
    There exists a labeling scheme $\LabelPztoPi_{\beta,\gamma}$ with labels of size $\Opoly$.
\end{lemma}
\begin{proof}
    \hlgray{Every vertex $u\in P$ stores its index in $P$.   Moreover, $u$ also stores the sequences $I_u$ and $V_u$ as defined in \mbox{\cref{lem:sequenceseasy}}.
    Additionally, for every $v\in V_u$ the label of $u$ stores $\dfirst{G\setminus\{u\}}{v}{P_1}{\gamma}$ and  $\dfirst{G\setminus\{u\}}{v}{P_2}{\gamma}$.} 
    \cref{lem:sequenceseasy} implies that the length of a label is $\Opoly$ .

    When two labels $\LabelPztoPi_\beta(b)$ and  $\LabelPztoPi_\beta(f)$ are given, one can compute $x = \min_{\le_{P_1}} \{v_i\mid b_i\ge_{P_2} b\}$ using the label of $f$ and the index of $b$, and by the first property of \cref{lem:sequenceseasy} it is guaranteed that $x$ is an $\ddfirst{G'}{b}{P_1}{\beta}{\eps \beta}$, and returns $\dfirst{\Gf}{x}{P_1}{\gamma}$ and  $\dfirst{\Gf}{x}{P_2}{\gamma}$, that were stored explicitly, as required.
\end{proof}

\subsection{The $\LabelPztoPz$ labeling}\label{sec:hard_prob}

In this section, we provide a labeling scheme for $P$ such that given the labels of two vertices $f<_P b$ one can compute $\ddfirst{G\setminus (P_1\cup\{f\})}{b}{P_2}{\beta}{\eps\beta}$.

\paragraph{A labeling scheme  for \boldmath$\eps=1$.}
We start with a warm-up solution for $\eps=1$ (i.e., when the goal is to compute $\ddfirst{G\setminus (P_1\cup\{f\})}{b}{P_2}{\beta}{\beta}$).

For any vertex $v\in P$, let $v_{last}$ be the last vertex in $P$ that has a path to $v$ in $G$ of length at most $\beta$ that does not touch any vertex of $P$ before $v$.
The pair $(v_{last},v)$ is called a \emph{detour}.
Let $D$ be the set of all detours $(v_{last},v)$ for $v\in\P$.
Recall that the \emph{size} of the detour $(v_{last},v)$ is the number of vertices in $P[v,v_{last}]$.
We partition the set of detours $D$ into exponential classes of sizes, as follows:
For any $i\in [0,\ceil{\log_{1.1}|P|+1}]$ let $D_i$ be the set of all detours whose size is in the range $[1.1^i,1.1^{i+1})$.
Let $x=1.1^i$.
We further partition every $D_i$ into subsets:
For any $k$, let $D_{i,k}$ be the subset of $D_i$ that contains all the detours $(v_{last},v)$ where both $v$ and $v_{last}$ are in the subpath $W_{i,k}=P[0.1\cdot k\cdot x, 0.1\cdot k\cdot x+1.2x]$ of $P$ (which we call a \emph{window}).
An important property is that any vertex $v$ of $\P$ is contained in $O(1)$ windows for a specific $D_i$ and in $O(\log n)$ windows in total.

\begin{claim}[Properties of the subsets]\label{claim:properties}
    \,
    \begin{itemize}
        \item Any vertex $v$ is contained in $O(\log n)$ windows.
        \item Let $m_{i,k}=W_{i,k}[\ceil{|W|/2}]$ be the middle of $W_{i,k}$.
        For any detour $(v_{last},v)\in D_{i,k}$ we have $v<_P m_{i,k}<_P v_{last}$.
    \end{itemize}

\end{claim}

Let $H_{i,k}$ be the (unweighted) graph that is composed of $W_{i,k}$ enriched with an edge $(v_{last},v)$ for  every detour $(v_{last},v)\in D_{i,k}$.
For two detours $(u,v)$ and $(x,y)$ with $u>x$ we say that the detours cross each other, and denote $(u,v)\cross(x,y)$ if $u>_Px>_Pv>_Py$.
The following claim is used here only for $H=H_{i,k}$, we prove a stronger claim, regarding subgraphs of $H_{i,k}$, which will be useful later.

\begin{claim}\label{claim:2-intervals}

    Let $H$ be a subgraph of $H_{i,k}$.
    If there is a $u$-to-$v$ path $S$ in $H$ that does not touch any vertex of $P$ before $v$, then there is a $u$-to-$v$ path in $H$ that does not touch any vertex of $P$ before $v$ and uses either one or two detours. If it uses two detours, then those detours must  cross.
\end{claim}
\begin{proof}
    Let $e_1=(b_{last},b)$ be a detour with the latest $b_{last}$ in $S$, and let $e_2=(a_{last},v)$ be the last detour used by $S$.
    If $b=v$ then the path that goes from $u$ to $b_{last}$ on $P$ and then on $e_1$ concludes the proof.
    Similarly, if $b_{last} = a_{last}$ then the path that goes from $u$ to $a_{last}$ and uses $e_2$ concludes the proof.
    Otherwise, we have $v<_P b$ and $a_{last} <_P b_{last}$.
    By \cref{claim:properties} we have $b<_P m_{i,k}$ and $m_{i,k}<_P a_{last}$.
     Therefore, $v<_Pb<_P a_{last}<_P b_{last}$ and so the path that goes from $u$ to $b_{last}$, on $P$ uses $e_1$, then goes from $b$ to $a_{last}$ on $P$, and then uses $e_2$ contains exactly two crossing detours.
    \end{proof}

The following claim is a direct consequence of \cref{claim:2-intervals}.

\begin{claim}\label{claim:approx2}
    Let $v<_Pu$ be two vertices in $W_{i,k}$.
    Then, if $v$ is reachable from $u$ in $H_{i,k}$ with a path that does not touch any vertex before $v$ then there is a path in $G$ from $u$ to $v$ that does not touch any vertex before $v$ and is of length at most $2\beta$.
\end{claim}
\begin{proof}
    By \cref{claim:2-intervals} there is a path $S$ from $u$ to $v$ that uses at most two detours of $H_{i,k}$.
    By definition, each such detour corresponds to a path of length at most $\beta$ in $G$.
    Moreover, since the total length of $\P$ is $0$, we have that by replacing each edge in $S$ with its corresponding path in $G$, we get a path in $G$ of length at most $2\beta$ that does not touch any vertex in $\P$ before $v$.
\end{proof}

Using \cref{claim:approx2}, we can use the auxiliary reachability procedure of Section~\ref{sec:auxiliary} for every set $D_{i,k}$ in order to obtain our $2$-approximation:
\hlgray{Every $v\in P$ for every $i,k$ such that $v\in W_{i,k}$, stores its reachability label in $H_{i,k}$, as well as the minimum vertex $a$ such that there is a detour $(a_{last},a)\in D_{i,k}$ with $a\le_P v\le_P a_{last}$, if exists.}
Recall that by \cref{claim:properties} every vertex of $P$ is contained in $O(\log n)$ windows so the size of $v$'s labels is $\tilde O(1)$.

We now explain how to use the information stored in the labels to complete our task.
Let $f$ and $b$ be two vertices with $f<_Pb$, and let $p=\dfirst{G\setminus (P_1\cup\{f\})}{b}{P_2}{\beta}$.
Consider the detour $(p_{last},p)$.
This detour appears in some set $D_{i,k}$
(if $(p_{last},p)$ appears in more than one $D_{i,k}$, consider one arbitrary such set).
It is clear that $p\le_P b\le_P p_{last}$ and therefore $b\in W_{i,k}$ .
If $f$ is not in this window, then  the label of $b$ stores $p$ explicitly as the minimum vertex that can be reached by a detour starting after $b$.
Otherwise, if $f$ is in the window, then we have the labels of the auxiliary reachability procedure of Section~\ref{sec:auxiliary} for both $f$ and $b$, which means we can find the first $p'>_P f$ that is reachable from $b$ in $H_{i,k}$.
Notice that $p'\le_P p$ since $(p_{last},p)\in D_{i,k}$ implies that $p$ is reachable from $b$ in $H_{i,k}$.
It follows from \cref{claim:approx2} that $p'$ is a $\ddfirst{G\setminus (P_1\cup\{f\})}{b}{P_2}{\beta}{\beta}$.

\paragraph{A labeling scheme for any \boldmath$\eps$.}
We are now ready to solve the problem for any $\eps$.
Let $A_1=(u_1,v_1)$ and $A_2=(u_2,v_2)$ with $u_1>u_2$ be two detours in $D$.
If $A_1\cross A_2$, and every path from $u_1$ to $v_2$ that does not touch any vertex before $v_2$ is of length more than $(1+\eps)\beta$, then we say that $A_1$ bad-crosses $A_2$ and denote $A_1\badcross A_2$.
Otherwise, we say that $A_1$ good-crosses $A_2$ and denote $A_1\goodcross A_2$.
Notice that if $A_1$ does not cross $A_2$, then $A_1\goodcross A_2$.

From now on, we focus on a single $D_{i,k}$, and use  the notation $D=D_{i,k}$, $m=m_{i,k}$, $W=W_{i,k}$ and $H=H_{i,k}$.

\begin{lemma}\label{lem:partition_D}
   $D$ can be partitioned into $t=O(1/\eps^4)$ subsets $D_1, D_2, \dots, D_t$ such that if $A_1,A_2\in D_i$ then $A_1\goodcross A_2$.
\end{lemma}
\begin{proof}
We present an algorithm that partitions $D$ (this is exactly the same algorithm as in the proof of \cref{lem:partition_B}, but the definition of a bad-cross in this context is different).
The algorithm runs in phases, creating the  subset $D_i$ in the $i$'th phase.
At the beginning of every phase, we initialize a set $D_{good}= \emptyset$.
During a phase, we examine the detours $(u,v)$ of $D$ in decreasing order of their first coordinate $u$.
When examining a detour $A$, we check if $A$ good-crosses all detours in $D_{good}$, and if so we add $A$ into $D_{good}$ and remove $A$ from $D$.
At the end of the $i$'th phase, we set $D_i=D_{good}$ and append $D_i$ to the partition and if $D=\emptyset$ the algorithm terminates.
Clearly $|D_{good}|\ge 1$, so the algorithm halts with a partition of the initial $D$.

By definition, for a detour $A_t=(u_t,v_t)$ in $D_t$, since $A_t\notin D_{t-1}$, there is a detour $A_{t-1}=(u_{t-1},v_{t-1})\in D_{t-1}$ with $u_{t-1}>u_t$ and $A_{t-1}\badcross A_t$.
Given $A_{i+1}$ (with $i\ge 1$) we define $A_{i}=(u_i,v_i)$ in a similar way.
Thus, we have a sequence $A_1,A_2, \dots, A_t$ of detours such that $A_i\in D_i$ and $A_i\badcross A_{i+1}$.

\begin{figure}[htb]
\begin{center}
\includegraphics[width=0.5\textwidth]{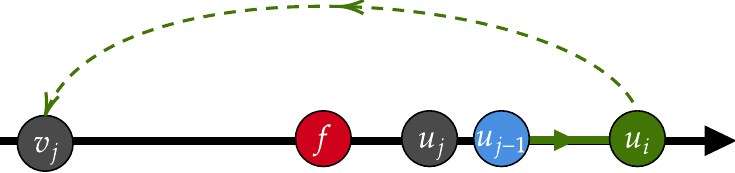}
\caption{An illustration of the proof of \cref{claim:bad_cross}.
The green dashed arrow represents a path of length at most $(1+\eps)\beta$ from $u_i$ to $v_j$ that does not use a vertex before $v_j$ on $P$.
Such a path exists due to the assumption that $A_i \goodcross A_j$.
By starting with $P[u_{j-1},u_i]$ (displayed as a thick green portion of $P$), we obtain a path of the same length from $u_{j-1}$ to $v_j$ that does not use any vertex before $v_j$, a contradiction to $A_j \badcross A_{j-1}$.
\label{fig:hard_bad_cross}
}
\end{center}
\end{figure}

\begin{claim}\label{claim:bad_cross}
For any $1\le i < j \le t$ we have $A_i\badcross A_j$.
\end{claim}
\begin{proof}
By \cref{claim:properties}, $m$ is a position such that for any detour $(u,v)\in D$ we have $u>_Pm>_Pv$.
Recall that $A_{k-1}\badcross A_{k}$ means in particular that $A_{k-1}\cross A_{k}$ thus, $u_{k-1}>_P u_k>_P m>_P v_{k-1}>_P v_k$.
By simple induction $u_1>_P u_2>_P \ldots>_P u_t>_P m>_P v_1>_P v_2>_P \ldots>_P v_t$.
Thus, $A_i\cross A_j$.
Assume by contradiction that $A_i\goodcross A_j$, then there is a path $S$ from $u_i$ to $v_j$ that does not touch any vertex before $v_j$ and is of length is at most $(1+\eps)\beta$.
The path $S'=P[u_{j-1},u_i]\cdot S$ (see \cref{fig:hard_bad_cross}) is a path from $u_{j-1}$ to $v_j$ that does not use any vertex before $v_j$ and $\len(S')=0+\len(S)\le (1+\eps)\beta$, a contradiction to $A_{j-1}\badcross A_j$.
\end{proof}

\begin{figure}[htb]
\begin{center}
\includegraphics[width=0.5\textwidth]{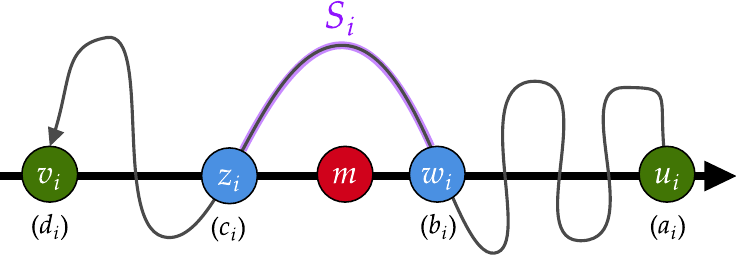}
\caption{A demonstration of $w_i, z_i$ and $S_i$. $P_i$ is black $u_i$-to-$v_i$ path.
\label{fig:hard}
}
\end{center}
\end{figure}

For every detour $A_i = (u_i,v_i)$ in the sequence, we define the vertices $w_i$ and $z_i$ as follows (see \cref{fig:hard}).
Let $P_i$ be the path corresponding to the detour $A_i$ in $G$.
The vertex $z_i$ is defined to be the first vertex visited by $P_i$ that is on $P$ and has $z_i <_P m$ (notice that $z_i$ is well defined because $v_i <_P  m$ is in $P_i$).
The vertex $w_i$ is the vertex of $P$ that precedes $z_i$ in $P_i$ (note that $w_i$ is well defined because $u_i >_P  m >_P  z_i$ is the first vertex of $P_i$).
Additionally, notice that $w_i \ge m$ due to the minimality of $z_i$ in $P_i$.
We also define $S_i = P_i[w_i,z_i]$, and observe that $S_i$ is internally disjoint from $P$.
If $S_i$ is to the left (resp. right) of $P$, we say that $(w_i,z_i)$ is a left (resp. right) pair.
Consider the sequence $W = (w_1,z_1),(w_2,z_2),\ldots ,(w_t,z_t)$.

By Erdős–Szekeres~\cite{ES35}, there is a subsequence $W'$ of $W$ with $|W'| \ge \sqrt{t}$ such that the $w_i$'s of the pairs in $W'$ are monotone (either non-increasing or non-decreasing).
By applying Erdős–Szekeres again, we have that there is a subsequence $W''$ of $W'$ such that the $z_i$ values are monotone and $|W''| \ge  \sqrt[4]{t}$.
Finally, each of the pairs in $W''$ is either a left or a right pair.
We assume w.l.o.g that most of the pairs in $W''$ are left pairs.
It follows that there is a subsequence $C$ of $W''$ such that both $w_i$ and $z_i$ values are monotone, all pairs in $C$ are left pairs, and $|C| \ge |W''|/2 \ge \frac{\sqrt[4]{t}}{2}$.
To conclude the proof of \cref{lem:partition_D}, it remains only to prove that $|C| \le \frac{1}{\eps}+1$, since then $\frac{1}{\eps}+1 \ge \frac{\sqrt[4]{t}}{2}$.

\begin{claim}
    $|C|\le 1/\eps+1$.
\end{claim}
\begin{proof}
    Let $C= (b_1,c_1),(b_2,c_2), \ldots (b_{|C|},c_{|C|})$.
    For every $i\in [|C|]$ such that $(b_i,c_i) = (w_j,z_j)$ we denote $a_i = u_j$, $d_i = v_j$ (see \cref{fig:hard}).
    From now on, we abuse notation by using $A_i$, $P_i$ and $S_i$ to refer to $A_j$, $P_j$ and $S_j$ respectively.

    Recall that both $b_i$ and $c_i$ are monotone.
    We distinguish between two cases:
    \paragraph{Case 1: either the \boldmath$b_i$'s or the \boldmath$c_i$'s are non-decreasing.}
    We assume the $c_i$'s are non-decreasing (the proof for the $b_i$'s is symmetric).
    In this case, we have
    $a_1\le_P  a_2 \le_P  \dots \le_P  a_{|C|}<m\le_P  c_{|C|}\le_P  c_{|C|-1}\le_P \dots\le_P  c_1\le_P  d_1\le_P  d_2\le_P  \dots\le_P  d_{|C|}$.

    We will prove by induction that for every $i\in[|C|]$ we have $\len(P_i[a_i,c_i])\le (1-(i-1)\eps)\beta$.
    For $i=1$ the claim follows since $\len(P_1[a_1,c_1])\le\len(P_1)\le \beta$.
    We assume the claim holds for $i$, and prove it holds for $i+1$.
    Notice that (by \cref{claim:bad_cross}) for any $i< |C|$ we have that the $A_i\badcross A_{i+1}$.
    This means that the length of any path from $a_i$ to $d_{i+1}$ that does not use any vertex of $P$ before $d_{i+1}$ is more than $(1+\eps)\beta$.
    In particular, for the path $S= P_i[a_i,c_i] \cdot P[c_i,c_{i+1}] \cdot P_{i+1}[c_{i+1},d_{i+1}]$ we have $\len(S)>(1+\eps)\beta$.
    Due to the induction hypothesis, $\len(P_i[a_i,c_i]) \le (1-(i-1)\eps)\beta$.
    Therefore, we have $(1-(i-1)\eps) \beta + 0 + \len(P_{i+1}[c_{i+1} ,d_{i+1}]) \ge (1+\eps)\beta$ which leads to $\len(P_{i+1}[c_{i+1} ,d_{i+1}]) \ge i \eps r$.
    Finally, we have $\len(P_{i+1}[a_{i+1},c_{i+1}]) = \len(P_{i+1})-\len(P_{i+1}[c_{i+1},d_{i+1}]) \le \beta -i\eps \beta = (1-i\eps) \beta$ as required.
    If $|C|>1/\eps+1$ we get that $\len(P_{|C|}[a_{|C|},c_{|C|}])$ is negative, a contradiction.

    \paragraph{Case 2: both the \boldmath$b_i$'s and the \boldmath$c_i$'s are decreasing.}
    We make the following claim:
    For every $i\in [|C|]$ there is a path $E_i$ from $a_i$ to $c_1$  with $\len(E_i)\le (1-(i-1)\eps)r$.
    Moreover, $b_i$ is on $E_i$ and $E_i[b_i,c_1]$ is internally disjoint from $P$ and is to the left of $P$.
    The claim  holds for $i=1$ by setting $E_1=P_1[a_1,c_1]$.
    We assume the claim holds for $i$, and prove it for $i+1$.

    Recall that $P_{i+1}$ is the path corresponding to the detour $(a_{i+1},d_{i+1})$ and $b_i>_P b_{i+1}\ge_P m >_P c_1>_P c_{i+1}$.
    Since both $E_i[b_i,c_1]$ and $S_{i+1}=P_{i+1}[b_{i+1},c_{i+1}]$ are internally disjoint from $P$ and go to the left of $P$, they must intersect at some vertex $z\notin P$.

    Since $A_i\badcross A_{i+1}$, any path from $a_i$ to $d_{i+1}$ that does not use any vertex before $d_{i+1}$ is of length more than $(1+\eps)\beta$.
    In particular, for the path $S=E_i[a_i,z]\cdot P_{i+1}[z,d_{i+1}]$ we have $\len(S)>(1+\eps)\beta$.
    Let $S'=P_{i+1}[a_{i+1},z]\cdot E_i[z,c_1]$. Notice that $S'$ is a path from $a_{i+1}$ to $c_1$ that goes through $b_{i+1}$ and $S'[b_{i+1},c_1]$, is internally disjoint from $P$, and is to the left of $P$.
    It remains to show that $\len(S')\le (1-i\eps)\beta$.
    By the induction hypothesis, $\len(E_i)\le (1-(i-1)\eps)\beta$.
    By definition of a detour we have $\len(P_{i+1})\le \beta$.
    Notice that $\len(E_i)=\len(E_i[a_i,z])+\len(E_i[z,c_1])$ and $\len(P_{i+1})=\len(P_{i+1}[a_{i+1},z])+\len(P_{i+1}[z,d_{i+1}])$.
    Therefore $\len(E_i)+\len(P_{i+1})=\len(S)+\len(S')$.
    Since $\len(S)>(1+\eps)\beta$ we have $\len(S')\le \beta+(1-(i-1)\eps)\beta - (1+\eps)\beta= (1-i\eps)\beta$, as required.
    If $|C|>1/\eps+1$, then there is a path from $a_t$ to $c_1$ with negative length, a contradiction.
\end{proof}

This concludes the proof of \cref{lem:partition_D}
\end{proof}

Let $\mathcal D= D_1,D_2,\dots,D_t$ be the partition of $D$ obtained by \cref{lem:partition_D}.
For $D_i\in\mathcal D$, let $H_{D_i}$ be the graph that contains $W$ (the window corresponding to $D$) and all the detours of ${D_i}$ (each detour $(v_{last},v)\in {D_i}$ corresponds to an edge from $v_{last}$ to $v$ in $H_{D_i}$).
The following claim is an extention of \cref{claim:approx2} for general $\eps$.

\begin{claim}\label{claim:approx_eps}
    Let $D_i\in\mathcal D$ be a set and let $v<u$ be two vertices in $W$.
    Then, if $v$ is reachable from $u$ in $H_{D_i}$ with a path not touching any vertex before $v$, then there is a path in $G$ from $u$ to $v$ that does not touch any vertex before $v$ and is of length at most $(1+\eps)\beta$.
\end{claim}
\begin{proof}
    By \cref{claim:2-intervals} there is a path $S$ in $H_{D_i}$ from $u$ to $v$ that uses either a single detour, or two crossing detours.
    In the case where $S$ uses a single detour, the corresponding path in $G$ is trivially of length at most $\beta$.
    In the other case, let $A_1 \cross A_2$ be the two detours.
    Since $A_1,A_2\in {D_i}$, by definition of the partition it must be that $A_1\goodcross A_2$, and therefore there is a path in $G$  of length at most $(1+\eps)\beta$ from $u_1$ to $v_2=v$ where $A_1=(u_1,v_1)$ and $A_2=(u_2,v_2)$, as required.
\end{proof}

\begin{lemma}\label{lem:labelPztoPz}
    There exists a labeling scheme $\LabelPztoPz_\beta$ with labels of size $\Opoly$.
\end{lemma}
\begin{proof}
We are using the auxiliary reachability procedure of Section~\ref{sec:auxiliary} in our labeling.
\hlgray{For every $D=D_{i,k}$ for every subset of the partition ${D_j}\in \mathcal D$, each vertex in $W_{i,k}$ stores the auxiliary reachability label of $v$ in $H_{D_i}$.}
Additionally, \hlgray{every vertex $v$ in $W_{i,k}$ stores the minimum vertex $a$ such that there is a detour $(a_{last},a)\in D$ with $a\le_P v\le_P a_{last}$, if exists.}
Recall that by \cref{claim:properties} every vertex of $P$ is contained in $O(\log n)$ windows, and by \cref{lem:partition_D} every set is partitioned into $O(1/\eps^4)$ subsets,  so the accumulated size of  labels kept in $v$ is $\Opoly$.

We now explain how to use the information stored in the labels to complete our task.
Let $f$ and $b$ be two vertices on $\P$ with $f<_Pb$, and let $p=\dfirst{G\setminus (P_1\cup\{f\})}{b}{P_2}{\beta}$.
Consider the detour $(p_{last},p)$.
This detour appears in some set $D_{i,k}$ ($i$ and $k$ can be determined by the size of $(p_{last},p)$ and the index of $p$).
(If $(p_{last},p)$ appears in more than one $D_{i,k}$, consider one arbitrary such set.)
Moreover, there exists a subset $D_j$ of $D_{i,k}$, obtained via \cref{lem:partition_D} with $(p_{last},p)\in D_j$.
It is clear that $p\le_P b\le_P p_{last}$ and therefore $b$ is in $W_{i,k}$.
If $f\in W_{i,k}$, then  the label of $b$ stores $p$ explicitly as the minimum vertex that can be reached by a detour starting after $b$.
Otherwise, if $f\in W_{i,k}$, then we have the auxiliary reachability labels of both $f$ and $b$ in $H_{D_j}$, which means we can find the first $p'>_Pf$ that is reachable from $b$ in $H_{D_i}$.
Notice that $p'\le_P p$ since $(p_{last},p)\in D_j$ implies that $p$ is reachable from $b$ in $H_{D_j}$.
It follows from \cref{claim:approx_eps} that $p'$ is an $\ddfirst{G\setminus (P_1\cup\{f\})}{b}{P_2}{\beta}{\eps\beta}$.
\end{proof}

\section{The $\LabelatfoP$ Labeling  (Proof of \cref{lem:atfoplabel})}\label{sec:atfop}

In this section we prove \cref{lem:atfoplabel}.
The setting in this section is as follows.
$G$ (for the ease of presentation we use here $G$ as the name of the graph) is a graph, $P$ and $P'$ are two $0$-length paths and $\alpha,\eps\in\mathbb R^+$.
Our goal is to develop a labeling scheme such that
given the labels of a vertex $a$ on $P'$ and a vertex $f$ not in $P'\cup P$, one can retrieve the index on $P$ of some vertex $b\in P$ such that $b$ is an $\ddfirst{\Gf}{a}{P}{\alpha}{\eps\alpha}$.

Similarly to the case of reachability, we shall define a set of canonical paths that facilitate the distribution of information about avoiding failed vertices. The choice of these canonical paths is more elaborate. The canonical paths are no longer disjoint, but an important feature we maintain is that every vertex $f$ of $G$ lies on a small number of canonical paths, so we can afford to store in $f$'s label information about these paths when $f$ fails.
We work with an estimation $\gamma$ that is a multiple of $\eps \alpha$, such that $\gamma-\eps \alpha
<\dist_G(a,b)\le\gamma$.
Notice that, $\dfirst{G}{P}{a}{\gamma} \le_P b$ since $\gamma\ge\dist_G(a,b)$.
We define $L_\gamma=\{(u,\dfirst{G}{u}{P}{\gamma})\mid u\in P'\text{ and }\dist_G(u,\dfirst{G}{u}{P}{\gamma})>\gamma-\eps \alpha \}$.

We define a set $C = C_{\gamma}$ of canonical paths (which we call $\gamma$-legitimate paths) as follows.

\begin{restatable}[$\gamma$-legitimate set of paths]{definition}{canoniclpathdef}\label{def:canonicalpaths}
A set $C$ of paths in $G$ is a \emph{$\gamma$-legitimate set of paths for $L$} if it has the following properties:
\begin{itemize}
    \item For each $S\in C$ we have $\len(S)\le (1+\eps)\gamma$ and $S=S_1\cdot S_2$ such that both $S_1$ and $S_2$ are shortest paths.
    \item For every $(c,c_{first})\in L$  let $u \ge_{P'} c$ be the first vertex of $P'$ with a path $(u \rightsquigarrow u')\in C$.
    If there are several such paths, let $u'$ be the first (in $P$) among all vertices $u'$.
    It holds that $u$ exists, and $u'$ is a $\ddfirst{G}{c}{P}{\gamma}{\eps\gamma}$.
    \item For every vertex $f$ of $G$ there are $O(1/\eps)$ paths in $C$ going through $f$.
\end{itemize}
\end{restatable}
We call the paths in $C$ \emph{canonical}. 
For every $c\in \Psf$ let $u \ge_{P'} c$ be the first vertex of $P'$ with a path $(u \rightsquigarrow u')\in C$ (if exists).
If there are several such paths, let $u'$ be the first (in $P$) among all values $u'$.
If $u$ exists, we say that $(u\rightsquigarrow u')$ is \emph{the canonical path of $c$}.

In \cref{sec:canonical} we prove the following lemma.
\begin{restatable}{lemma}{canonicalexistslem}\label{lem:canonical_exists}
    For every $G,P',P,\eps,\gamma, L$ there is a $\gamma$-legitimate set of paths.
\end{restatable}

Denote such a set by $C=C_{\gamma}$.
We partition $\Psf$ into maximal contiguous intervals of vertices $c$ with $(c,\cdot)\in L_\gamma$ that have the same canonical path.
We call these intervals the canonical $\gamma$-intervals of $\Psf$ and $\P$.
For every such interval $\I$ of vertices $p$ with the same canonical path $S_\I^\gamma = (x \rightsquigarrow y)$, we
call $S_\I^\gamma$ the canonical path of interval $\I$.

In \cref{sec:onlysuffix}, we prove the following.
\begin{restatable}{lemma}{onlysuffixlem}\label{lem:onlysuffixLabeling}
    For a graph $G$, two paths $P$ and $P'$ of length $0$, a path $A=(u\rightsquigarrow u')$ from the last vertex of $P'$ to $u'$ which is the first vertex of $P$ and numbers $\alpha, \eps \in \mathbb{R}^+$ such that:
    \begin{enumerate}
    \item $\len(A) \le (1+\eps)\alpha$.
    \item $A= A_1 \cdot A_2$ such that $A_1$ and $A_2$ are shortest paths in $G$.
    \end{enumerate}
    There exists a labeling scheme $\Labelonlysuffix=\Labelonlysuffix_{G,P',A,P,\alpha,\eps}(v)$    such that given the labels of two vertices $a \in P'$ and $f \in A \setminus P'$, one can obtain a vertex $b\in P$ such that $\dist_{\Gf}(a,b) \le (1+\eps) \alpha$ and $b \le_P \dfirst{G \setminus A[\ldots f]}{a}{P}{\alpha}$ or conclude that $\dfirst{G \setminus A[\ldots f]}{a}{P}{\alpha}=null$.
    The size of each label is $\Opoly$.
\end{restatable}

We are now ready to prove \cref{lem:atfoplabel}, which we restate here.
\atfoplabellem*
\begin{proof}

Let $\eps'=\eps/2$ be an approximation factor and let $\Gamma=\{i\eps'\alpha\mid i\in[0,\ceil{1/\eps'}]\}$ be the set of multiples of $\eps'\alpha$.
For every $\gamma\in\Gamma$ let $C_\gamma$ be a set of $\gamma$-canonical paths, computed with respect to $\eps'$ (such a set exists due to \cref{lem:canonical_exists}).

\hlgray{
For every vertex $v\in P'$ and for every $\gamma\in\Gamma$  such that $(v,\cdot)\in L_\gamma$ the label of $v$ stores: The endpoints $(u_{\gamma},u'_{\gamma})$ of its canonical path $S_{I_v}\in C_\gamma$. In addition $v$ stores $\Labelonlysuffix_{G,P'[1,u_\gamma],S_{I_v},P[u'_\gamma,|P|],\alpha,\eps'}(v)$}

\hlgray{For every vertex $v\in G\setminus P'$ and for every $\gamma\in\Gamma$ and for every $S=(u\rightsquigarrow u')\in C_\gamma$ such that $v\in S$ the label of $v$ stores: $\dfirst{G\setminus\{v\}}{u}{P}{\alpha+2\eps'\alpha}$ and $\Labelonlysuffix_{G,P'[1,u],S,P[u',|P|],\alpha,\eps'}(v)$.}

\paragraph{Size.}
By \cref{def:canonicalpaths} every $v\in G$ is on $O(1/\eps')$ canonical paths.
Moreover, for every canonical path, the size of $\Labelonlysuffix(v)$ is $\Opoly$ by \cref{lem:onlysuffixLabeling}.
Thus, the total size of every $\LabelatfoP(v)$ is $\Opoly$.

\paragraph{Decoding.}
Given the labels of $a\in P'$ and $f\in G\setminus P'$  one can compute the index on $P$ of some vertex $b$ such that $b$ is an $\ddfirst{G\setminus f}{a}{P}{\alpha}{\eps\alpha}$ as follows.
For every $\gamma\in\Gamma$ such that $(v,\cdot)\in L_\gamma$  let $S_\gamma=(u_{\gamma},u'_{\gamma})$ be the canonical path of $a$ in $C_\gamma$.
If $f\notin S_{\gamma}$ (which can be obtained from the label of $f$) then $b_\gamma=u'_{\gamma}$.
Otherwise, let $b_{\gamma}^1$ be a vertex on $P$ such that $\dist_{\Gf}(a,b_{\gamma}^1) \le (1+\eps') \alpha$ and $b_{\gamma}^1 \le_P \dfirst{G \setminus S_\gamma[\ldots f]}{a}{P}{\alpha}$ or set $b_{\gamma}^1=null$ if $\dfirst{G \setminus S_\gamma[\ldots f]}{a}{P}{\alpha}=null$, by using $\Labelonlysuffix$.
Let $b_\gamma^2$ be $\dfirst{\Gf}{u_{\gamma}}{P}{\alpha+2\eps'\alpha}$ (stored in the label of $f$).
Let $b_\gamma=\min_{\le_P}\{b^1_\gamma,b^2_\gamma\}$.
Finally, we return $b=\min_{\le_P}\{b_\gamma\mid \gamma\in\Gamma\}$.

\paragraph{Correctness.}
First, we show that $\dist_{\Gf}(a,b)\le (1+2\eps')\alpha =(1+\eps)\alpha$.
\begin{itemize}
    \item If $b=b_\gamma=u'_{\gamma}$ it must be that $f\notin S_\gamma$ and therefore by \cref{def:canonicalpaths} (first property)  we have $\dist_{\Gf}(a,u'_{\gamma})\le\len(P[a,u_\gamma]\cdot S_\gamma)\le (1+\eps')\gamma\le (1+\eps)\alpha$.
    \item If $b=b^1_\gamma$ for some $\gamma$, then by \cref{lem:onlysuffixLabeling} we have $\dist_{\Gf}(a,b^1_\gamma)\le (1+\eps')\alpha\le(1+\eps)\alpha$.
    \item Otherwise, if $b=b^2_\gamma=\dfirst{\Gf}{u_\gamma}{P}{(1+2\eps')\alpha}$, by definition $\dist_{\Gf}(a,b^2_\gamma)\le (1+2\eps')\alpha=(1+\eps)\alpha$.
\end{itemize}

It remains to prove that $b\le_P \dfirst{\Gf}{a}{P}{\alpha}$.
Let $b'=\dfirst{\Gf}{a}{P}{\alpha}$ and let $R$ be a shortest path from $a$ to $b'$ in $\Gf$.
Let $\gamma=x\in\Gamma\mid x\ge \dist_G(a,b')\}$,  and let $S_\gamma=(u_\gamma,u'_\gamma)$ be the canonical path of $a$ in $C_\gamma$.
    We distinguish between three cases:
    \begin{itemize}
        \item If $f\notin S_\gamma$,
        by \cref{def:canonicalpaths} (second property) $b_\gamma=u'_{\gamma}\le_P\dfirst{G}{a}{P}{\gamma}\le_P b'$ by definition of $\gamma$.

        \item If $R\cap S_{\gamma}[u_\gamma,f]=\emptyset$,  then $b^1_\gamma$ is a vertex on $P$ such that $b_{\gamma}^1 \le_P \dfirst{G \setminus S_\gamma[u_\gamma, f]}{a}{P}{\alpha}\le_P b'$ since there exists a path ($R$) in $G \setminus S_\gamma[u_\gamma, f]$ from $a$ to $b'$ of length at most $\alpha$.

        \item Otherwise (if $R\cap S_{\gamma}[u_\gamma,f)\ne\emptyset$) we have that $b^2_\gamma\le_P b'$ by the following claim.
\end{itemize}

      \begin{claim}\label{clm:useS1}
            $b^2_\gamma=\dfirst{\Gf}{u_{\gamma}}{P}{\alpha+2\eps'\alpha}\le_P \dfirst{\Gf}{a}{P}{\alpha}=b'$. 
        \end{claim}
        \begin{proof}
            Recall that $(a,\dfirst{G}{a}{P}{\gamma})\in L_\gamma$ which means that $\dist_G(a,\dfirst{G}{a}{P}{\gamma})>\gamma-\eps\alpha$.
            Assume by contradiction that $b'<_P b^2_\gamma$.
            Let $z$ be some vertex in $R\cap S_\gamma[u_\gamma,f)$.
            Let $\ell=\len (S_\gamma[u_\gamma,z])$.
            It must be that $\len(R[u_\gamma,z])\le \ell-2\eps'\alpha$ since otherwise, $S_\gamma[u_\gamma,z]\cdot R[z,b']$ is a path of length at most $\alpha+2\eps'\alpha$, which contradicts the minimality of $\dfirst{\Gf}{u_{\gamma}}{P}{\alpha+2\eps'\alpha}$ on $P$.
            Thus, $\len(R[u,z])\le \ell-2\eps'\alpha$.
            Now, consider the path $R[a,z]\cdot S_\gamma[z,u'_\gamma]$ (in $G$).
            This is a path of length at most $(1+\eps')\gamma-2\eps'\alpha\le \gamma -\eps'\alpha$ from $a$ to $u'_\gamma$ in $G$.
            Therefore, $u'_\gamma\ge_P \dfirst{G}{a}{P}{\gamma}$.
            By \cref{def:canonicalpaths} (second property) it must be that $u'_\gamma\le_P \dfirst{G}{a}{P}{\gamma}$ which leads to $u'_\gamma= \dfirst{G}{a}{P}{\gamma}$.
            We have shown $\dist_G(a,\dfirst{G}{a}{P}{\gamma})<\gamma-\eps\alpha$ this contradicts $(a,\dfirst{G}{a}{P}{\gamma})\in L_\gamma$.
        \end{proof}
In each case we have shown that $b_\gamma\le_P b'$, and therefore by the minimality of $b$ on $P$, we have $b\le_P b'$.
\end{proof}

\subsection{Selecting canonical paths}\label{sec:canonical}
In this section we solve the following problem.
We are given a weighted directed planar graph $G$ with two paths $P'$ and $P$ where the total length of each path is $0$.
In addition, we are given a length bound $\gamma$, and an approximation factor $\eps>0$.
Finally, we are given a set $L\subseteq P' \times P$ of pairs $(u,u_{first})$ where $u\in P'$ and $u_{first}=\dfirst{G}{u}{P}{\gamma}$.
Notice that $L$ has the following properties:
\begin{enumerate}
    \item \textbf{Uniqueness:} For every $u \in P'$, there is at most one pair $(u,u_{first})$ in $L$.
    \item \textbf{Proximity:} Each pair $(u,u_{first})\in L$ has $\dist_G(u,u_{first})\le \gamma$.
    \item \textbf{Monotonicity:} For every two pairs $(u,u_{first}),(v,v_{first})\in L$, we have that if $u <_{P'} v$ then $u_{first} \le_P v_{first}$.
\end{enumerate}

Our goal in this problem is to find a  \emph{$\gamma$-legitimate set of canonical paths for $L$} as defined in \cref{def:canonicalpaths}:
\canoniclpathdef*
We call the paths in $C$ \textit{canonical} paths. 

The rest of this section is dedicated to the proof of \cref{lem:canonical_exists}.
\canonicalexistslem*

For every $(u,u_{first})\in L$ let $A_u$ be a shortest path from $u$ to $u_{first}$ (notice that $\len(A_u)\le \gamma$).
For every $u<_{P'}v$ if $A_u\cap A_v\ne \emptyset$ we say that $A_u$ crosses\footnote{A more natural terminology would be to say that $A_u$ intersects $A_v$. However, due to a similarity to \cref{sec:easy,sec:hard_prob} we prefer to use the same terminology of crossing here.} $A_v$ and denote $A_u\cross A_v$.
For every $A_u\cross A_v$ (with $u<_{P'} v$) let $z_{u,v}$ be the first vertex on $A_v$ which is also on $A_u$  (see \cref{fig:canonical-algorithm}).
Let $P_{u,v}=A_v[v,z_{u,v}]$ and $S_{u,v}= A_u[z_{u,v},u_{first}]$ be the prefix of $A_v$ until $z_{u,v}$ and the suffix of $A_u$ from $z_{u,v}$, respectively.
Let $A_{u,v}=P_{u,v}\cdot S_{u,v}$.
If $\len(A_{u,v})\le (1+\eps)\gamma$ we say that $A_u$ good-crosses $A_v$ and we denote $A_u\goodcross A_v$; otherwise $A_u$ bad-crosses $A_v$, and we denote $A_u\badcross A_v$.

 \begin{figure}[htb]
  \begin{center} \includegraphics[scale=0.7]{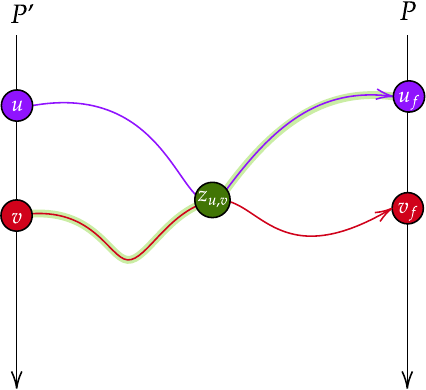}
  \caption{
  An illustration of the path $A_{u,v}=P_{u,v}\cdot S_{u,v}$.\label{fig:canonical-algorithm}}
   \end{center}
\end{figure}
\paragraph{Algorithm.}
Consider the pairs $(u,u_{first})$ in $L$ to be ordered according to the $<_{P'}$ order of $u$.
We construct the set $C$ using the following algorithm.
The algorithm initializes $C=\emptyset$ and a pair $(a,a_{first})$ to be the first pair in $L$.
The algorithm repeats the following procedure until the halting condition is met.
The algorithm finds the last pair $(b,b_{first})\in L$ such that $A_a\goodcross A_b$.
We distinguish between two cases:
If $a\ne b$, the algorithm inserts $A_{a,b}$ to $C$ and sets $(a,a_{first})\leftarrow (b,b_{first})$.
Otherwise, if $a=b$, the algorithm inserts $A_{a,b}=A_a$ to $C$ and sets $(a,a_{first})$ to be the pair in $L$ following $(a,a_{first})$; if there is no such pair, the algorithm halts and returns $C$.

We prove that $C$ is indeed a $\gamma$-legitimate set of canonical paths.

\begin{claim}\label{claim:canonical_first}
For each $S\in C$ we have $\len(S)\le (1+\eps)\gamma$ and $S=S_1\cdot S_2$ such that both $S_1$ and $S_2$ are shortest paths.
\end{claim}
\begin{proof}
    Let $S\in C$, by the algorithm $S=A_{a,b}$ for some vertices $a$ and $b$ such that $A_a \goodcross A_b$.
    By definition of $\goodcross$ we have $\len(S)=\len(A_{a,b})\le (1+\eps)\gamma$.
    Moreover, $S=A_{a,b}=P_{a,b}\cdot S_{a,b}$ and both $P_{a,b}$ and $S_{a,b}$ are shortest paths (as they are subpaths of shortest paths).
\end{proof}

\begin{claim}
For every $(c,c_{first})\in L$  let $u \ge_{P'} c$ be the first vertex of $P'$ with a path $(u \rightsquigarrow u')\in C$.
If there are several such paths, let $u'$ be the first (in $P$) among all vertices $u'$.
It holds that $u$ exists, and $u'$ is a $\ddfirst{G}{c}{P}{\gamma}{\eps\gamma}$.
\end{claim}
\begin{proof}
Let $(w,w_{first})$ be the last pair in $L$.
Notice that the algorithm necessarily iterates the last pair, and therefore adds a path starting in $w \ge_{P'} c$ to $C$.
Therefore, the vertices $u$ and $u'$ exist.
Let $A_{a,b}=(u \rightsquigarrow u')$.
In particular $u=b$ and $u'=a_{first}$.
Notice that $b\ge_{P'} c$.
Moreover it must be that $a\le_{P'} c$ since otherwise $C$ would also contain some other path $A_{x,a}$ or $A_{x,a-1} $ (where $(a-1,\cdot) \in L$ is the pair preceding $(a,a_{first})$ in $L$ and $x\in P'$) contradicting the minimality of $b$.
It follows from the monotonicity of $L$ that $u'=a_{first} \le_P c_{first}=\dfirst{G}{c}{P}{\gamma}$.

Consider the path $S'=P'[c,b]\cdot A_{a,b}$ from $c$ to $a_{first}$.
The path $S'$ has length $\len(S')=0 + \len(A_{a,b})\le (1+\eps)\gamma$ by \cref{claim:canonical_first}.
It follows that $a_{first}$ is an $\ddfirst{G}{c}{P}{\gamma}{\eps\gamma}$.
\end{proof}

It remains to prove the last property.
For a set $C' \subseteq \{A_{u} \mid (u,u_{first}) \in L  \}$ and a vertex $f$ we say that $C'$ is an \emph{$f$-bad} set if $f$ is on every path in $C'$, and every two different paths in $C'$ bad-cross each other.
We prove the following helpful lemma regarding $f$-bad sets.
\begin{lemma}\label{lem:fbadbound}
    Let $C'$ be an $f$-bad set for some vertex $f$ of  $G$.
    It must hold that $|C'| \le 1/\eps +1$
\end{lemma}
\begin{proof}
    Let $C' = A_{a_1},A_{a_2}, \ldots , A_{a_t}$ be the paths in $C'$ ordered such that $a_1 \le_{P'} a_2 \le_{P'} \ldots \le_{P'} a_{t} = a_{|C'|}$.
    We claim that for every $i\in [t]$, we have $\len(A_{a_i}[f,({a_i})_{first}])\le (1-(i-1)\eps)\gamma$.

    For $i=1$, this is true since $f$ is on $A_{a_1}$ and therefore $\len(A_{a_1}[f,({a_1})_{first}]\le \len(A_{a_1})\le \gamma$.
    We thus assume the claim holds for $i$, and prove it for $i+1$.

    Consider the path
    $S^*=A_{a_{i+1}}[a_{i+1},f] \cdot A_{a_i}[f,({a_i})_{first}]$.
    We first claim that $\len(S^*)\ge \len(A_{a_i,a_{i+1}})$.
    Recall that $z=z_{a_i,a_{i+1}}$ is the first vertex in $A_{a_{i+1}}$ that is also on $A_{a_i}$.
    Notice that $f\ge_{A_{a_{i+1}}}z$.
    There are two cases to consider:

 \begin{figure}[t]
  \begin{center} \includegraphics[scale=0.6]{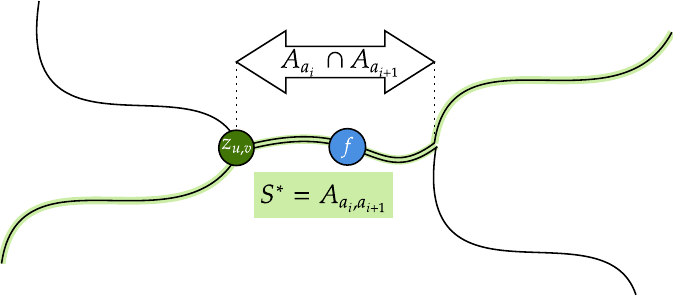}
  \caption{
  An illustration of the case $f\ge_{A_{a_i}}z$ where $S^*=A_{a_i,a_{i+1}}$.\label{fig:case_f>z}}
   \end{center}
\end{figure}
    \begin{itemize}
        \item if $f\ge_{A_{a_{i}}} z$ (see \cref{fig:case_f>z}):
        Since both $A_{a_i}$ and $A_{a_{i+1}}$ are shortest paths, we can assume that $A_{a_{i+1}}[z,f]=A_{a_i}[z,f]$.
        Therefore,
        \begin{align*}
        A_{a_i,a_{i+1}}&=A_{a_{i+1}}[a_{i+1},z]\cdot A_{a_i}[z,({a_i})_{first}]\\&= A_{a_{i+1}}[a_{i+1},z]\cdot A_{a_i}[z,f]\cdot A_{a_i}[f,({a_i})_{first}] \\&= A_{a_{i+1}}[a_{i+1},z]\cdot A_{a_{i+1}} [z,f]\cdot A_{a_i}[f,({a_i})_{first}]\\&=A_{a_{i+1}}[a_{i+1},f]\cdot A_{a_i}[f,({a_i})_{first}]=S^*
        \end{align*}
        Hence, $\len(S^*)=\len(A_{a_i,a_{i+1}})$.

 \begin{figure}[b]
  \begin{center} \includegraphics[scale=0.7]{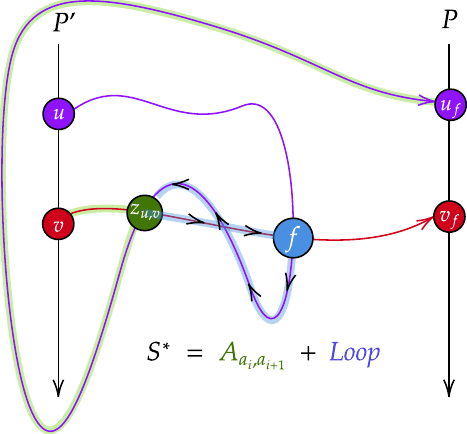}
  \caption{
  An illustration of the case $f\le_{A_{a_i}}z$.\label{fig:case_f<z}}
   \end{center}
\end{figure}
        \item if $f\le_{A_{a_{i}}} z$ (see \cref{fig:case_f<z}):
        In this case, there exists two subpaths, $A_{a_i}[f,z]$ and $A_{a_{i+1}}[z,f]$
        So, we have
        \begin{align*}
            S^*&=A_{a_{i+1}}[a_{i+1},f]\cdot A_{a_i}[f,({a_i})_{first}]
            \\&=A_{a_{i+1}}[a_{i+1},z]\cdot A_{a_{i+1}}[z,f]\cdot A_{a_i}[f,z]\cdot A_{a_i}[z,({a_i})_{first}]\\
            \text{and \quad}A_{a_i,a_{i+1}}&=A_{a_{i+1}}[a_{i+1},z]\quad\quad\quad\quad\quad\cdot\quad\quad\quad\quad\,\,\, A_{a_i}[z,({a_i})_{first}]
        \end{align*}

        Therefore, $\len(S^*)=\len(A_{a_i,a_{i+1}})+\len(A_{a_{i+1}}[z,f]\cdot A_{a_i}[f,z])\ge \len(A_{a_i,a_{i+1}})$.
    \end{itemize}

    Since $A_{a_i}$ and $A_{a_{i+1}}$ are two different paths in $C'$, it must be that $A_{a_i}\badcross A_{a_{i+1}}$ which means $\len(A_{a_i,a_{i+1}})>(1+\eps)\gamma$.
    Therefore, $\len(S^*)>(1+\eps)\gamma$.

    By definition, $\gamma\ge \len(A_{a_{i+1}})=\len(A_{a_{i+1}}[a_{i+1},f])+\len(A_{a_{i+1}}[f,({a_{i+1}})_{first}])$, which implies
    \begin{equation}\label{eq:len1}
        \len(A_{a_{i+1}}[f,({a_{i+1}})_{first}]) \le \gamma - \len(A_{a_{i+1}}[a_{i+1},f])
    \end{equation}

    By the induction hypothesis $\len(A_{a_i}[f,({a_i})_{first}])\le (1-(i-1)\eps)\gamma$ which implies
    \begin{equation}\label{eq:len2}
        0 \le   (1-(i-1)\eps)\gamma - \len(A_{a_i}[f,({a_i})_{first}])
    \end{equation}
    By summing \cref{eq:len1,eq:len2}, we obtain
    \begin{align*}
        \len(A_{a_{i+1}}[f,({a_{i+1}})_{first}])&\le \gamma+(1-(i-1)\eps)\gamma-\len(A_{a_{i+1}}[a_{i+1},f])-\len(A_{a_i}[f,({a_i})_{first}])
        \\&\le(2-(i-1)\eps)\gamma-\len(S^*)
        \\&<(2-(i-1)\eps)\gamma-(1+\eps)\gamma=(1-i\eps)\gamma
    \end{align*}
    as required.

    Assume by contradiction that $t= |C'| >1/\eps+1$, and deduce that $A_{a_t}[f,({a_t})_{first}]$ is a path of negative length, in contradiction.
    Thus, $|C'|\le 1/\eps+1$ as required.
\end{proof}

Equipped with \cref{lem:fbadbound}, we are finally ready to prove the last property.
\begin{lemma}\label{lem:fewcanonical}
For every vertex $f$ of $G$ there are $O(1/\eps)$ canonical paths going through $f$.
\end{lemma}
\begin{proof}
    Let $f$ be a vertex of $G$ and let $C_f=\{A_{a,b}\in C\mid f \text{ is on }A_{a,b}\}$.
    Recall that for every $A_{a,b}\in C_f$, $A_{a,b}=P_{a,b}\cdot S_{a,b}$.
    We define two subsets $C_f^P=\{ A_{a,b} \in C_f \mid f \in P_{a,b}  \}$ and $C_f^S=\{ A_{a,b} \in C_f \mid f \in S_{a,b}  \}$.
    Note that $C_f^P\cup C_f^S = C_f$.

    We make the following claim.
    \begin{claim}\label{claim:CfS_small}
        $|C_f^S| \le 2/\eps + 2$.
    \end{claim}
    \begin{proof}
        We partition $C^S_f$ into two sets $C_E$ and $C_O$.
        $C_E$ and $C_O$ contain canonical paths of $C^S_f$ that were added to $C$ in even and odd steps of the algorithm, respectively.
        We proceed by proving that $C'_E=\{A_{u} \mid A_{u,v} \in C_E \}$ is $f$-bad, which implies that $|C_E| = |C'_E| \le 1 /\eps+1$ via \cref{lem:fbadbound}.
        The size of $C_O$ can be bounded using identical arguments, which  yields the claim.

        Let $C_E = A_{a_1,b_1},A_{a_2,b_2}, \ldots ,A_{a_t,b_t}$ be the canonical paths of $C_E$ ordered such that $a_1\le_{P'} a_2 \le_{P'} \ldots \le_{P'} a_t$.
         Clearly, all the paths in $C'_E$ intersect in $f$, so we just need to show that every pair of different path forms a bad-cross.
        We start by observing that if the algorithm adds $A_{a',b'}$ to $C$, and two iterations later adds $A_{a'',b''}$ to $C$, it must be the case that $a'' >_{P'} b'$.
        It follows from the monotonicity of the $a$ values of $A_{a,b}$ added by the algorithm, and by the fact that $C_E$ contains only paths that were added on an even step of the algorithm that $b_i <_{P'} a_j$ for every $i<j \in [t]$.

        We claim that for every $i<j \in [t]$, we have $A_{a_i} \badcross A_{a_j}$. This follows from $A_{a_i} \cross A_{a_j}$ (due to the intersection of $S_{a_i}$ and $S_{a_j}$ in $f$), and from $b_i$ ($<_{P'} a_j$) being the last vertex on $P'$ that participates in a pair of $L$ and has $A_{a_i} \goodcross A_{b_i}$.
    \end{proof}
        In addition, we make the following claim.
    \begin{claim}\label{claim:CfP_small}
        $|C_f^P| \le 3/\eps + 3$.
    \end{claim}
    \begin{proof}
        The proof is very similar to the proof of \cref{claim:CfS_small}, but require a some subtle adjustments.

        We partition $C^P_f$ into three sets $C_0$, $C_1$ and $C_2$.
        $C_0$, $C_1$ and $C_2$ contain canonical paths of $C^P_f$ that were added to $C$ in steps numbered with $0$, $1$, and $2$ modulo $3$ of the algorithm, respectively (i.e. $C_1$ contains path added in steps $1,4,7,\ldots$).
        We proceed by proving that $C'_0=\{A_{v} \mid A_{u,v} \in C_0 \}$ is $f$-bad, which implies that $|C_0| = |C'_0| \le 1 /\eps+1$ via \cref{lem:fbadbound}.
        The sizes of $C_1$ and $C_2$ can be bounded using identical arguments, which  yields the claim.

        Let $C_0 = A_{a_1,b_1},A_{a_2,b_2}, \ldots , A_{a_t,b_t}$ be the canonical paths of $C_0$ ordered such that $a_1\le_{P'} a_2 \le_{P'} \ldots \le_{P'} a_t$.
         Clearly, all the paths in $C'_0$ intersect in $f$, so we just need to show that every pair of different path forms a bad-cross.

        We claim that for every $i<j \in [t]$, we have $A_{b_i} \badcross A_{b_j}$, by considering two cases.
        \paragraph{Case 1: $a_i = b_i$.}
        In this case, no pair $(a',b')$ in $L$ after $(a_i,({a_i})_{first})$ has $A_{a_i} \goodcross A_{a'}$ .
        In particular, $A_{a_i} \badcross A_{b_j}$ since $a_i <_{P'} a_j \le_{P'} b_j$.
        Since $a_i=b_i$, we have $A_{b_i} \badcross A_{b_j}$ as required.

        \paragraph{Case 2: $a_i \neq b_i$.}
        Let $3k$ be the iteration number in which $A_{a_i,b_i}$ was added to $C$.
        In this case,  the algorithm will process the pair $(b_i,({b_{i}})_{first})$ in iteration $3k+1$.
        In this iteration, the path $A_{b_i,c}$ would be added to $C$ such that $(c,c_{first}) \in L$ is the last pair such that $A_{b_i} \goodcross A_c$.
        In iteration $3k+2$, some path will be added to $C$, and all following iterations (in particular, the iteration in which $A_{a_j,b_j}$ was added) will process pairs $(a',b')$ such that $b' \ge_{P'} a' >_{P'} c$.
        Due to the maximality of $c$ and the fact that $A_{b_j}$ crosses $A_{b_i}$ in $f$, we must have $A_{b_i} \badcross A_{b_j}$ as required.
    \end{proof}
    Combining \cref{claim:CfS_small,claim:CfP_small}  we complete the proof of \cref{lem:fewcanonical}.
\end{proof}

\subsection{The $\Labelonlysuffix$ labeling (proof of \cref{lem:onlysuffixLabeling})}
\label{sec:onlysuffix}
In this section, we prove \cref{lem:onlysuffixLabeling}; The settings are as follows.
We are given a graph $G$, two paths $P$ and $P'$ of length $0$, a path $A=(u\rightsquigarrow u')$ from the last vertex of $P'$ to $u'$ which is the first vertex of $P$ and numbers $\alpha, \eps \in \mathbb{R}^+$ such that:
\begin{enumerate}
    \item $\len(A) \le (1+\eps)\alpha$.
    \item $A= A_1 \cdot A_2$ such that $A_1$ and $A_2$ are shortest paths in $G$.
\end{enumerate}
For most of this section, we focus on a vertex $a \in P'$.
For every vertex $f\in A$, we define $b^f_{suf}=\dfirst{G\setminus A[u,f]}{a}{P}{\alpha}$, where $P[|P|]$ is the last vertex of $P$.
Note that the set of vertices $f$ such that $b^f_{suf}$ is well defined forms a (possibly empty) prefix of $A$.
We denote as $z'$ the last vertex on $A$ such that $b^{z'}_{suf}$ is well defined (if there is no such vertex, we say that $z'=null$).
We make the following observation regarding $b_{suf}^f$.
\begin{observation}\label{obs:bfsuff_mon}
   For every $f_1\le_{A} f_2 \le_{A} z'$ we have $b_{suf}^{f_1}\le_{P} b_{suf}^{f_2}$.
\end{observation}

Most of the section is devoted to the proof of the following technical lemma.
\begin{lemma}\label{lem:only_s2_sequnces}
    If $z' \neq null$, there are sequences $F=(u=f_1 \le_A f_2 \le_A \ldots \le_A f_t =z')$ and $B=(b_1,b_2, \ldots ,b_{t-1})$, with $t=O(1/\eps)$, such that for every $f_i \le_A f <_A f_{i+1}$, we have:
    \begin{enumerate}
        \item $b^{f}_{suf} \ge_P b_i$.
        \item $\dist_{\Gf}(a,b_i) \le (1 + \eps) \alpha$.
    \end{enumerate}
\end{lemma}
\begin{proof}
First, we reduce the problem of finding such sequences to the case where $A$ is internally disjoint from $P'$.
Let $x$ be the first vertex on $P'$ that is on $A$.
Let $\Tilde{P'} = P'[a,x]$ and $\Tilde{A}=A[x,u']$, and assume we have sequences $\Tilde{F}$ and $\Tilde{B}$ as in the statement of the lemma for $\Tilde{P'}$ and $\Tilde{A}$ in $G \setminus A[u,x)$. Note that $\Tilde{P'}$ and $\Tilde{A}$ are internally disjoint due to the definition of $x$.
We claim that the sequences $F=(u)\cdot\Tilde{F}$ and $B=(u')\cdot \Tilde{B}$ satisfy the lemma, for $P'$ and $A$, due to the following.
\begin{itemize}
    \item For every $f \ge_A x$, the sequences $\Tilde{B}$ and $\Tilde{F}$ are satisfactory for $P'$ and $A$, as $b^f_{suf}$ in $G$ is the same as $b^f_{suf}$ in $G\setminus A[u,x)$ for $f \ge_A x$, and distances in $\Gf$ are shorter than distances in $G \setminus(A[u,x) \cup {f})$.
    \item  For every $f<_A x$ we have that $P'[a,x]\cdot A[x,u']$ is a path of length at most $(1+\eps)\alpha$ to $u'$ in $G\setminus A[u,f]$, and $b^f_{suf}\ge_P u'$ by definition.
\end{itemize}
So, from now on we assume that $A$ is internally disjoint from $P'$.

Recall that $A=A_1\cdot A_2$, and let $z$ be the last vertex of $A_1$ such that $b_{suf}^z$ is well defined (notice that $z'$ is not necessarily on $A_1$).
We show how to compute sequences $F$ and $B$ satisfying the required conditions for every $f\in A_1[u,z]$.
If $z'=z$ (i.e. the last vertex with well defined $b^f_{suf}$ on $A$ is in $A_1$)- these $F$ and $B$ sequences conclude the lemma.
Otherwise, $z$ is the last vertex on $A_1$.
In this case, we apply the same construction for $A_2[z,z']$, and the concatenation of the $F$ and $B$ sequences concludes the lemma.
We describe the construction for $A_1$, the construction for $A_2$ is similar.
Note that if $z=null$ then the lemma follows trivially.
We can therefore assume that $b^u_{suf}$ is well defined.

For every vertex $f\in A[u,z']$, let $D_f$ be a shortest path from $a$ to $b^f_{suf}$ in $G\setminus A[u,f]$.
Let $r_f$ be the last vertex on $D_f$ that is also on $A_1$ (note that $r_f$ may be undefined).
We denote as $\ell_f$ the first vertex on $A_1$ that is also on $D_f$.
Notice that, by definition we have $\ell_f \le_{D_f} r_f$ and $\ell_f \le_{A_1} r_f$.
Since both $A_1$ and $D_f$ are shortest paths, we can assume $D_f[\ell_f ,r_f] = A_1[\ell_f,r_f]$.

For two vertices $f_1 <_{A_1} f_2$, we say that $f_1$ crosses $f_2$, denoted as $f_1 \cross f_2$ if $\ell_{f_2} \in A_1[\ell_{f_1}, r_{f_1}]$.
For two crossing vertices $f_1 <_{A_1} f_2$, we say that $f_1$ bad-crosses $f_2$  and denote $f_1 \badcross f_2$, if $\dist_{G\setminus A_1[u,f_2]}(a,b_{suf}^{f_1}) > (1+\eps)\alpha$.

\paragraph{Algorithm.}
We present the following algorithm (see \cref{alg:nos1} below) that generates sequences  $F=(u=f_1,f_2,\dots ,f_t=z)$ and $B=(b_1,b_2,\dots ,b_{t-1})$.
The algorithm initializes the following variables.
\begin{enumerate}
    \item A vertex $f$ initially set to $u$: $f$ is meant to iterate $A_1$ from left to right.

    \item A vertex $R$ initially set to $r_u$ (or $null$ if $r_u$ is undefined). The vertex $R$ keeps track of the leftmost value of $r_f$ that was encountered so far.

    \item $b'$ keeps track of the $b^{f_R}_{suf}$ value of the vertex $f_R$ from which the value $R$ was obtained (even if $R=null$).
    Initially, $b'$ is set to $b^u_{suf}$.

\end{enumerate}

Having initialized $f$,$R$, and $b'$, the algorithm initializes the sequences as $F = (f)$ and $B = (b^f_{suf})$.

The algorithm runs the following procedure repeatedly until a terminating condition is met.
The algorithm finds the vertex $x$ which is the first vertex in $A_1[f,z]$ such that $\dist_{G \setminus A[u,x]}(a, b_{suf}^f) > (1 + \eps)\alpha$.
If there is no such $x$, the algorithm appends $z$ to $F$ and terminates.

Otherwise, the algorithm assigns $f\leftarrow x$, appends $f$ to $F$ and $b^f_{suf}$ to $B$, and proceeds according to the following cases.

If $D_f \cap A_1 = \emptyset$, the algorithm appends $z$ to $F$ and terminates.

If $D_f \cap A_1 \neq \emptyset$ and $\ell_f >_{A_1} R$, the algorithm appends $\ell_f$ and $z$ to $F$ (in that order), appends $b'$ to $B$ and terminates.
Finally, if $r_f <_{A_1} R$, the algorithm updates $R \leftarrow r_f$ and $b' \leftarrow b^f_{suf}$.

\begin{algorithm}
\caption{Partition $A_1$}\label{alg:nos1}
\SetKwInOut{Input}{Input}\SetKwInOut{Output}{Output}
\Input{$A=A_1\cdot A_2$, $G$, $u$, $u'$, $a$, $z$, $\alpha$, $\eps$}
\Output{Sequences $F$ and $B$}

Initialize  $f \leftarrow u$, $R \leftarrow r_f$, $b' \leftarrow b_{suf}^f$, $F \leftarrow (f)$, and $B \leftarrow (b_{suf}^f)$\label{ln:init_loop}\;
\While{true}{
    Let $x$ be the first vertex in $A_1[f,z]$ such that $\dist_{G \setminus A[u,x]}(a, b_{suf}^f) > (1 + \eps)\alpha$\label{ln:iteration}\;
    \If{$x$ does not exist}{
        Append $z$ to $F$ and \Return $F,B$\label{ln:loop_no_x}\;
    }
    $f \leftarrow x$\;
    Append $f$ to $F$ and $b_{suf}^f$ to $B$\label{ln:loop_append}\;
    \If{$D_f \cap A_1 = \emptyset$}{
        Append $z$ to $F$ and \Return $F,B$\label{ln:loop_disj}\;
    }
    \Else{
        \If{$\ell_f >_{A_1} R$}{
            Append $\ell_f$, and $z$ to $F$, and  $b'$ to $B$, and \Return $F,B$ \label{ln:loop_no_cross}\;
        }
        \Else{
            \If{$R >_{A_1} r_f$}{
                $R \leftarrow r_f$ and $b' \leftarrow b_{suf}^f$\label{ln:loop_R}\;
            }
        }
    }
}
\end{algorithm}

\paragraph{Correctness.}

Notice that the following invariant holds at any time during the execution of \cref{alg:nos1}:
\begin{invariant}\label{inv:Rmin}
    At the beginning of every iteration (\cref{ln:iteration}), for every $f\in F$ we have
    $R\le_{A_1} r_{f}$.\footnote{If $r_f$ is undefined we consider $z$ as $r_f$.}
\end{invariant}

Let $F= (f_1,f_2, \ldots ,f_t)$ and $B=(b_1,b_2,\dots,b_{t-1})$ be the output of the algorithm.

We make the following helpful claim.
\begin{claim}\label{claim:bad_cross_t_minus_4}
    For every $i\in [1..t-4]$, it holds that $f_i \badcross f_{i+1}$.
\end{claim}
\begin{proof}
    Since $i \le t-4$, the vertex $f_{i+1}$ must have been added to $F$ in line \cref{ln:loop_append} and both $D_{f_{i+1}} \cap A_1 \neq \emptyset$ and $\ell_{f_{i+1}} \le_{A_1} R$ (otherwise, the algorithm terminates with $t \le i+3$).
    Due to \cref{inv:Rmin}, we have $R \le_{A_1} r_{f_i}$ and therefore $\ell_{f_{i+1}} \le_{A_1} r_{f_i}$.
    Since $f_{i+1}$ is the first vertex in $A[f_{i},z]$ such that $\dist_{G\setminus A[u,f_{i+1}]}(a,b^{f_i}_{suf})) > (1+\eps)\alpha$, we must have $f_{i+1} \ge_{A_1} \ell_{f_i}$.
    Otherwise, $D_{f_i}$ is a path from $a$ to $b^{f_i}_{suf}$ in $G \setminus A[u,f_{i+1}]$ of length $\len(D_{f_i}) \le \alpha$.
    We have shown that $\ell_{f_{i+1}} \in A_1[\ell_{f_i} , r_{f_i}]$ and therefore $f_i \cross f_{i+1}$.
    From the definition of $f_{i+1}$ (\cref{ln:iteration}) we have that $\dist_{G\setminus A[u,f_{i+1}]}(a,b^{f_i}_{suf}) > (1+\eps) \alpha$ and therefore $f_i \badcross f_{i+1}$ as required.
\end{proof}

We define a sequence of useful paths in $G$.
For every $f_i\in F$ such that $\ell_{f_i} \in A_1$, we denote $S_i=A[u,\ell_{f_i}]\cdot D_{f_i}[\ell_{f_i},b_{suf}^{f_i}]$.
The paths $S_i$ are instrumental in the proofs of the following claims.

\begin{claim}\label{claim:Si_not_long}
    Let $i\in[1..t]$ such that $r_{f_i}\in A_1$ then, for any $f>_{A_1}r_{f_i}$ we have $\dist_{\Gf}(a,b_{suf}^{f_i})\le \len(S_i) \le (1+\eps)\alpha$.
\end{claim}
\begin{proof}
First note that by definitions of $S_i$ and $r_{f_i}$, we have $A_1[r_{f_i}+1,z]\cap S_i=\emptyset$ (and recall that $A_1[\ell_{f_i},r_{f_i}]=D_{f_i}[\ell_{f_i},r_{f_i}]$).
Therefore, $P'[a,u]\cdot S_i$ is a path in $\Gf$ for any $f >_A r_{f_i}$ and it follows that $\dist_{\Gf}(a,b^{f_i}_{suf})\le \len(P')+\len(S_i)=\len(S_i)$.

It remains to prove $\len(S_i)\le (1+\eps)\alpha$.
We consider two cases regarding $b_{suf}^{f_i}$:
(1) $b_{suf}^{f_i}=u'$, in this case since $D_{f_i}[\ell_{f_i},u']$ is a shortest path, we have $\len(D_{f_i}[\ell_{f_i},u'])\le\len(A[\ell_{f_i},u'])$.
Thus, $\len(S_i)=\len(A[u,\ell_{f_i}])+\len(D_{f_i}[\ell_{f_i},u'])\le\len(A[u,\ell_{f_i}])+\len(A[\ell_{f_i},u'])=\len(A)\le(1+\eps)\alpha$.

The second case is (2) $b_{suf}^{f_i}>_P u'$.
From definition of $b_{suf}^{f_i}$ we have that $\dist_{G\setminus A[u,f_i]}(a,u')>\alpha$.
In particular, by definition of $\ell_{f_i}$ we have $\len(D_{f_i}[a,\ell_{f_i}]\cdot A[\ell_{f_i},u'])>\alpha$.
Moreover we have $\len(D_{f_i}[a,\ell_{f_i}]\cdot D_{f_i}[\ell_{f_i},b_{suf}^{f_i}])=\len(D_{f_i})\le \alpha$ and $\len(A[u,\ell_{f_i}]\cdot A[\ell_{f_i},u'])=\len(A)\le (1+\eps)\alpha$.
Combining the above we get $\len(S_i)=\len(A[u,\ell_{f_i}]\cdot D_{f_i}[\ell_{f_i},b_{suf}^{f_i}])\le \alpha+(1+\eps)\alpha-\alpha=(1+\eps)\alpha$.
\end{proof}

We prove the following claim by induction:
\begin{claim}\label{claim:bad_no_s1}
For every $i\in [1,t-4]$ it holds that $\len(S_i)\le (1-(i-2)\eps)\alpha$.
\end{claim}
\begin{proof}
Notice that if $t<5$ the claim is vacuously true.
Notice that for $i\in [1,t-4]$ we have $b_i=b_{suf}^{f_i}$, $\ell_{f_i}$ and $r_{f_i}$ are well defined and by \cref{claim:bad_cross_t_minus_4} we have $f_i\badcross f_{i+1}$.

We prove the claim by induction on $i$.
For $i=1$ the claim $\len(S_1)\le (1+\eps)\alpha$ follows from \cref{claim:Si_not_long}.

We assume the claim holds for $i$ and prove for $i+1\le t-4$.
By the algorithm (\cref{ln:iteration}), $\dist_{G\setminus A[u,f_{i+1}]}(a,b_i)>(1+\eps)\alpha$ and in particular, $\len(D_{f_{i+1}}[a,\ell_{f_{i+1}}]\cdot S_i[\ell_{f_{i+1}},b_i])>(1+\eps)\alpha$ (note that $\ell_{f_{i+1}}\in A[\ell_{f_i},r_{f_i}]$ since $f_i\badcross f_{i+1}$).
Moreover, $\len(S_i[u,\ell_{f_{i+1}}]\cdot S_i[\ell_{f_{i+1}},b_i])=\len(S_i)\le(1-(i-2)\eps)\alpha$ and $\len(D_{f_{i+1}}[a,\ell_{f_{i+1}}]\cdot D_{f_{i+1}}[\ell_{f_{i+1}},b_{i+1}])=\len(D_{f_{i+1}})\le \alpha$.
Combining the above we get $\len(S_{i+1})=\len(S_{i+1}[u,\ell_{f_{i+1}}]\cdot S_{i+1}[\ell_{f_{i+1}},b_{i+1}])\le \alpha+(1-(i-2)\eps)\alpha-(1+\eps)\alpha=(1-(i-1)\eps)\alpha$ as required.
\end{proof}

The following claim is a direct consequence of \cref{claim:bad_no_s1}.

\begin{claim}\label{claim:F_is_small}
    $|F|=O(1/\eps)$
\end{claim}
\begin{proof}
    Assume by contradiction that $t= |F| >1/\eps+7$, and deduce by \cref{claim:bad_no_s1} that $S_i$ is a path of negative length, in contradiction.
    Thus, $|F|\le 1/\eps+7$ as required.
\end{proof}

We prove the following claim.
\begin{claim}\label{claim:F_and_B}
    For every $f\in A_1$ such that  $f_i \le_{A_1} f <_{A_1} f_{i+1}$, we have:
    \begin{enumerate}
        \item $b^f_{suf} \ge_P b_i$.
        \item $\dist_{G\setminus \{f\}}(a,b_i) \le (1 + \eps) \alpha$.
    \end{enumerate}
\end{claim}
\begin{proof}
We consider several cases regarding $b_i$.

\paragraph{Case 1: $b_i=b_{suf}^{f_i}$.}
In this case by \cref{obs:bfsuff_mon} since $f_i\le_{A_1} f$ we have $b_{suf}^f\ge_{P} b_{suf}^{f_i}=b_i$.
We consider two sub-cases regarding $f_{i+1}$:
\begin{itemize}
    \item If $f_{i+1}$ was added in Lines \ref{ln:loop_no_x}, \ref{ln:loop_append} or \ref{ln:loop_disj}.
    In each of these lines, $f_{i+1}$ is set to be the first $x$ in $A_1[f_{i},z]$ such that $\dist_{G \setminus A[u,x]}(a, b_i) > (1 + \eps)\alpha$ or is set to $z$ if there is no such $x$.
    Either way, due to $f <_{A_1} f_{i+1}$ we have $\dist_{G \setminus \{f\}}(a, b_i) \le \dist_{G \setminus A[u,f]}(a, b_i) \le (1 + \eps)\alpha$ as required.
    \item If $f_{i+1}=\ell_{f_i}$ was added in \cref{ln:loop_no_cross}, then $\dist_{G\setminus \{f\}}(a,b_i)  \le \dist_{G\setminus A_1[u,\ell_{f_i}-1]}(a,b_i)\stackrel{(*)}{\le} \len(D_{f_i}) \le \alpha   \le (1 + \eps) \alpha$,
    where the inequality (*) follows from $D_{f_i}\cap A_1[u,\ell_{f_i}-1]=\emptyset$ by defintion of $\ell_{f_i}$.
\end{itemize}

\paragraph{Case 2: $b_i$ was added in \cref{ln:loop_no_cross}.}
In this case $f_i=\ell_{f_{i-1}}$ and $f_{i+1}=z$.
In addition $b_i=b'$ at this time of the algorithm.
Consider the values of $R$ and $b'$ in the iteration of the algorithm where \cref{ln:loop_no_cross} was executed.
Notice that $b_i=b'=b_{suf}^{f_j}$ and $R=r_{f_j}$ for some $j<i$ (as assigned in either \cref{ln:init_loop} or \cref{ln:loop_R}) such that $r_{f_j}\in A_1$.
It follows from \cref{obs:bfsuff_mon} that $b^f_{suf} \ge_P b^{f_i}_{suf} \ge_P b^{f_j}_{suf}=b_i$ as required.
The algorithm executes \cref{ln:loop_no_cross} since $f\ge_{A_1}\ell_{f_{i-1}}>_{A_1}R=r_{f_j}$.
Finally, by \cref{claim:Si_not_long} we have $\dist_{\Gf}(a,b_{suf}^{f_j})\le \len(S_j)\le (1+\eps)\alpha$.
\end{proof}

Combining \cref{claim:F_and_B,claim:F_is_small}, concludes the proof of \cref{lem:only_s2_sequnces}.
\end{proof}

We are now ready to present the labeling scheme for $\Labelonlysuffix$, proving \cref{lem:onlysuffixLabeling}.

\onlysuffixlem*
\begin{proof}
\hlgray{For every vertex $v\in P'$:
Let $z'_v$ be the last vertex of $A$ such that $b^{z'}_{suf}$ is well defined.
If $z'=null$, the label of $v$ stores a flag.
Otherwise, $v$ stores the sequences $F$ and $B$ obtained by applying \mbox{\cref{lem:only_s2_sequnces}} on $v$ and $b^{z'}_{suf}$.}
In addition \hlgray{for every vertex $v\in A$:
the label of $v$ stores $v$'s index in $A$.}
From \cref{lem:only_s2_sequnces}, it is clear that the size of the label is $\Opoly$.

Given the labels of $a\in P'$ and $f\in A$, one can obtain a vertex $b\in P$ such that $\dist_{\Gf}(a,b) \le (1+\eps) \alpha$ and $b \le_P \dfirst{G \setminus A[\ldots f]}{a}{P}{\alpha}$ or conclude that $\dfirst{G \setminus A[\ldots f]}{a}{P}{\alpha}=null$, as follows.
If $z'_a=null$ (marked in the label of $a$ with a flag) or $f>_A z'_a$ ($z'_a=F[|F|]$), then we conclude that $\dfirst{G \setminus A[\ldots f]}{a}{P}{\alpha}=null$.
If $f=z'_a$, we simply return $b^{z'_a}_{suf}$.
Otherwise, let $i$ be the index such that  $f_i \le_A f <_A f_{i+1}$, return $b_i$, which by \cref{lem:only_s2_sequnces} satisfies the requirements.
\end{proof}

\section{The $\LabelatoP$ Labeling (Proof of \cref{lem:atopflabel})}\label{sec:atopf}
In this section we prove \cref{lem:atopflabel}.
The settings in this section are as follows.
$G$ (for the ease of presentation we use here $G$ as the name of the graph) is a graph, $P'$ is a $0$-length path, $P$ is a path without outgoing edges which lies on a single face, and $\alpha,\eps\in\mathbb R^+$.
For $f\in P$ let $P_1$ and $P_2$ be the prefix and suffix of $P$ before and after $f$ (without $f$), respectively.
Our goal is to develop a labeling scheme such that
given the labels of vertices $a\in P'$ and $f\in P$, one can retrieve two vertices $b_1$ and $b_2$ which are $\ddfirst{\Gf}{a}{P_1}{\alpha}{\eps\alpha}$ and $\ddfirst{\Gf}{a}{P_2}{\alpha}{\eps\alpha}$, respectively.
The size of each label is required to be $\Opoly$.

Let $z$ be the first vertex of $\Psf$.
Let $\Gamma=\{i\eps\alpha\mid i\in[0,\ceil{1/\eps}]\}$ be the set of multiples of $\eps\alpha$.
Let $F$ be the graph obtained from $G$ by removing all edges of $P$.
For every $\beta\in\Gamma$ let $P_\beta=\{v\in P\mid \beta-\eps\alpha<\dist_F(z,v)\le \beta\}$ denote the set of vertices of $P$ whose distance from $z$ in $F$ (which is exactly the distance in $G$ via paths that are \textbf{internally disjoint} from $P$) is in the interval $(\beta-\eps\alpha,\beta]$.
Notice that since $P$ lies on a single face, all paths entering $P$ enter from the same side.

For every $a\in P'$ let $P_\beta(a)=\{v\in P_\beta \mid \dist_F(a,v)\le \alpha\}$ denote the subset of vertices of $P_\beta$ whose distance from $a$ in $F$ is at most $\alpha$.
By definition $P_\beta(a) \subseteq P_\beta$ and planarity dictates the following:
\begin{claim}\label{claim:Az}
For any $\beta\in\Gamma$, considering the sets $P_{\beta}, P_\beta(a)$ as sequences ordered according to the order along $\P$.
There is a set of at most two intervals of consecutive vertices of $P_\beta$ such that:
(i) every vertex in $P_\beta(a)$ is in one of these two intervals,
and (ii) for every vertex $w \in P_\beta$ in each of these intervals, $\dist_F(a,w)\le (1+\eps)\alpha$.
Finally, (iii) the endpoints of each interval are in $P_\beta$.
\end{claim}
\begin{proof}
Recall that $P$ lies on a single face of $G$.
Let $P^\circ$ be the cycle of the boundary of the face $P$ lies on.
We shall show that with respect to the cyclic order on $\P^\circ$, a single interval in the statement of the claim suffices.
This interval consists of at most two intervals of $\P$.

Consider all pairs of consecutive vertices (in the cyclic order along $\P_\beta(a)$) $u,v \in P_\beta(a)$.
Let $C_{uv}$ be the (undirected) cycle formed by the $a$-to-$u$ shortest path in $F$, the $a$-to-$v$ shortest path in $F$, and $\P^\circ[u,v]$.
Since the union of the shortest paths forming the above cycles is a tree (a subtrtee of the shortest path tree rooted at $a$) whose leaves are all on $\P$, the vertex $z$ is strictly enclosed by at most one of these cycles.
Let $C_{u^*v^*}$ be the cycle that strictly encloses $z$.
We choose the interval to start at $v^*$ and end at $u^*$ (i.e., the interval containing all vertices of $P_\beta$ except those in $\P(u^*,v^*)$).
By definition, the interval contains all vertices of $P_\beta(a)$, so property (i) is satisfied.
To show property (ii), consider a vertex $w \in P_\beta$ in this interval.
By choice of the interval, $z$ is enclosed by $C_{u^*v^*}$ and $w$ is not strictly enclosed by $C_{u^*v^*}$.
Hence, the shortest path from $z$ to $w$ must intersect the $a$-to-$u^*$ path or the $a$-to-$v^*$ path.
Suppose without loss of generality that the intersection is with the $a$-to-$u^*$ path (the other case is identical with $v^*$ taking the role of $u^*$), and let $x$ be an intersection vertex.
If $w \in  P_\beta(a)$ then property (ii) is satisfied by definition of $P_\beta(a)$.
Otherwise, suppose for the sake of contradiction that property (ii) is not satisfied.
That is, $\dist_F(a,w)>(1+\eps)\alpha$.
In particular, the sum of lengths of $a$-to-$x$ prefix of the $a$-to-$u^*$ path and the $x$-to-$w$ suffix of the $z$-to-$w$ path is at least $(1+\eps)\alpha$.
But since the sum of lengths of the $a$-to-$u^*$ path and of the $z$-to-$w$ path is less than $\beta + \alpha$,
we get that the sum of the length of the $z$-to-$x$ prefix of the $z$-to-$w$ path and the $x$-to-$u^*$ suffix of the $a$-to-$u^*$ path must be less than $\beta-\eps\alpha$, a contradiction to $u \in P_\beta$.
\end{proof}

We are now ready to prove \cref{lem:atopflabel}.
\atopflabellem*
\begin{proof}
For every $a\in P'$ the label of $a$ stores: \hlgray{$\dfirst{G}{a}{P}{\alpha}$ and for every $\beta\in \Gamma$, the set of (at most two) intervals obtained by \mbox{\cref{claim:Az}}.}
For every $f\in P$ the label of $f$ stores: \hlgray{for every $\beta\in \Gamma$, the successor of $f$ in $P_\beta$.}
(Every vertex of $P$ is stored with its index on $P$.)

\paragraph{Size.}
It is clear that the size of each label is $O(|\Gamma|)=\Opoly$.

\paragraph{Decoding.}
Given the labels of $a\in P'$ and $f\in P$, we obtain $b_1$ and $b_2$ as follows.
First, if $\dfirst{G}{a}{P}{\alpha}<_P f$ then $b_1=\dfirst{G}{a}{P}{\alpha}$, otherwise $b_1=null$.
To compute $b_2$, we iterate over all $\beta\in\Gamma$.
If $f$ is in one of the intervals obtained by \cref{claim:Az}, then $b^\beta_2$ is the succesor of $f$ in $P_\beta$.
Otherwise, $b^\beta_2$ is the first vertex on $P_2$ which is an endpoint of an interval, or $null$ if there is no such endpoint.
Finally, we set $b_2=\min_{\le_P}\{b_2^\beta\mid \beta\in\Gamma\}$.

\paragraph{Correctness.}
It is straightforward that if $\dfirst{G}{a}{P}{\alpha}\ge_P f$ then $\dfirst{\Gf}{a}{P_1}{\alpha}=null$ and therefore $b_1=null$ is a valid answer.
If $\dfirst{G}{a}{P}{\alpha}<_P f$ then clearly since $P$ does not have outgoing edges $\dfirst{G}{a}{P}{\alpha}=\dfirst{\Gf}{a}{P_1}{\alpha}$ is an $\ddfirst{\Gf}{a}{P_1}{\alpha}{\eps\alpha}$.

For every $\beta\in\Gamma$, in all cases $b_2^\beta$ is in one of the intervals obtained by \cref{claim:Az} and so $b^\beta_2\in P_\beta$.
Thus, by \cref{claim:Az} (ii), $\dist_G(a,b^\beta_2)\le \dist_F(a,b_2^\beta)\le (1+\eps)\alpha$.

Let $b^*=\dfirst{\Gf}{a}{P_2}{\alpha}$ and let $\beta=q\in\Gamma\mid q\ge\dist_F(z,b^*)\}$.
It remains to prove $b_2\le_P b^*$.
Notice that $b^*\in P_\beta$.
Moreover $b^*\in P_\beta(a)\subseteq P_\beta$.
Let $I=P[u,v]$ be the endpoints of the interval obtained by \cref{claim:Az} containing $b^*$ (such an interval exists by (i)).
If $f\in P[u,v]$ then $b_2^\beta$ is the successor $w$ of $f$ in $P_\beta$ which implies $w\le_P b^*$.
If $f<_P u$, then $b_2^\beta\le_P u\le_P b^*$.
To conclude, $b^\beta_2\le_P b^*$ and by the minimality of $b_2$ among all values of $\beta$ we have $b_2\le_P b^*$. \qedhere

\end{proof}

\bibliography{ref}

\newcommand{\etalchar}[1]{$^{#1}$}
\begin{thebibliography}{GMWWN18}

\bibitem[ACC{\etalchar{+}}96]{ArikatiCCDSZ96}
Srinivasa~Rao Arikati, Danny~Z. Chen, L.~Paul Chew, Gautam Das, Michiel H.~M.
  Smid, and Christos~D. Zaroliagis.
\newblock Planar spanners and approximate shortest path queries among obstacles
  in the plane.
\newblock In {\em 4th {ESA}}, volume 1136, pages 514--528, 1996.

\bibitem[ACG12]{AbrahamCG12}
Ittai Abraham, Shiri Chechik, and Cyril Gavoille.
\newblock Fully dynamic approximate distance oracles for planar graphs via
  forbidden-set distance labels.
\newblock In {\em 44th {STOC}}, pages 1199--1218, 2012.

\bibitem[ACGP16]{AbrahamCGP16}
Ittai Abraham, Shiri Chechik, Cyril Gavoille, and David Peleg.
\newblock Forbidden-set distance labels for graphs of bounded doubling
  dimension.
\newblock {\em {ACM} Trans. Algorithms}, 12(2):22:1--22:17, 2016.

\bibitem[ADK17]{alstrup2015optimal}
Stephen Alstrup, S{\o}ren Dahlgaard, and Mathias B{\ae}k~Tejs Knudsen.
\newblock Optimal induced universal graphs and adjacency labeling for trees.
\newblock {\em J. {ACM}}, 64(4):27:1--27:22, 2017.

\bibitem[AKTZ19]{alstrup2015adjacency}
Stephen Alstrup, Haim Kaplan, Mikkel Thorup, and Uri Zwick.
\newblock Adjacency labeling schemes and induced-universal graphs.
\newblock {\em {SIAM} J. Discret. Math.}, 33(1):116--137, 2019.

\bibitem[AN17]{AlonN17}
Noga Alon and Rajko Nenadov.
\newblock Optimal induced universal graphs for bounded-degree graphs.
\newblock In {\em 28th SODA}, pages 1149--1157, 2017.

\bibitem[BCG{\etalchar{+}}22]{Bar-NatanCGMW22}
Aviv Bar{-}Natan, Panagiotis Charalampopoulos, Pawel Gawrychowski, Shay Mozes,
  and Oren Weimann.
\newblock Fault-tolerant distance labeling for planar graphs.
\newblock {\em Theor. Comput. Sci.}, 918:48--59, 2022.

\bibitem[BCR18]{BaswanaCR18}
Surender Baswana, Keerti Choudhary, and Liam Roditty.
\newblock Fault-tolerant subgraph for single-source reachability: General and
  optimal.
\newblock {\em {SIAM} J. Comput.}, 47(1):80--95, 2018.

\bibitem[BGL07]{bonichon2007short}
Nicolas Bonichon, Cyril Gavoille, and Arnaud Labourel.
\newblock Short labels by traversal and jumping.
\newblock {\em Electronic Notes in Discrete Mathematics}, 28:153--160, 2007.

\bibitem[BLM12]{Baswana}
Surender Baswana, Utkarsh Lath, and Anuradha~S. Mehta.
\newblock Single source distance oracle for planar digraphs avoiding a failed
  node or link.
\newblock In {\em 23rd {SODA}}, pages 223--232, 2012.

\bibitem[Cab12]{Cabello12}
Sergio Cabello.
\newblock Many distances in planar graphs.
\newblock {\em Algorithmica}, 62(1-2):361--381, 2012.

\bibitem[CDW17]{Cohen-AddadDW17}
Vincent Cohen{-}Addad, S{\o}ren Dahlgaard, and Christian Wulff{-}Nilsen.
\newblock Fast and compact exact distance oracle for planar graphs.
\newblock In {\em 58th {FOCS}}, pages 962--973, 2017.

\bibitem[CGK09]{Courcelle09}
Bruno Courcelle, Cyril Gavoille, and Mamadou~M. Kant{\'{e}}.
\newblock Compact labelings for efficient first-order model-checking.
\newblock {\em Journal of Combinatorial Optimization}, 21(1):19--46, 2009.

\bibitem[CGL{\etalchar{+}}23]{CharalampopoulosGLMPWW23}
Panagiotis Charalampopoulos, Pawel Gawrychowski, Yaowei Long, Shay Mozes, Seth
  Pettie, Oren Weimann, and Christian Wulff{-}Nilsen.
\newblock Almost optimal exact distance oracles for planar graphs.
\newblock {\em J. {ACM}}, 70(2):12:1--12:50, 2023.

\bibitem[CGMW19]{Charalampopoulos19}
Panagiotis Charalampopoulos, Pawel Gawrychowski, Shay Mozes, and Oren Weimann.
\newblock Almost optimal distance oracles for planar graphs.
\newblock In {\em 51st {STOC}}, pages 138--151, 2019.

\bibitem[Cho16]{choudhary2016DualFaultTolerant}
Keerti Choudhary.
\newblock An optimal dual fault tolerant reachability oracle.
\newblock In {\em 43rd {ICALP}}, pages 130:1--130:13, 2016.

\bibitem[CMT19]{faultyOracle}
Panagiotis Charalampopoulos, Shay Mozes, and Benjamin Tebeka.
\newblock Exact distance oracles for planar graphs with failing vertices.
\newblock In {\em 30th {SODA}}, pages 2110--2123, 2019.

\bibitem[CT10]{CourcelleT07}
Bruno Courcelle and Andrew Twigg.
\newblock Constrained-path labellings on graphs of bounded clique-width.
\newblock {\em Theory of Computing Systems}, 47(2):531--567, 2010.

\bibitem[CX00]{ChenX00}
Danny~Z. Chen and Jinhui Xu.
\newblock Shortest path queries in planar graphs.
\newblock In {\em 32nd {STOC}}, pages 469--478, 2000.

\bibitem[Dji96]{Djidjev96}
Hristo Djidjev.
\newblock Efficient algorithms for shortest path queries in planar digraphs.
\newblock In {\em 22nd {WG}}, volume 1197, pages 151--165, 1996.

\bibitem[ES35]{ES35}
P.~Erd\"os and G.~Szekeres.
\newblock A combinatorial problem in geometry.
\newblock {\em Compositio Mathematica}, 2:463--470, 1935.

\bibitem[FKMS07]{FeigenbaumKMS07}
Joan Feigenbaum, David~R. Karger, Vahab~S. Mirrokni, and Rahul Sami.
\newblock Subjective-cost policy routing.
\newblock {\em Theor. Comput. Sci.}, 378(2):175--189, 2007.

\bibitem[FR06]{FR06}
J.~Fakcharoenphol and S.~Rao.
\newblock Planar graphs, negative weight edges, shortest paths, and near linear
  time.
\newblock {\em J. Comput. Syst. Sci.}, 72(5):868--889, 2006.

\bibitem[GGI{\etalchar{+}}17]{GeorgiadisGIPU17}
Loukas Georgiadis, Daniel Graf, Giuseppe~F. Italiano, Nikos Parotsidis, and
  Przemyslaw Uznanski.
\newblock All-pairs 2-reachability in ${O}(n^\omega \log n)$ time.
\newblock In {\em 44th {ICALP}}, pages 74:1--74:14, 2017.

\bibitem[GIP17]{GeorgiadisIP17}
Loukas Georgiadis, Giuseppe~F. Italiano, and Nikos Parotsidis.
\newblock Strong connectivity in directed graphs under failures, with
  applications.
\newblock In {\em 28th {SODA}}, pages 1880--1899, 2017.

\bibitem[GKK{\etalchar{+}}01]{GavoilleKKPP01}
Cyril Gavoille, Michal Katz, Nir~A. Katz, Christophe Paul, and David Peleg.
\newblock Approximate distance labeling schemes.
\newblock In {\em 9th {ESA}}, pages 476--487, 2001.

\bibitem[GKR01]{GuptaKR01}
Anupam Gupta, Amit Kumar, and Rajeev Rastogi.
\newblock Traveling with a pez dispenser (or, routing issues in {MPLS)}.
\newblock In {\em 42nd {FOCS}}, pages 148--157, 2001.

\bibitem[GMWWN18]{ourSODA2018}
Pawel Gawrychowski, Shay Mozes, Oren Weimann, and Christian Wulff-Nilsen.
\newblock Better tradeoffs for exact distance oracles in planar graphs.
\newblock In {\em SODA}, 2018.

\bibitem[GPPR04]{GPPR04}
Cyril Gavoille, David Peleg, St{\'e}phane P{\'e}rennes, and Ran Raz.
\newblock Distance labeling in graphs.
\newblock {\em Journal of Algorithms}, 53(1):85--112, 2004.

\bibitem[GU23]{GawrychowskiU23}
Pawel Gawrychowski and Przemyslaw Uznanski.
\newblock Better distance labeling for unweighted planar graphs.
\newblock {\em Algorithmica}, 85(6):1805--1823, 2023.

\bibitem[HL09]{HsuL09}
Tai{-}Hsin Hsu and Hsueh{-}I Lu.
\newblock An optimal labeling for node connectivity.
\newblock In {\em 20th {ISAAC}}, 2009.

\bibitem[HLNW17]{Henzinger2017}
Monika Henzinger, Andrea Lincoln, Stefan Neumann, and Virginia~Vassilevska
  Williams.
\newblock Conditional hardness for sensitivity problems.
\newblock In {\em 8th {ITCS}}, pages 26:1--26:31, 2017.

\bibitem[HRT15]{HolmRT15}
Jacob Holm, Eva Rotenberg, and Mikkel Thorup.
\newblock Planar reachability in linear space and constant time.
\newblock In {\em 56th {FOCS}}, pages 370--389, 2015.

\bibitem[IKP21]{Reachability}
Giuseppe~F. Italiano, Adam Karczmarz, and Nikos Parotsidis.
\newblock Planar reachability under single vertex or edge failures.
\newblock In {\em 32nd {SODA}}, pages 2739--2758, 2021.

\bibitem[KKKP04]{KatzKKP04}
Michal Katz, Nir~A. Katz, Amos Korman, and David Peleg.
\newblock Labeling schemes for flow and connectivity.
\newblock {\em {SIAM} J. Comput.}, 34(1):23--40, 2004.

\bibitem[Kle05]{Klein05}
Philip~N Klein.
\newblock Multiple-source shortest paths in planar graphs.
\newblock In {\em SODA}, pages 146--155, 2005.

\bibitem[KM]{planarbook}
P.N. Klein and S.~Mozes.
\newblock Optimization algorithms for planar graphs.
\newblock \url{http://planarity.org}.
\newblock Book draft.

\bibitem[KNR92]{Kannan}
Sampath Kannan, Moni Naor, and Steven Rudich.
\newblock Implicit representation of graphs.
\newblock {\em SIAM Journal on Discrete Mathematics}, 5(4):596--603, 1992.

\bibitem[Kor10]{Korman10}
Amos Korman.
\newblock Labeling schemes for vertex connectivity.
\newblock {\em {ACM} Trans. Algorithms}, 6(2):39:1--39:10, 2010.

\bibitem[KS99]{king2002fully}
Valerie King and Garry Sagert.
\newblock A fully dynamic algorithm for maintaining the transitive closure.
\newblock In {\em 31st {STOC}}, pages 492--498, 1999.

\bibitem[LT79]{LTsep}
Richard~J. Lipton and Robert~Endre Tarjan.
\newblock A separator theorem for planar graphs.
\newblock {\em SIAM J. Appl. Math.}, 36(2):177--189, 1979.

\bibitem[MS12]{MozesS2012}
Shay Mozes and Christian Sommer.
\newblock Exact distance oracles for planar graphs.
\newblock In {\em 23rd {SODA}}, pages 209--222, 2012.

\bibitem[Nus11]{Nussbaum11}
Yahav Nussbaum.
\newblock Improved distance queries in planar graphs.
\newblock In {\em 12th {WADS}}, pages 642--653, 2011.

\bibitem[Pel05]{Peleg05}
David Peleg.
\newblock Informative labeling schemes for graphs.
\newblock {\em Theor. Comput. Sci.}, 340(3):577--593, 2005.

\bibitem[PRSWN16]{petersen2015near}
Casper Petersen, Noy Rotbart, Jakob~Grue Simonsen, and Christian Wulff-Nilsen.
\newblock Near-optimal adjacency labeling scheme for power-law graphs.
\newblock In {\em 43rd ICALP}, pages 133:1--133:15, 2016.

\bibitem[Rot16]{rotbart2016new}
Noy~Galil Rotbart.
\newblock {\em New Ideas on Labeling Schemes}.
\newblock PhD thesis, University of Copenhagen, 2016.

\bibitem[Tho04]{Thorup04}
Mikkel Thorup.
\newblock Compact oracles for reachability and approximate distances in planar
  digraphs.
\newblock {\em J. {ACM}}, 51(6):993--1024, 2004.

\bibitem[Twi06]{TwiggPhD}
Andrew~D. Twigg.
\newblock {Compact forbidden-set routing}.
\newblock Technical Report UCAM-CL-TR-678, University of Cambridge, Computer
  Laboratory, December 2006.

\bibitem[vdBS19]{brand2019sensitivity}
Jan van~den Brand and Thatchaphol Saranurak.
\newblock Sensitive distance and reachability oracles for large batch updates.
\newblock In {\em 60th {FOCS}}, pages 424--435, 2019.

\bibitem[WN10]{Wulff-Nilsen10}
Christian Wulff-Nilsen.
\newblock {\em Algorithms for planar graphs and graphs in metric spaces}.
\newblock PhD thesis, University of Copenhagen, 2010.

\end{thebibliography}
\bibliographystyle{alpha}
\end{document}